\documentclass[sn-mathphys,Numbered]{sn-jnl}% Math and Physical Sciences Reference Style

\usepackage{graphicx}%
\usepackage{multirow}%
\usepackage{amsmath,amssymb,amsfonts}%
\usepackage{amsthm}%
\usepackage{mathrsfs}%
\usepackage[title]{appendix}%
\usepackage{xcolor}%
\usepackage{textcomp}%
\usepackage{manyfoot}%
\usepackage{booktabs}%
\usepackage{algorithm}%
\usepackage{algorithmicx}%
\usepackage{algpseudocode}%
\usepackage{listings}%

\usepackage{braket}
\usepackage{hyperref}
\usepackage{float}
\usepackage{color}
\usepackage{sidecap}

\usepackage{dcolumn}
\usepackage{bm}
\usepackage[T1]{fontenc}
\usepackage{etoolbox}

\theoremstyle{plain}

\newtheorem{defi}{Definition}
\newtheorem{defiA}{Definition}

\newtheorem{lemma}[defi]{Lemma}%environment for lemmas in main text, which are restated in appendix
\newtheorem{lemmaT}[defi]{Lemma}%environment for lemmas not restated in appendix
\newtheorem{lemmaA}[defiA]{Lemma}%environment for the restated lemma in appendix
\newtheorem{coro}[defi]{Corollary}
\newtheorem{coroA}[defiA]{Corollary}

\newtheorem{theo}[defi]{Theorem}
\newtheorem{theoA}[defiA]{Theorem}
\newtheorem{hypo}{Hypothesis}
\theoremstyle{definition}
\def\tr{{\rm tr}}   %matrix trace
\def\TR{{\rm TR}}   %map trace
\def\W{\mathbf{W}}  %Rates
\def\J{\mathcal{J}} %Jump superoperator
\def\RR{\mathbf{R}} %Escape rate operator
\def\F{\mathbf{F}}  %Tilter of transition matrix
\def\FF{\mathcal{F}}    %Tilter of transition operator
\def\Q{\mathbf{Q}}  %Transition operator of counted jumps
\def\P{\mathbf{P}}  %Transition operator of all jumps
\def\PP{\mathbb{P}} %Probability law
\def\L{\mathbf{L}}  %Classical Markov generator
\def\LL{\mathcal{L}}    %GKLS generator
\def\A{\mathfrak{A}}    %Set of counted jumps
\def\EE{\mathfrak{E}}   %Set of allowed jumps
\def\E{\mathbb{E}}  %Expectation
\def\S{\mathbf{S}}  %Projection
\def\D{\mathbf{D}}  %Diagonal operator
\def\Id{{\bf I}_{M_k(\mathbb{C})}}    %Identity map
\def\I{{\bf 1}}   %Identity Matrix
\def\Iv{\underline{1}}   %One vector
\begin{document}

\title[Article Title]{Bounds on Fluctuations of First Passage Times for Counting Observables in Classical and Quantum Markov Processes}

\author[1,2]{\fnm{George} \sur{Bakewell-Smith}}
\equalcont{These authors contributed equally to this work.}

\author[1,2,3]{\fnm{Federico} \sur{Girotti}}
\equalcont{These authors contributed equally to this work.}

\author[1,2]{\fnm{M\u{a}d\u{a}lin} \sur{Gu\c{t}\u{a}}}

\author[4,2]{\fnm{Juan} \sur{P. Garrahan}}

\affil[1]{\orgdiv{School of Mathematical Sciences}, \orgname{University of Nottingham}, \orgaddress{\city{Nottingham}, \postcode{NG7 2RD}, \country{UK}}}

\affil[2]{\orgdiv{Centre for the Mathematics and Theoretical Physics of Quantum Non-Equilibrium Systems}, \orgname{University of Nottingham}, \orgaddress{\city{Nottingham}, \postcode{NG7 2RD}, \country{UK}}}

\affil[3]{\orgdiv{Department of Mathematics}, \orgname{Politecnico di Milano}, \orgaddress{\street{Piazza L. da Vinci 32}, \city{Milan}, \postcode{20133}, \country{Italy}}}

\affil[4]{\orgdiv{School of Physics and Astronomy}, \orgname{University of Nottingham}, \orgaddress{\city{Nottingham}, \postcode{NG7 2RD}, \country{UK}}}

\abstract{
We study the statistics of first passage times (FPTs) of trajectory observables in both classical and quantum Markov processes. We consider specifically the FPTs of {\em counting observables}, that is, the times to reach a certain threshold of a trajectory quantity which takes values in the positive integers and is non-decreasing in time. For classical continuous-time Markov chains we rigorously prove: (i) a large deviation principle (LDP) for FPTs, whose corollary is a strong law of large numbers; (ii) 
a concentration inequality for the FPT of the dynamical activity, which provides an upper bound to the probability of its fluctuations to all orders; and (iii) an upper bound to the probability of the tails for the FPT of an arbitrary counting observable. For quantum Markov processes we rigorously prove: 
(iv) the quantum version of the LDP, and subsequent strong law of large numbers, for the FPTs of generic counts of quantum jumps; (v) a concentration bound for the the FPT of total number of quantum jumps, which provides an upper bound to the probability of its fluctuations to all orders, 
together with a similar bound for the sub-class of quantum reset processes which requires less strict irreducibility conditions;  and (vi) a tail bound for the FPT of arbitrary counts. Our results allow to extend to FPTs the so-called ``inverse thermodynamic uncertainty relations'' that upper bound the size of fluctuations in time-integrated quantities. We illustrate our results with simple examples. 
}


\maketitle

\section{Introduction and Summary of Results}\label{sec:intro}
%%%%%NEW

The evolution of most physical systems, whether they be classical or quantum, is characterised by  fluctuations owing to their interaction with the environment \cite{chandler1987introduction,gardiner2004handbook,gardiner2004quantum,seifert2012stochastic,limmer2024statistical}. The study of such dynamical fluctuations in stochastic systems is of central importance for several reasons. On a practical note, fluctuations have a significant impact on the estimation precision of unknown parameters of the dynamics and on the performance of a system for practical uses, such as in heat engines, mechanical and biological clocks, biological motors, or quantum machines. On a conceptual level, quantifying the probabilities of fluctuations away from typical behaviour allows to systematically classify and explain the properties of the dynamics. 

In a stochastic system the main quantities of interest are time-integrated observables of the trajectories of the dynamics, both of time-asymmetric quantities such as particle and energy currents
\cite{seifert2012stochastic,limmer2024statistical}, and of time-symmetric quantities such as dynamical activities \cite{merolle2005space-time,lecomte2007thermodynamic,garrahan2007dynamical,maes2020frenesy:}, or their quantum analogues such as counts of emissions into the environment \cite{plenio1998the-quantum-jump,gardiner2004quantum}. The statistics of these can be quantified using the tools of the theory of large deviations (LD) \cite{touchette2009the-large,dembo2010large}, which in turn allows one to define an ensemble method (akin to the configuration ensemble method of equilibrium statistical mechanics) for stochastic trajectories, with the associated concept of (dynamical) phases and (dynamical) phase transitions \cite{merolle2005space-time,garrahan2007dynamical,hedges2009dynamic,garrahan2018aspects,jack2019ergodicity}. Closely related to time-integrated quantities are their first passage times (FPTs) \cite{redner2001a-guide}, i.e., the times it takes for the observable to reach a certain fixed threshold. Trajectory ensembles can also be studied in terms of FPTs using LD methods \cite{budini2014fluctuating,kiukas2015equivalence}.

A class of fundamental results for trajectory observables are so-called {\em thermodynamic uncertainty relations} (TURs) \cite{barato2015thermodynamic,gingrich2016dissipation}. In their broadest acceptation, TURs are general \textit{lower bounds} on the probability of fluctuations of time-integrated trajectory observables, expressed in terms of global quantities such as entropy production or dynamical activity. The importance of TURs stems from the fact that they point out fundamental physical limitations, for example that larger precision of estimation (implying smaller fluctuations) can only be attained at the cost of higher entropy production and/or higher dynamical activity. Just like for time-integrated quantities, a form of the TUR also provides lower bounds to the probabilities of fluctuations of FPTs \cite{garrahan2017simple,gingrich2017fundamental}. For a sample of the by now large literature on TURs see for example 
\cite{pietzonka2016universal,polettini2016tightening,pietzonka2017finite-time,horowitz2017proof,proesmans2017discrete-time,dechant2018entropic,brandner2018thermodynamic,macieszczak2018unified,barato2018a-unifying,nardini2018process,carollo2019unraveling,guarnieri2019thermodynamics,chun2019effect,koyuk2019operationally,koyuk2020thermodynamic,liu2020thermodynamic,hasegawa2020quantum,hasegawa2022thermodynamic,vo2022unified} and for a review \cite{horowitz2020thermodynamic}.

In contrast to the lower bounds of the TUR, in recent works \cite{bakewell-smith2023general,girotti2023concentration}
we addressed the problem of formulating general {\em upper bounds} to the probability of fluctuations of time-integrated quantities in both classical \cite{bakewell-smith2023general} 
and quantum \cite{girotti2023concentration} stochastic systems using concentration techniques. 
We called these bounds ``inverse TURs'' (iTUR) \cite{bakewell-smith2023general}. Together with the TUR,  the iTUR provides a range that limits the probability of observing a fluctuation away from the average (and not just a one-sided bound). In fact, we showed in \cite{bakewell-smith2023general} that the iTUR can often be a tighter bound than the TUR. This can offer an advantage for estimation, which is countered by the fact that while the iTUR also depends on general features of the system (relaxation time, shortest time-scale) it is not as operationally accessible as the TUR.

In this paper we extend these upper bounds (or iTURs) to the FPTs of counting observables, i.e. quantities corresponding to the accumulated number of changes of the state of the system along time. We derive these upper bounds on FPTs for both classical and quantum  stochastic systems. Anticipating the notation that we will define in detail below, if $\EE$ indicates the set of all possible jumps given the dynamics (all possible transitions between configurations in the classical case, or all possible dissipative events in the quantum case) and $\A \subseteq \EE$ the subset of jumps that we are interested in, let $K_\A(t)$ be the associated stochastic process that counts the number of jumps 
belonging to $\A$ that occur in a stochastic trajectory up to time $t$. The FPT at level $k \in \mathbb{N}$, which we denote $T_\A(k)$, is the time at which $K_\A(t)$ reaches the value $k$ 
for the first time. For example, consider a quantum system whose state evolves according to a quantum Markov semigroup coupled to several ``emission channels'': when the experimenter performs continuous-time counting of measurements in each of these channels they observe a random sequence of clicks from the sensors; the FPT $T_\A(k)$ is when the count of the clicks of interest reaches $k$. (Throughout we assume the dynamics to be irreducible which implies the existence of a unique stationary state.)

This paper presents two sets of results. The first set is for the FPTs in classical systems evolving with continuous-time Markov dynamics. The second set of results generalises the classical ones to continuous-time quantum Markov chains. Our results below can be summarised as follows:

\smallskip 

(i) We show rigorously that the FPTs of classical counting observables, $T_\A^c(k)$, satisfy a large deviation principle or LDP (Theorem \ref{theo:ldpC} below), which, loosely speaking, means that (cf.\ \cite{budini2014fluctuating,garrahan2017simple})
\begin{equation}
    \PP\left (\frac{T^{\rm c}_\A(k)}{k} =dt \right )\asymp e^{-kI^{\rm c}_\A(dt)} ,
    \nonumber
\end{equation}
(where the label ``c'' stands for ``classical'')
for a non-negative generalised function $I^{\rm c}_\A$, called the LD {\em rate function}. $I^{\rm c}_\A$ admits a unique minimum in $\langle t_\A \rangle_{\rm c}$, where it vanishes. This fact, together with some regularity properties of the rate function, implies the strong law of large numbers:
\begin{equation}
    \lim_{k \rightarrow +\infty}\frac{T^{\rm c}_\A(k)}{k} =\langle t_\A \rangle_{\rm c} \quad {\rm a.s.}. 
    \nonumber
\end{equation}
One can see that $\langle t_\A \rangle_{\rm c}$ is the expected time that it takes to observe a jump in $\A$ in the stationary regime. The law of large numbers states that $T^{\rm c}_\A(k)/k$ converges a.s.\ to some asymptotic value and the LDP ensures that fluctuations away from that value decay exponentially fast in the threshold $k$. 
Notice that these are statements \textit{asymptotic} in $k$. 

\smallskip

(ii) Building on the above, we prove an upper bound for the probability of fluctuations of the FPT of the {\em dynamical activity} away from the asymptotic limit and for \textit{any finite threshold $k$} in the form of a concentration bound (Theorem \ref{theo:dynActC} below)
\begin{equation} 
    \max\left \{\PP \left(\frac{T^{\rm c}_{\EE}(k)}{k}\geq \langle t_{\EE} \rangle_{\rm c} + \gamma \right), \PP \left(\frac{T^{\rm c}_{\EE}(k)}{k}\leq \langle t_{\EE} \rangle_{\rm c} - \gamma \right) \right \}
    \leq
    e^{-k\hat{I}^{\rm c}_\EE(\gamma)} ,
    \nonumber
\end{equation}
for some bounding function $\hat{I}^{\rm c}_\EE$ that depends on general properties of the dynamics. 

\smallskip

(iii) We also prove a more general tail bound for right deviations, valid for every $k$ and for any kind of counting observable (not just the activity) in the form of a concentration bound (Theorem \ref{theo:countObs}) 
\begin{equation} 
    \PP\left(\frac{T^{\rm c}_\A(k)}{k} \geq \langle t_\A \rangle_{\rm c}+ \gamma \right)
    \leq 
    e^{-k\tilde{I}^{\rm c}_\A(\gamma)}, 
    \nonumber
\end{equation}
for every $\gamma \geq \overline{\gamma}\geq 0$, where $\overline{\gamma}$ is a constant depending on $\A$ and on the largest average FPT of the system\footnote{
    A quantity with an analogous meaning appears in some recent bounds obtained for the empirical time to reach a particular state in classical Markov processes \cite{bebon2023controlling}.
}, and where the bounding function $\tilde{I}^{\rm c}_\A$ is expressed in terms of general properties of the dynamics. 

\smallskip

(iv) We prove rigorously that the distribution of FPTs of counts of quantum jumps, $T_\A^q(k)$, also obeys a LDP (\ref{theo:ldpQ}), 
\begin{equation}
   \PP\left (\frac{T^{\rm q}_\A(k)}{k} =dt\right )\asymp e^{-kI^{\rm q}_\A(dt)} ,
    \nonumber
\end{equation}
(where the label ``q'' stands for ``quantum'') with rate function  $I^{\rm q}_\A$ which vanishes at its unique minimum $\langle t_\A \rangle_{\rm q}$. The strong law of large numbers follows,
\begin{equation}
    \lim_{k \rightarrow +\infty}\frac{T^{\rm q}_\A(k)}{k} =\langle t_\A \rangle_{\rm q} \quad {\rm a.s.},
    \nonumber
\end{equation}
with $\langle t_\A \rangle_{\rm c}$ the expected time that it takes to observe a quantum jump in $\A$. 

\smallskip

(v) For the case where the observable is the count of all quantum jumps, we prove the existence of a concentration bound for fluctuations of all orders of the corresponding FPT valid for all $k$, i.e., a quantum analogue of the classical result (ii)
\begin{equation} 
    \max \left \{\PP \left(\frac{T_{\EE}^{\rm q}(k)}{k}\geq \langle t_{\EE} \rangle_{\rm q} + \gamma \right) , \PP\left(\frac{T_{\EE}^{\rm q}(k)}{k}\leq \langle t_{\EE} \rangle_{\rm q} - \gamma \right) \right \}
    \leq
    e^{-k\hat{I}^{\rm q}_\EE(\gamma)} ,
    \nonumber
\end{equation}
where the bounding function $\hat{I}^{\rm q}_\EE$ depends on general properties of the dynamics. We prove this for generic quantum Markov chains (Theorem \ref{theo:dynActQ}) which requires a stronger irreducibility assumption on the dynamics than in the classical case (Hypothesis \ref{hypo:irrQuantPhi}). We also prove a more specific concentration bound for the total count in {\em quantum reset processes} (Theorem \ref{theo:dynActReset}) which requires a less stringent irreducibility condition (Hypothesis \ref{hypo:irrQuantL}). 

\smallskip

(vi) We prove a tail bound for right deviations of the FPT $T^{\rm q}_\A(k)$ of a generic quantum jump count for every $k$ 
\begin{equation} 
    \PP\left (\frac{T^{\rm q}_\A(k)}{k} \geq \langle t_\A \rangle_{\rm c}+\gamma\right )\leq e^{-k\tilde{I}^{\rm q}_\A(\gamma)},
\nonumber
\end{equation}
valid for $\gamma \geq \overline{\gamma}\geq 0$, where $\overline{\gamma}$ depends on $\A$ and on the largest average FPT of the system (Theorem \ref{theo:countObsQ}). 

\smallskip

The proof of every bound we derive below goes through Chernoff's inequality (\cite[Section 2.2]{boucheron2013concentration}) and is obtained upper bounding the corresponding moment generating function (MGF) of the FPTs. Easy corollaries are upper bounds of the variance of the FPTs for every value of $k$ in the stationary regime. This complements the lower bounds to probability of fluctuations given by the TURs, previously obtained in the form of upper bounds on the rate functions or on the relative precision \cite{garrahan2017simple,van-vu2022thermodynamics,hasegawa2022thermodynamic,pietzonka2023thermodynamic}. 

\smallskip 

The rest of the manuscript is organised as follows. The main text is split into two parts: Sec.~\ref{CMarkov} focuses on classical Markov processes, while Sec.~\ref{sec:qmarkov} focuses on quantum Markov processes. 
For the classical case, we introduce notation, definitions,  existing and preliminary results in Subsec.~\ref{sec:pn}. The LDP for classical FPTs is presented in Subsec.~\ref{sec.LD.classical}, the  classical concentration inequality for the FPT of the activity in Subsec.~\ref{sec:dynActC}, and the FPT tail bound in Subsec.~\ref{sec:countObs}. For quantum Markov processes we introduce notation, definitions and previous results in Subsec.~\ref{sec:qPrelims}. The LDP for FPT of quantum jump counts is presented in Subsec.~\ref{subsec:qLDP}, the concentration bound for the FPT of total number of counts in Subsec.~\ref{sec:genQuant}, the FPT of the total count in quantum reset processes  Subsec.~\ref{sec:reset}, and the tail bound for the FPT of more general quantum counting processes in Subsec.~\ref{sec:countObsQ}. In Sec.~\ref{sec:Conc} we provide our conclusions. The Appendices contain the proofs of the new theorems and lemmas presented in the main text.

%%%%%%%%%%%%%%%%%%%%%%%%%%%%%%%%%%%
\section{Large Deviation Principle and Concentration Bounds for FPTs in Classical Markov Processes}\label{CMarkov}
%%%%%%%%%%%%%%%%%%%%%%%%%%%%%%%%%%%%%%%

In this Section we begin by recalling the necessary notions regarding classical Markov chains, introducing notation and stating the assumptions that we make in this paper. In particular, we define first passage times (FPTs) corresponding to counting observables, that is, time-additive observables of trajectories which are non-decreasing (in contrast to currents), cf. \cite{garrahan2017simple}. After that, we prove that the sequence of FPTs satisfies a Large Deviation Principle \cite{touchette2009the-large} and we provide an expression for the {\em rate function} (the scaled logarithm of the probability). Then, we provide an upper bound on the fluctuations of the FPT for the dynamical activity \cite{garrahan2007dynamical,maes2020frenesy:}, and a tail bound for FPTs for generic counting observables. We illustrate our results with simple models, and in particular we discuss the behaviour of the bound for the FPT corresponding to the dynamical activity when the system is at conditions of metastability, i.e., near a first-order phase crossover.

\subsection{Preliminaries and Notation} \label{sec:pn}

Let us consider a continuous-time Markov chain $X:=(X_t)_{t\geq0}$ taking values in a finite configuration space $E$. For an initial distribution $\nu$ over configurations we denote by $\mathbb{P}_\nu$ the corresponding law of the process $X$ and $\mathbb{E}_\nu$ its integral. The stochastic generator of $X$ has the form:
\[
\L=\W-\RR,
\]
where the off-diagonal part, $\W=\sum_{x\neq y}w_{xy}\ket{x}\bra{y}$, encodes the 
rates of jumps (with $w_{xy}$ the transition rate from configuration $x$ to configuration $y$), and the diagonal part, $\RR=\sum_x R_x\ket{x}\bra{x}$, the escape rates (with $R_x = \sum_y w_{xy}$ the escape rate from configuration $x$). The generator $\L$ acts on complex valued functions $f:E \rightarrow \mathbb{C}$ via right multiplication, i.e., 
\[
f(x) \mapsto (\L  f)(x)=\sum_{y \in E}\L_{xy}f(y)=\mathbb{E}_{\delta_x}[f(X_1)] , 
\]
which is similar to the ``Heisenberg picture'' in quantum mechanics. The natural norm to consider in this setting is the $\infty$-norm, i.e.
\[
\|f\|_{\infty}:=\max_{x \in E}|f(x)|, \quad \|\L\|_{\infty\rightarrow \infty}:=\max_{\|f\|_\infty=1}\|\L f\|_\infty.
\]
By duality, $\L$ acts also on complex valued measures on $E$ (which can be identified with their density $\nu:E \rightarrow \mathbb{C}$) via left multiplication (cf.\ Schr\"odinger picture in quantum mechanics):
\[
\nu(x) \mapsto (\L_*\nu)(x):=(\nu \L)(x)=\sum_{y \in E}\nu(y)\L_{yx}=\mathbb{P}_\nu(X_1=x).
\]
Here the natural norm is the dual norm with respect to the $\infty$-norm, which is the $1$-norm:
\[
\|\nu\|_{1}:=\sum_{x \in E}|\nu(x)|, \quad \|\L_*\|_{1 \rightarrow 1}:=\max_{\|\nu\|_1=1}\|\L_*\nu  \|_1=\|\L\|_{\infty\rightarrow \infty}.
\]
% The off diagonal parts of $\L$ represent the transition rates $w_{xy}$, and the diagonals the escape rates $R_x=\sum_{y:y\neq x}w_{xy}$, and we have $x,y \in E$. 
We will often use the notation $\langle \nu, f \rangle$ to denote the integral of $f$ with respect to $\nu$, i.e. $\sum_{x \in E}\nu(x)f(x).$ We state below our main assumption for classical dynamics.

\bigskip \begin{hypo}[Irreducibility of $\L$]\label{hypo:irrC}
    There exists a unique fully supported measure $\hat{\pi}$ satisfying $\hat{\pi}\L=0$.
\end{hypo}

\bigskip The process $X$ can be equivalently described in terms of the corresponding jump process and holding times: indeed any trajectory takes the form of a sequence
\[
\omega=\{(x_0,t_0),(x_1,t_1),(x_2,t_2),\cdots,(x_k,t_k),\dots \},
\]
where $x_i$ is the state of the system before the $i$-th jump and $t_i$ is the time between the $(i-1)$-th and the $i$-th jump (we set $t_0=0$). The process describing the different states of the system along time (\textit{jump process}) is a discrete time Markov process with transition matrix given by 
\begin{equation}
\label{eq:P}
    \P=\RR^{-1}\W
\end{equation}
(notice that due to irreducibility, $R_x>0$ for every $x \in E$). If $\L$ is irreducible, then $\P$ is irreducible too. Indeed 
\[\begin{split}
    \hat{\pi} \L= 0 \Leftrightarrow \hat{\pi} \W= \hat{\pi}\RR \Leftrightarrow \hat{\pi} \RR \P= \hat{\pi} \RR
\end{split}
\]
and $\hat{\pi} \mapsto \hat{\pi}\RR$ is a positive linear bijection. Therefore the unique invariant measure of $\P$, denoted $\pi$ is related to the invariant measure of the continuous-time generator by
\begin{equation}
\label{eq:pi}
\pi=\frac{\hat{\pi}\RR}{\langle \hat{\pi}, \RR \Iv \rangle}
\end{equation}
where $\Iv$ stands for the function identically equal to $1$. Irreducibility of the dynamics means that $\pi$ has full support. Conditional to the jump process, the holding times $t_{i}$ are independent and $t_i$ (for $i\geq 1$) is exponentially distributed with parameter $R_{x_{i-1}}$.

\bigskip In practical applications one might be able to observe only certain jumps of the trajectory: we denote by $K_{xy}(t)$ the process that counts the number of transitions $x \to y$ up to time $t$, and more generally, given a nonempty subset $\A$ of the set of possible jumps $\EE:=\{(x,y):x,y\in E:w_{xy}>0\}$, we denote 
\[
K_\A(t)=\sum_{(x,y) \in \A}K_{xy}(t)
\]
the stochastic process that counts the number of jumps in $\A$ up to time $t$. The {\em dynamical activity}, or total number of jumps, is the observable corresponding to $\A=\EE$. The {\em first passage time} (FPT), $T_\A(k)$ for a trajectory observable $K_\A$ corresponding to the value $k \in \mathbb{N}$ is defined as:
\begin{equation}\label{fptdef}
T_\A(k)=\inf_{t\geq 0}\{t:K_\A(t)=k\}.
\end{equation}
In particular, the first passage time for the total activity corresponding to the level $k$ is given by the sum of the first $k$ holding times:
\begin{equation}\label{fptcounts}
T_\EE(k)=\sum_{i=1}^{k}t_i.
\end{equation}
Using the properties of the holding times described above, one finds that the moment generating function (MGF) of $T_\EE(k)$ is well defined for $u<R_{\min}:=\min_x R_x$ and is given by
\begin{equation} \label{eq:MGFcda}
\mathbb{E}_\nu[e^{uT_\EE(k)}]=\left \langle \nu ,\left ( \frac{\mathbb{\RR}}{\RR-u}\P\right )^k \Iv \right \rangle=\left \langle \nu ,\left ( \frac{\I}{\RR-u}\W\right )^k \Iv \right \rangle.
\end{equation}
An analogous formula can be found for every counting observable of the type in Eq. \eqref{fptdef}. First of all, it is useful to consider the following splitting of the evolution generator $\L$:
\begin{equation}\label{eq:splitL}
\L=\W_1+\underbrace{\W_2-\RR}_{\L_\infty},
\end{equation}
where $\W_1$ holds the rates of transitions in $\A$ and $\W_2$ the rates of transitions not in $\A$. Notice that we can always write the first passage time corresponding to the level $k$ as a sum of times between subsequent jumps in $\A$:
\[T_\A(k)=\sum_{i=1}^{k}s_i, \quad s_i:=T_\A(i)-T_\A(i-1).
\]
The process $Y = (y_0,\dots, y_k,\dots)$ determined by the state of the system at the sequence of times $\{T_\A(k)\}_{k=0}^{+\infty}$ is a discrete time Markov process with transition matrix given by
\begin{equation} \label{eq:qdef}
\Q:=-\frac{\I}{\L_\infty}\W_1.
\end{equation}
Indeed, using that $\L_\infty=\W_2 - \RR$, we can write
\begin{equation}\label{eq:Linfty.inverse}
-\frac{\I}{\L_\infty}=\frac{\I}{\RR-\W_2}=\frac{\I}{\I-\RR^{-1}\W_2}\frac{\I}{\RR}=\sum_{k \geq 0} \left ( \frac{\I}{\RR}\W_2 \right )^k\frac{\I}{\RR},
\end{equation}
and therefore
\begin{equation} \label{eq:qexp}
    \Q=\sum_{k \geq 0} \left ( \frac{\I}{\RR}\W_2 \right )^k\frac{\I}{\RR}\W_1.
\end{equation}

Since $\RR^{-1}\W_1$ and $\RR^{-1}\W_2$ are the sub-Markov operators that encode the probabilities of jumps which do and do not, respectively, belong to $\A$, Eq. \eqref{eq:qexp} expresses the fact that the probability of the jump $x \to y$ %in $\A$ 
for the process $Y$ is obtained by  summing up the probabilities of all possible trajectories of the jump process associated to $X$ that start in $x$, arrive in a state $z$ such that $(z,y) \in \A$ by using only jumps in $\A^C$, and then jump from $z$ to $y$. Integrating over all such possible paths, one can also show that for every $u < \overline{\lambda}:=-\max\{\Re(z): z \in {\rm Sp}(\L_\infty)\}$, the MGF of $T_\A(k)$ can be written as:
\begin{equation}\label{mgfDefC}
\begin{split}
\E_\nu[e^{uT_\A(k)}]&=\left \langle \nu, \left(\frac{\L_\infty}{u+\L_\infty}\Q \right)^k\Iv\right \rangle.
\end{split}
\end{equation}

We remark that for suitable choices of initial distributions and for finite $k$'s, $\E_\nu[e^{uT_\A(k)}]$ might be well defined even for some values of $u$ bigger or equal than $\overline{\lambda}$; nevertheless Theorem \ref{theo:ldpC} shows that in the large $k$ limit, the only values which play a nontrivial role are $u <\overline{\lambda}$.

For the case of dynamical activity, we have already mentioned that $\Q=\P$ is irreducible; more generally, $\Q$ is only irreducible on the subspace $\{y\in E:\exists x\in E:\,(x,y)\in\A\}$, the complement of which is transient. Indeed, $\Q$ admits as unique invariant measure
\begin{equation}\label{eq:invmeas}\varphi=\frac{\hat{\pi} \W_1}{
%\sum_{x,y\in E} \hat{\pi}(x)(\W_{1})_{ xy}
\langle \hat{\pi}, \W_1 \Iv \rangle
},\end{equation}
which, in general, is not fully supported.

From the expression of the moment generating function, using standard theory (see the Theorem \ref{theo:ldpC} below) one obtains that under $\mathbb{P}_\nu$ (for every initial law $\nu$) the following convergence holds true almost surely:
\begin{equation}
\label{eq:mean.t}
\frac{T_\A(k)}{k} \xrightarrow{a.s.} \langle t_\A \rangle:=\left \langle \varphi,-\frac{\I}{\L_\infty}\Iv\right \rangle.
\end{equation}
%FPTs can be used to estimate kinetical 
Lemma \ref{lem:tech1} below ensures that the expressions appearing in \eqref{eq:qdef}, \eqref{eq:qexp} and \eqref{mgfDefC} are well defined and that the identities are true. Before stating the lemma, we need to recall a few notions that will also be useful in the rest of the paper. Given a matrix $\mathbf{A} \in M_n(\mathbb{C})$, the spectral radius of $\mathbf{A}$ is defined as
\[
r(\mathbf{A}):=\max\{|z| :\, z \in {\rm Sp}(\mathbf{A})\}.
\]
The spectral radius is fundamental in studying the convergence of the geometric series $\sum_{k\geq 0} \mathbf{A}^k$, since Gelfand formula states that
\[
\lim_{k\rightarrow +\infty} \|\mathbf{A}^k\|^{\frac{1}{k}}=r(\mathbf{A}).
\]
Therefore, if $r(\mathbf{A})<1$, the series converges.
\begin{lemma} \label{lem:tech1}
    The following statements hold true:
    \begin{enumerate}
        \item $\overline{\lambda}:=-\max\{\Re(z):z \in {\rm Sp}(\L_\infty)\}>0$, hence $\L_\infty$ is invertible;
        \item $r(\RR^{-1}\W_2)<1,$ therefore $\sum_{k \geq 0} \S^k$ is well defined with $\S = \RR^{-1}\W_2$ and one has
        \[
        -\frac{\I}{\L_\infty}=\sum_{k\geq 0}\left ( \frac{\I}{\RR}\W_2 \right )^k\frac{\I}{\RR},
        \]
        \item for every $u <\overline{\lambda}$, one has
        \[
        \E_\nu[e^{uT_{\A}(k)}]=\left \langle \nu, \left(\frac{\L_\infty}{u+\L_\infty}\Q \right)^k\Iv\right \rangle,
        \]
        \item $\left \|\L_\infty^{-1} \right \|^{-1}_{\infty \rightarrow \infty}\leq \overline{\lambda}$.
    \end{enumerate}
\end{lemma}

The proof of Lemma \ref{lem:tech1} can be found in Appendix \ref{app:tech1}. Loosely speaking, items 1 and 2 hold true because  $\L_\infty$ and $\RR^{-1}\W_2$ are the counterparts of $\RR$ and $\P$, respectively, 
obtained by considering a restricted set of jumps in the original irreducible Markov process.

\bigskip 
Finally, we introduce here some Hilbert space notions which will be needed in formulating our results and will be used in their proofs.
The space of complex functions on $E$ can be turned into a Hilbert space 
$L^2_\pi(E)$ using the inner product $\langle \cdot,\cdot\rangle_\pi$ with respect to the invariant measure $\pi$ defined in equation \eqref{eq:pi}
\[
\langle f,g \rangle_\pi:=\sum_{x\in E}\pi(x)\bar{f}(x)g(x).
\]

We use the notation $\|f\|_\pi$ for the corresponding norm. The adjoint $\mathbf{A}^\dagger$ of an operator $\mathbf{A}$ on $L^2_\pi(E)$ has matrix elements
$$
\mathbf{A}^\dagger_{xy}:=\frac{\pi(y)}{\pi(x)}\mathbf{A}_{yx}.
$$
From this it follows that $\P^\dagger$ is a transition operator in its own right. 
An important quantity in this work is the \emph{absolute spectral gap} of $\P$, which we denote by $\varepsilon$ and is defined as the spectral gap of $\P^\dagger \P$ (the multiplicative symmetrisation of $\P$):
\begin{equation}
\label{eq:spectral.gap.PP}\varepsilon:=1- \max\{\| \P f\|_\pi : \, \langle \pi, f \rangle=0, \, \|f\|_\pi=1\}.
\end{equation}

Using this we define the \emph{Le\'{o}n-Perron operator} $\hat{\P}$ associated to $\P$ as(\cite{LP04})
\begin{equation}
\label{eq.LeonPerron}
\hat{\P}:=(1-\varepsilon)\I+\varepsilon\Pi,
\end{equation}
where $\Pi$ is the map $\Pi: f\mapsto \langle f,\Iv\rangle_\pi\Iv$. $\hat{\P}$ is a self-adjoint transition operator which is simple to handle and will allow us to derive upper bounds for the fluctuations of FPTs of $\P$.

\subsection{Results on Classical Markov processes}
In this section we describe our results for classical Markov processes. We then illustrate these results by considering three simple specific examples.

%%%%%%%%%%%%%%%%%%%%%%%%%%%%%%%%%%%%%%%
\subsubsection{Large Deviation Principle for General Counting Observables}\label{sec.LD.classical}
%%%%%%%%%%%%%%%%%%%%%%%%%%%%%%%%%%%%%%%%

We recall that, given a function $I_\A:\mathbb{R} \rightarrow [0,+\infty]$, the stochastic process $\{T_\A(k)/k\}$ is said to satisfy a Large Deviation Principle with rate function $I_\A$ if for every Borel measurable set $B \subseteq \mathbb{R}$ one has that \cite{touchette2009the-large,dembo2010large}
\[
\begin{split}
&-\inf_{t \in \overset{\circ}{B}}I_\A(t) \leq \liminf_{k\rightarrow +\infty} \frac{1}{k}\log\left (\PP_\nu\left (\frac{T_\A(k)}{k} \in B \right ) \right ), \\
&\limsup_{k\rightarrow +\infty} \frac{1}{k}\log\left (\PP_\nu\left (\frac{T_\A(k)}{k} \in B \right ) \right ) \leq -\inf_{t \in \overline{B}}I_\A(t),
\end{split}\]
where $\overset{\circ}{B}$ and $\overline{B}$ denote the interior and the closure of $B$, respectively. The rate function $I_\A(t)$ is called good if it has compact level sets. 
%The following result holds.
%Loosely speaking, it means that
% $$
% \PP_\nu\left (\frac{T_\A(k)}{k} =t \right ) \approx e^{-kI(t)}.$$

\bigskip\begin{theo} \label{theo:ldpC}
Let us consider any nonempty subset $\A$ of the set of possible jumps. The collection of corresponding FPTs $\{T_{\A}(k)/k\}$ satisfies a LDP with good rate function given by
$$
I_\A(t):=\sup_{u \in \mathbb{R}}\{ut-\log(r(u))\}$$
where
$$
r(u)=\begin{cases} r \left (\Q_u \right ) & \text{ if } u < \overline{\lambda}\\
+\infty & \text{otherwise}\end{cases}$$
where $\Q_u:=-(u+\L_\infty)^{-1}\W_1$ and $\overline{\lambda}:=-\max\{\Re(z):z \in {\rm Sp}(\L_\infty)\}.$
\end{theo}

The proof of Theorem \ref{theo:ldpC} can be found in Appendix \ref{app:LDPc}; the proof highlights some properties of $r(u)$, which imply (as one would expect) that $I_\A(t)=+\infty$ for $t \leq 0$ and that
\[\lim_{t \rightarrow 0^{+}}I_\A(t)=+\infty,\quad \lim_{t \rightarrow +\infty}I_\A(t)=+\infty,\quad \lim_{t \rightarrow +\infty}I_\A^\prime(t)=\overline{\lambda}.
\]
Moreover, $I_\A(t)$ has a unique minimum in $\langle t_{\A} \rangle$, where it is equal to $0$. The strong law of large numbers is a consequence of the smoothness of $r(u)$ around $0$; see for instance \cite[Theorem II.6.3 and Theorem
II.6.4]{ellis2006entropy}. 
We refer to \cite{budini2014fluctuating,garrahan2017simple} for a more in depth discussion of the physical meaning of this result.

%%%%%%%%%%%%%%%%%%%%%%%%%%%%%%%%%%
\subsubsection{Concentration Bound for Dynamical Activity}\label{sec:dynActC}
%%%%%%%%%%%%%%%%%%%%%%%%%%%%%%%%%%%

Recall that we consider a classical continuous time Markov process with generator $\L$ whose jumps can be described by a discrete time process with transition matrix $\P$, cf. Eq. \eqref{eq:P}. The dynamical activity $K_\EE(t)$ is the total number of configuration changes (referred to also as jumps) occurring in a trajectory up to time $t$ \cite{lecomte2007thermodynamic,garrahan2018aspects,maes2020frenesy:}. The corresponding first passage time $T_\EE(k)$, is the time of the $k$-th jump, cf. Eq. \eqref{fptdef}. The first result of this paper is an upper bound on the probability that the average jump time $T_\EE(k)/k$ deviates from its asymptotic or stationary mean:
$$\langle t_\EE \rangle=\sum_{x\in E}\pi(x)\frac{1}{R_x}.
$$
We now introduce two quantities which appear in the bounds of Theorem \ref{theo:dynActC} below:
\begin{itemize}
\item[1.] the second moment at stationarity:
\begin{equation}
\label{eq:b.c}
2 b_c^2:=\sum_{x \in E}\pi(x) \frac{2}{R^2_x};
\end{equation}
\item[2.] the longest expected waiting time:
\begin{equation}
    c_c:=\max_{x \in E} \left \{\frac{1}{R_x}\right\} = \frac{1}{R_{\rm min}}.
    \label{eq:cc}    
\end{equation}
\end{itemize}

\begin{theo}[Fluctuations of FPT for Activity]\label{theo:dynActC}
Suppose Hypothesis \ref{hypo:irrC} holds ($\L$ is irreducible) and let $\varepsilon$ be the spectral gap of $\P^\dag \P$, cf. Eq. \eqref{eq:spectral.gap.PP}. For every $\gamma > 0$ the following holds true:
\begin{equation*}
\begin{split}
&\PP_\nu \left (\frac{T_{\EE}(k)}{k}\geq \langle t_{\EE} \rangle + \gamma \right) \leq C(\nu) \exp \left ( -k \frac{\gamma^2 \varepsilon}{4b_c^2}h\left ( \frac{5c_c\gamma}{2 b_c^2}\right )\right )\\
{\rm and}\\
&\PP_\nu \left (\frac{T_{\EE}(k)}{k} \leq \langle t_{\EE} \rangle - \gamma \right) \leq C(\nu) \exp \left ( -k \frac{\gamma^2 \varepsilon}{4b_c^2}h\left ( \frac{5c_c\gamma}{2 b_c^2}\right )\right ), \quad k\in \mathbb{N},
\end{split}
\end{equation*}
where $h(x):=(\sqrt{1+x}+\frac{x}{2}+1)^{-1}$ and $C(\nu):=\max_{x\in E} \left\{\nu(x)/\pi(x)\right\}$.
\end{theo}
The proof of Theorem \ref{theo:dynActC} can be found in Appendix \ref{app:dynActC} and follows the same line as in \cite[Theorem 3.3]{lezaud1998chernoff-type}. From the proof, one can see that if $\P$ is self-adjoint, one can derive an upper bound with a slightly different expression which contains the spectral gap of $\P$ instead of its absolute spectral gap.

Let us make few considerations regarding the quantities appearing in the bound. $C(\nu)$ accounts for the difference between the initial measure and the stationary one, in particular $C(\pi)=1$. The absolute spectral gap $\varepsilon$ controls the speed at which an arbitrary density $\nu$ converges to the invariant measure $\pi$ under iterations of the transition operator $\P_*$: indeed, for every $k \geq 1$
\begin{equation*}
\|\P_*^{k}(\nu-\pi)\|_1  \leq \left \|\P^{\dagger k}\left (\frac{\nu^{1/2}}{\pi^{1/2}}-\Iv\right )\right \|_\pi \leq 2\varepsilon^{\frac{k}{2}}\left ( 1-\sum_{x \in E} \nu(x)^{1/2}\pi(x)^{1/2}\right ).
\end{equation*}
This enables one to upper bound the deviation probability of $T_\EE(k)$ using stationary properties of the system. We remark that the use of the spectral gap of $\P^\dagger \P$ instead of the one of $\P$ allows to bound the fluctuations of the first passage time for every $k \geq 1$ and not only asymptotically in $k$. Small values of $\varepsilon$ can correspond in some models to big fluctuations of the first passage time (cf.\ Ref.~\cite{bakewell-smith2023general} and example \ref{ex:phaseTran} below).

The quantity $b_c^2$ encodes the variance of $T_\EE(k)$ in the stationary regime. Indeed, the distribution of the interarrival times $t_i$ at stationarity is the same as the random variable obtained drawing a state $x$ from the invariant distribution $\pi$ and then sampling from an independent exponential random variable with parameter $R_x$. Such a random variable has a variance equal to
\[
2\sum_{x \in E} \pi(x) \frac{1}{R_x^2}-\left(\sum_{x \in E} \pi(x) \frac{1}{R_x}\right)^2.
\]
Notice that the following inequalities hold true:
\[
b_c^2 \leq 2\sum_{x \in E} \pi(x) \frac{1}{R_x^2}-\left(\sum_{x \in E} \pi(x) \frac{1}{R_x} \right)^2\leq 2b_c^2,
\]
hence the variance of the interrarival times at stationarity and $b_c^2$ (see Eq. \eqref{eq:b.c}) differ at most by a factor $2$. The bigger $b_c^2$, the bigger the fluctuations of the first passage time. Finally, as one might reasonably expect, the dependence of the bound on $c_c$ is such that the bigger $c_c$, the heavier the right tail. Notice that the ratio between $b_c^2$ and $c_c$ that appears in the bound can be controlled by the average at stationarity:
\[
\frac{R_{\rm min}}{R_{\max}}\langle t_\EE \rangle \leq \frac{b_c^2}{c_c}=\sum_{x \in E}\pi(x)\frac{R_{\rm min}}{R_x^2} \leq \langle t_\EE\rangle.
\]

On the other hand, $\varepsilon$ and $b_c^2$ are quite independent from each other. For example, if one modifies uniformly the speed of the Markov process $X$, i.e. $\L \rightarrow \lambda\L$ for some positive $\lambda$, one has that the jump process does not change and therefore $\varepsilon$ remains the same, while $b_c^2 \rightarrow \lambda^{-2}b_c^2$. Notice that the bound has the right scaling with respect to this group of transformations: indeed, the bound becomes
\[
C(\nu) \exp \left ( -k \frac{(\lambda\gamma)^2 \varepsilon}{4b_c^2}h\left ( \frac{5c_c\lambda \gamma}{2 b_c^2}\right )\right ),
\]
which corresponds to the upper bound for deviations of the order $\lambda \gamma$ for the original dynamics.

As a consequence of the proof of Theorem \ref{theo:dynActC}, we obtain an upper bound on the variance at stationarity of the FPT corresponding to the dynamical activity. This result complements the lower bound (or thermodynamic uncertainty relation) for the FPT of the activity obtained in \cite{garrahan2017simple}:

\bigskip \begin{coro}\label{dynActiTURC}
The variance of the first passage time for the total activity at stationarity is bounded from above by:
\[
\frac{{\rm var}_\pi(T_{\EE}(k))}{k}\leq \left(1 + \frac{2}{\varepsilon}\right)b_c^2.
\]
\end{coro}
The proof of Corollary \ref{dynActiTURC} can be found in Appendix \ref{app:dynActC}.

\subsubsection{Tail Bound for General Counting Observables}\label{sec:countObs}
Our second main result is a concentration bound on the tails of the distribution of the FPT for general counting observables, $T_{\A}(k)/k$. Similarly to the above result, this bounds the probability that $ 
T_{\A}(k)/k$ deviates from $\langle t_{\A} \rangle$. Recall that $\L_\infty$ is a sub-Markov generator describing the jumps in $\A^C$, cf. Eq. \eqref{eq:splitL}. We introduce the following notation
\[
\beta: = \left \|\frac{\I}{\L_{\infty}}\right \|_{\infty\rightarrow \infty} .
\]
In the case of the dynamical activity it is simply given by $\beta=\max_x 1/R_x$. In general, 
$\beta$ satisfies 
$\beta\geq \langle t_{\A} \rangle$ by equation \eqref{eq:mean.t}, and as we show below, it can be interpreted as the longest timescale of the system. Indeed, since $-\L^{-1}_\infty$ is a positivity preserving map, one has that
$$\|\L_\infty^{-1}\|_{\infty\rightarrow \infty}=\|\L_\infty^{-1}\Iv\|_\infty=\max_{x \in E}\sum_{y \in E}|L_{\infty,xy}^{-1}|,$$
and since
$$\|\L_\infty^{-1}\Iv\|_\infty=\max_{\nu}- \langle \nu, \L_\infty^{-1}\Iv \rangle=\max_{\nu}\mathbb{E}_\nu[T_{\A}(1)], $$
we obtain
$$
\beta= \max_{\nu}\mathbb{E}_\nu[T_{\A}(1)], 
$$
where $\nu$ is a probability density on the state space. We can now state our second main result:

\bigskip

\begin{theo}[Rare Fluctuations of General Counting Observable FPTs]\label{theo:countObs}
Let $\L$ be irreducible and $\A\subseteq \EE$ be nonempty. For every $k\in \mathbb{N}$ and $\gamma > \beta - \langle t_\A \rangle$
\[
\PP_\nu \left (\frac{T_{\A}(k)}{k} \geq \langle t_{\A} \rangle + \gamma \right) \leq \exp \left ( -k \left(\frac{\gamma+\langle t_{\A} \rangle-\beta}{\beta}-\log\left(\frac{\gamma+\langle t_{\A} \rangle}{\beta}\right)\right)\right ). 
\]
\end{theo}

The proof of the Theorem \ref{theo:countObs} can be found in Appendix \ref{app:countObs}.

\bigskip Here we  comment briefly on the rather simple idea behind it. Let $Z$ be the sum of $k$ independent exponential random variables with parameter $\beta^{-1}$, then by applying the Chernoff bound one obtains that for every $0 \leq u < \beta^{-1}$
\[
\mathbb{P}(Z /k\geq \beta + \gamma^\prime) \leq \exp \left (-k \left (u (\beta+ \gamma^\prime) +\log(\beta) + \log(\beta^{-1}-u) \right ) \right ).
\]
Optimising in $u$ in the previous equation, one gets
\begin{equation} \label{eq:expbound}
\mathbb{P}(Z/k \geq \beta + \gamma^\prime) \leq \exp \left (-k \left (\gamma^\prime \beta^{-1}-\log(1+\gamma^\prime\beta^{-1})\right ) \right ).
\end{equation}
As the interarrival times are distributed according to the matrix-exponential distribution (\cite{BFT08}) with rate matrix $-\L_\infty$, and $\beta=\|\L^{-1}_{\infty*}\|_{1\rightarrow 1}$, the moment generating function of $T_\A (k)$ is bounded from above by that of $Z$. Equation \eqref{eq:expbound} then provides the bound in Theorem \ref{theo:countObs}.

Unlike the case of Theorem \ref{theo:dynActC}, the bound in Theorem \ref{theo:countObs} does not cover small fluctuations and this makes it impossible to use to derive any bound on the variance of $T_{\A}(k)$ in the spirit of Corollary \ref{dynActiTURC}. Nevertheless, using the explicit expression of the variance (see Lemma \ref{lem:asympVar}), one can derive the following bound.

\begin{coro} \label{coro:cgiTUR}
Given any non-empty set of jumps $\A$, the variance of the corresponding first passage time at stationarity is bounded from above by:
\[
\frac{{\rm var}_\varphi(T_{\A}(k))}{k}\leq \left ( 1+ \frac{2}{\tilde{\varepsilon}}\right )\beta^2,
\] 
where
$$\tilde{\varepsilon}:=1-\max\{\|\Q f\|_\infty: \|f\|_\infty=1, \,\langle\varphi,f \rangle=0\}.$$ 
\end{coro}

The proof of Corollary \ref{coro:cgiTUR}  can be found in Appendix \ref{app:countObs}. We recall that $\varphi$ is the unique invariant law for ${\bf Q}$ and was defined in Equation \eqref{eq:invmeas}. We remark that Corollary \ref{coro:cgiTUR} together with Chebyschev inequality provides bounds on small deviations as well.

\bigskip The constant $\beta$ may be difficult to compute, especially for large systems. However, it is not hard to check that Theorem \ref{theo:countObs} keeps holding true if we replace $\beta$ with any $\tilde{\beta}\geq \beta$. The corollary that follows shows that it is possible to upper bound $\beta$ (and obtain alternative concentration bounds for the FPT) in terms of the following simpler quantities of the system:
\begin{itemize}
    \item[1.] maximum escape rate:
    \begin{equation}
    \label{eq:Rmax}
    R_{\rm max} := \max_{x\in E}\{R_x\};
    \end{equation}
    \item[2.] minimum transition rate:
    \begin{equation}
    \label{eq:wmin}
    w_{\rm min} := \min_{x,y\in E}\{w_{xy} : w_{xy} > 0\};
    \end{equation}
    \item[3.] minimax jump distance $\tilde{k}$: the minimum $k\in \mathbb{N}$ such that for any initial state $i\in E$ there exists 
    a trajectory $(i_0=i ,i_1, \dots, i_l)$ with $l\leq k$ such that $w_{i_j, i_{j+1}}>0$ for all $j=0,\dots l-1$ 
    and the trajectory ends with a jump in $\A$, i.e. $(i_{l-1}, i_{l})\in \A$.
    %\[
    %\tilde{k} := \inf_{k\in \mathbb{N}}\left \{k : \left \|\W_2^*\frac{\I}{\RR}\right\|_{1\rightarrow 1}^k<1\right \}.
    %\]
   
\end{itemize}

\begin{figure}[t]
    \begin{center}
    \includegraphics[width=\linewidth]{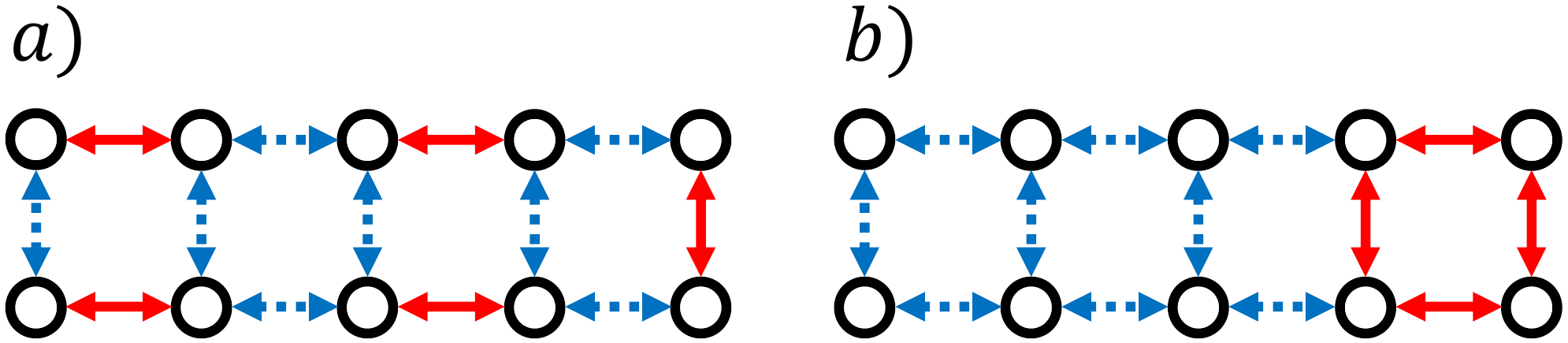}
    \end{center}
    \caption{{\bf Minimax Jump Distance.}
    Configurations of a discrete system are represented by circles, and allowed transitions between them by arrows. Jumps in $\A$ (full/red) contribute to the observable, whilst jumps in $\A^C$ (dotted/blue) do not. (a) System where red jumps are distributed throughout the graph, in this case $\tilde{k}=1$. (b) Uneven distribution of jumps in $\A$, in this case $\tilde{k}=4$.
    }
    \label{minimaxDiagram}
\end{figure}

While $R_{\rm max}$ and $w_{\rm min}$ can be computed easily in terms of the transition rates, the minimax jump distance $\tilde{k}$ can be read off the graph of the process, see Figure \ref{minimaxDiagram}. 
The fact that it is finite follows from the irreducibility of the Markov process.
The corollary is stated below:

\bigskip

\begin{coro}[Simple Upper Bound on $\beta$]\label{countObsBeta}
For general counting observables, the norm $\beta:=\left\|\L_{\infty}^{-1}\right\|_{\infty\rightarrow \infty}$ is bounded from above by:
\[
\beta \leq c_c\tilde{k}\max_{(x,y) \notin \A} \left \{ \frac{R_{x}}{w_{xy}}\right \}^{\tilde{k}-1}\max_{(x,y) \in \A} \left \{ \frac{R_{x}}{w_{xy}}\right \}\leq c_c\tilde{k}\left (\frac{R_{\rm max}}{w_{\rm min}}\right)^{\tilde{k}}=:\tilde{\beta},
\]
with $c_c, R_{\rm max}, w_{\rm min}$ defined in Eqs. \eqref{eq:cc}, \eqref{eq:Rmax}, and respectively \eqref{eq:wmin}.
The concentration bound in Theorem \ref{theo:countObs} holds with $\beta$ replaced by any of the two upper bounds above.
\end{coro}

The proof of Corollary \ref{countObsBeta}  can be found in Appendix \ref{app:countObs}. In the case of total activity, i.e. when $\A=\EE$, one can easily see that $\beta=c_c$.

%%%%%%%%%%%%%%%%%%%%%%%%%%%%%%%%
\subsection{Examples: Classical Concentration Bounds for Markov Processes}
%%%%%%%%%%%%%%%%%%%%%%%%%%%%%%%%%%%%

In this section we illustrate the main results of the classical part of the paper with three simple examples.

%%%%%%%%%%%%%%%%%%%%%%%%%%%%%%%%%%%%%%%%%%
\subsubsection{Statistics of Dynamical Activity in a Three-Level System}\label{ex:3LvlAct}
%%%%%%%%%%%%%%%%%%%%%%%%%%%%%%%%%%%%%%%%%

We illustrate the results of Theorem \ref{theo:dynActC} with the model of a simple three-level system as sketched in Fig.~\ref{rf3level}(a): the set of configurations is $E = \{ 0,1,2 \}$, with reversible transitions $w_{01}=w_{10}=\omega$, $w_{02}=w_{20}=\upsilon$ and $w_{12}=w_{21}=\kappa$. Assuming that $\omega$ is the largest rate, the longest expected waiting time is $c_c=\frac{1}{\kappa+\upsilon}$, whilst $\langle t_{\EE} \rangle$, $b_c^2$ and $\varepsilon$ can easily be determined from the three-dimensional generator $\L$. 
%From Corollary \ref{countObsBeta} 
In addition, we have $\beta = c_c$, for $\A=\EE$. In Fig.~\ref{rf3level}(b) we show the 
exact long time rate function of the activity (full curve/black) 
for a particular set of values of the transitions rates, together with the lower bound from Theorem~\ref{theo:dynActC} (dashed/blue),
$$\tilde{I}_\EE(\gamma)=\frac{\gamma^2\varepsilon}{4b_c^2}h\left(\frac{5c_c\gamma}{2b_c^2}\right) ,$$
and the general lower bound from Theorem~\ref{theo:countObs} (dotted/red),
$$\hat{I}_\EE(\gamma)=\frac{\gamma + \langle t_{\EE} \rangle}{c_c} - 1-\log\left(\frac{\gamma + \langle t_{\EE} \rangle}{c_c}\right). $$
We see that the bound from Theorem~\ref{theo:countObs}  $\hat{I}_\EE(\gamma)$ 
is closer to the exact result than that from Theorem~\ref{theo:dynActC} for large enough deviations. 
Indeed, for $\gamma \gg 1$, one has $\hat{I}_{\EE}(\gamma)\asymp\frac{\gamma}{c_c}$ while $\tilde{I}_{\EE}(\gamma)\asymp \frac{\gamma\varepsilon}{5c_c}<\frac{\gamma}{c_c}$. For comparison, Fig.~\ref{rf3level}(b) we also show the {\em upper} bound to the rate function, the so-called TUR, from \cite{garrahan2017simple} (dot-dashed/magenta): the combination of the TUR and the ``inverse TUR'' from Theorems~\ref{theo:dynActC},\ref{theo:countObs} upper and lower bound the true rate function thus restricting the range of probabilities of rare events of the activity.

\begin{figure}[t]
    \begin{center}
    \includegraphics[width=1\linewidth]{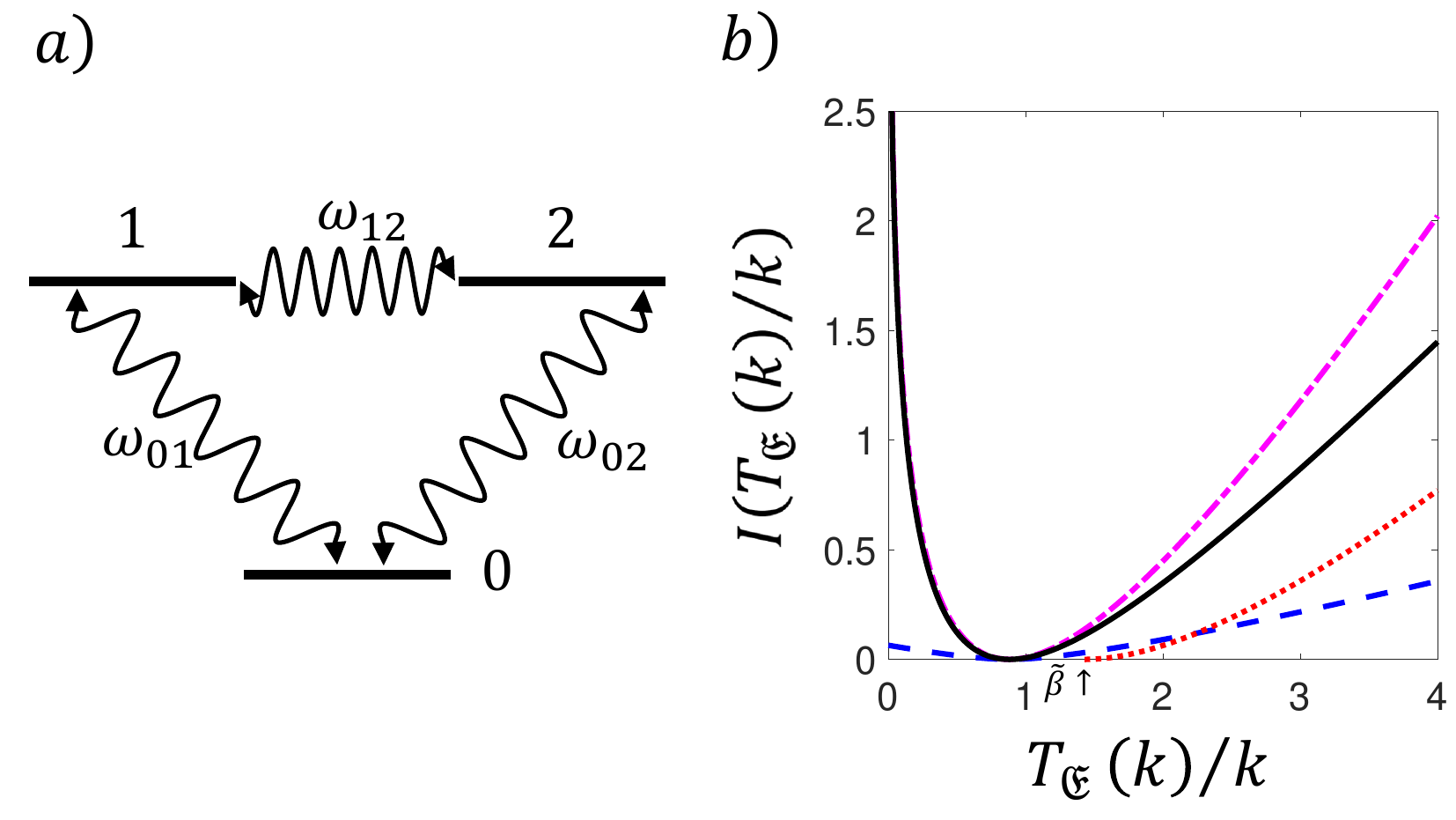}
    \end{center}
    \caption{
        {\bf Bounds on the rate function of the FPT of the activity in a classical three-level system.}
        (a) Sketch of the three-level system. (b) Rate function $I(T_{\EE}(k)/k)$ of the FPT for the dynamical activity, for the case with rates $w_{01}=w_{10}=\omega=1$, $w_{02}=w_{20}=\upsilon=0.5$, $w_{12}=w_{21}=\kappa=0.2$. We show the exact rate function (full curve/black) and the lower bound specific to the activity from Theorem \ref{theo:dynActC} (dashed/blue). We also show the 
        the generic tail bound for counting observables from Theorem \ref{theo:countObs} (dotted/red) which is valid in the region $T_{\EE}(k)/k>\tilde{\beta}=1/R_{\rm min}=1/(\kappa+\upsilon)$ (indicated by the arrow). For comparison we include the upper bound on the rate function (dot-dashed/magenta), known as the TUR \cite{garrahan2017simple}.
    }
    \label{rf3level}
\end{figure}

\begin{figure}[t]
    \begin{center}
    \includegraphics[width=0.6\linewidth]{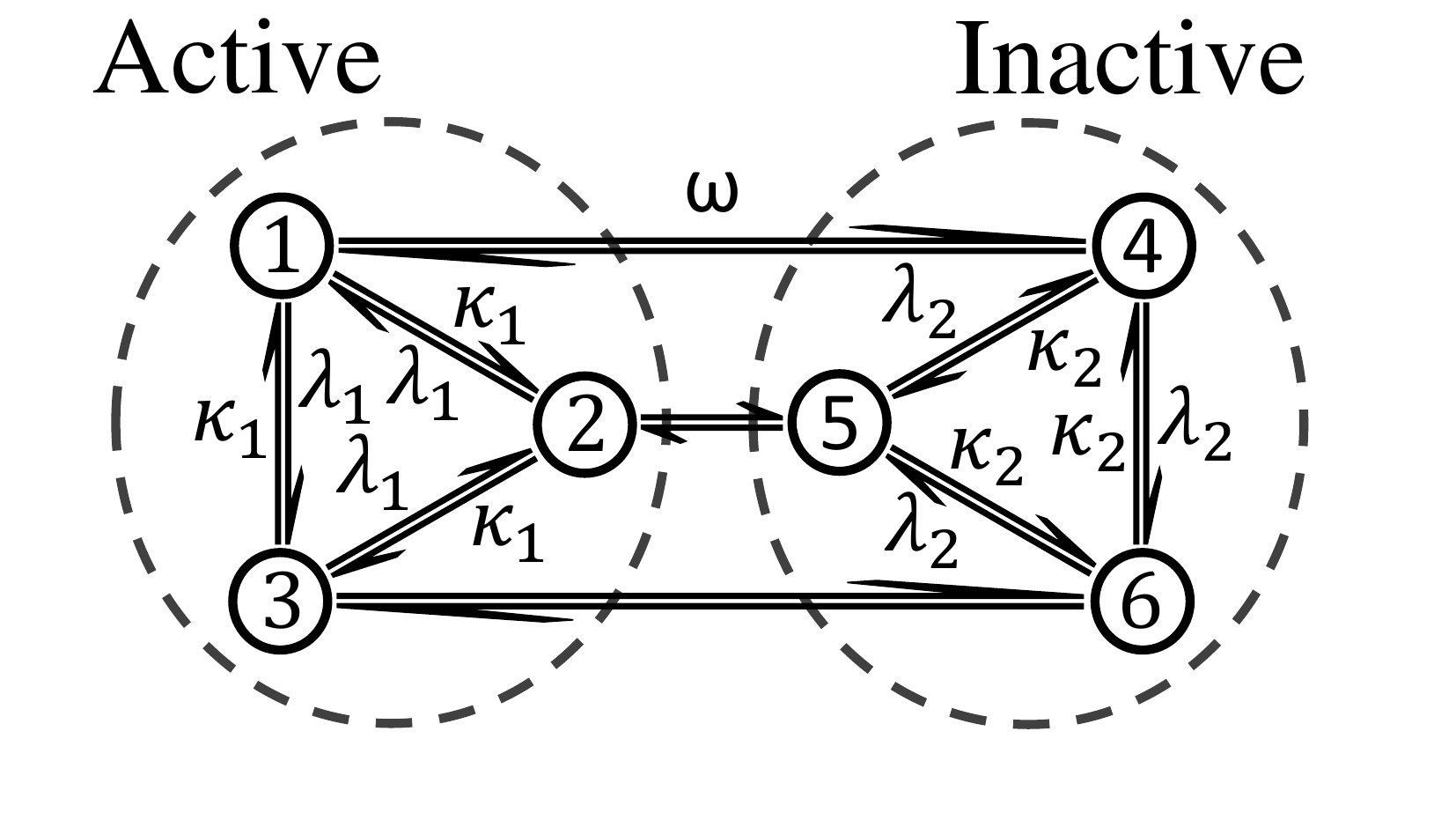}
    \end{center}
    \caption{{\bf Six-state dynamical system.}
    Sketch of a six-state system with two phases, the active phase $E_1$ in the left circle and the inactive phase $E_2$ in the right circle. The phases are separated by edges controlled by $\omega$. For small $\omega$, each phase is metastable, and for $\omega\rightarrow 0$ the size of FPT fluctuations increases. This increase is captured by $\varepsilon$.
    }
    \label{6statePh}
\end{figure}

%%%%%%%%%%%%%%%%%%%%%%%%%%
\subsubsection{Metastability and the Absolute Spectral Gap}\label{ex:phaseTran}
%%%%%%%%%%%%%%%%%%%%%%%%

In the following example we show how closing the absolute spectral gap of the discrete time generator ${\bf P}$ leads to an increase of the fluctuations of the total activity FPT in a simple model, as predicted by the concentration bound in Theorem \ref{theo:dynActC}. We consider the six-state system introduced in \cite{bakewell-smith2023general} composed of two three-state subsystems connected by edges controlled by a rate parameter $\omega$, see Fig.~\ref{6statePh}. For $\omega\rightarrow 0$ the spectral gap of the real part of the generator $\Re(\L)$ vanishes and the configuration space $E$ breaks up into two disconnected components, $E_1 = \{1,2,3\}$ and $E_2 = \{4,5,6\}$. When $\omega$ is non-zero but much smaller than the other rates, the system is metastable, with $E_1$ and $E_2$ becoming long-lived metastable ``phases'', since relaxation within $E_{1,2}$ will be much faster than relaxation in the whole of $E$. 

We now study the statistics of the FPT of the activity in this model. We consider the case where the internal rates in $E_1$ are much larger than those in $E_2$, while maintaining the metastability condition, $\lambda_1,\kappa_1\gg\lambda_2,\kappa_2 \gg \omega$. We call $E_1$ and $E_2$ the ``active phase'' and ``inactive phase'', respectively, as the activity in stationary trajectories is much larger while the system is in $E_1$ than in $E_2$. 
The rate matrix can be written as,
\[
\W = \begin{pmatrix}
    \tilde{\W}_1 & 0\\
    0 & \tilde \W_2
\end{pmatrix}+ \omega \begin{pmatrix}
        0 & \I_3\\
        \I_3 & 0
    \end{pmatrix},
\]
where $\tilde{\W}_{1,2}$ are the $3\times 3$ rate matrices for internal transitions in $E_{1,2} = \{1,2,3\}$, and the 3-dimensional identity is denoted $\I_3$. In Theorem \ref{theo:dynActC} we require the discrete time operator $\P=\RR^{-1}\W$ and its adjoint $\P^\dag$, to form the multiplicative symmertisation
\[
\P^\dag \P = \begin{pmatrix}
    \frac{\tilde{\W_1}^\dag \tilde{\W_1}}{R_1^2} & 0\\
    0 & \frac{\tilde{\W_2}^\dag \tilde{\W_2}}{R_2^2}
\end{pmatrix} + \omega^2\begin{pmatrix}
    \frac{1}{R_2^2}\I_3 & 0\\
    0 & \frac{1}{R_1^2}\I_3
\end{pmatrix}
+\omega\begin{pmatrix}
    0 & \frac{\tilde{\W}_1^\dag}{R_1^2} + \frac{\tilde{\W}_2}{R_2^2}\\
    \frac{\tilde{\W}_2^\dag}{R_2^2} + \frac{\tilde{\W}_1}{R_1^2} & 0
\end{pmatrix},
\]
where $R_{1,2} = \lambda_{1,2} + \kappa_{1,2}$. At $\omega=0$, the spectrum of $\P^\dag \P$ is equal to the union of the spectra of $\tilde{\P}_1^\dag\tilde{\P}_1$ and $\tilde{\P}_2^\dag\tilde{\P}_2$, where $\tilde{\P}_1 = \frac{\tilde{\W}_1}{R_1}$ is the discrete time transition matrix on each $E_1$ and $\tilde{\P}_2$ is that of $E_2$; hence, the algebraic multiplicity of the eigenvalue $1$ is $2$. By continuity of the spectrum for analytic perturbation, the absolute spectral gap vanishes as $\omega\rightarrow 0$. Corollary \ref{dynActiTURC} then implies that the upper bound on the variance of the FPT will explode as this ``phase transition'' point is approached. Fluctuations of $T_\EE(k)$'s get bigger as well: since $\omega$ is much less than either of the escape rates within each metastable phase, the system gets ``stuck'' in either phase, resulting in larger fluctuations of the observed FPT.
%\jpg{What is this section saying about the bounds? This section is not clear}

The behaviour of the fluctuations of $T_\EE(k)$ as $\omega \to 0$ is not immediately apparent based on the form of the expression for the variance given by Lemma \ref{lem:asympVar}. We remark that for finite $k$ the variance remains finite even as the gap closes. This can be seen by rewriting ${\rm var}_\pi(T_\EE(k))$ as

\begin{equation*}
    \begin{split}
    \frac{{\rm var}_\pi\left( T_{\EE}(k)\right)}{k} = \left \langle \pi, \D \Iv \right\rangle ^2 &+ 2 \left \langle \pi, \D\left ( \I + \frac{\sum_{i=1}^{k-1}\sum_{j=1}^{i}\P^j}{k}\right )(\I -\Pi)\D\Iv\right\rangle\\
    \end{split}
    \end{equation*}
which is uniformly bounded in $\omega$. Recall that in the case of total activity $\varphi = \pi$, $\L_\infty^{-1}=-\RR^{-1} = -\D$ and $\Q = \P$. In the limit $k\to \infty$ the expression reduces to the first two terms of Lemma \ref{lem:asympVar} and the behaviour depends solely on $(\I - \P)^{-1}$ and whether this causes a divergence as $\omega \to 0$. From Figure \ref{fig:varcomp} one can see that for this model, the asymptotic variance does diverge and for finite $k$ the fluctuations remain finite as expected. We can however see fingerprints of the asymptotic behaviour for intermediate $k$, which is captured by the upper bound in Corollary \ref{dynActiTURC}.

\begin{figure}[t]
    \begin{center}
    \includegraphics[width=0.5\linewidth]{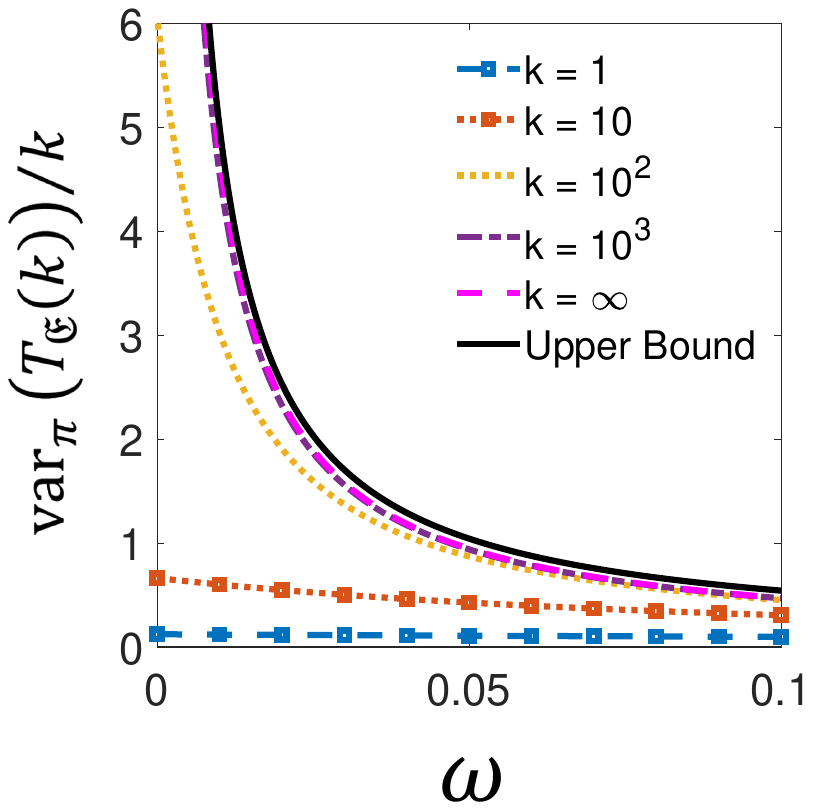}
    \end{center}
    \caption{{\bf Upper bound on the variance of the FPT for activity in a six-state system.} 
    Upper bound (full/black) on the scaled variance of ${\rm var}_\pi(T_\EE(k))/k$ given by Corollary \ref{dynActiTURC}. This is valid for all $k$. We compare with the exact variance (cf Lemma \ref{lem:asympVar}) for several values of $k$: $1$ (dashed-marked/blue), $10^1$ (dotted-marked/orange), $10^2$ (dotted/yellow), $10^3$ (dot-dashed/purple) and for $k=\infty$ (dashed/magenta). We compare these quantities as the controlling parameter $\omega\to 0$ and the system approaches a phase transition. The system is the model given in figure \ref{6statePh} with rates $\lambda_1 = 30$, $\mu_1 = 10$ and $\lambda_2 = 0.3$, $\mu_2 = 0.1$.}
    \label{fig:varcomp}
\end{figure}

\begin{figure}[t]
    \begin{center}
    \includegraphics[width=0.5\linewidth]{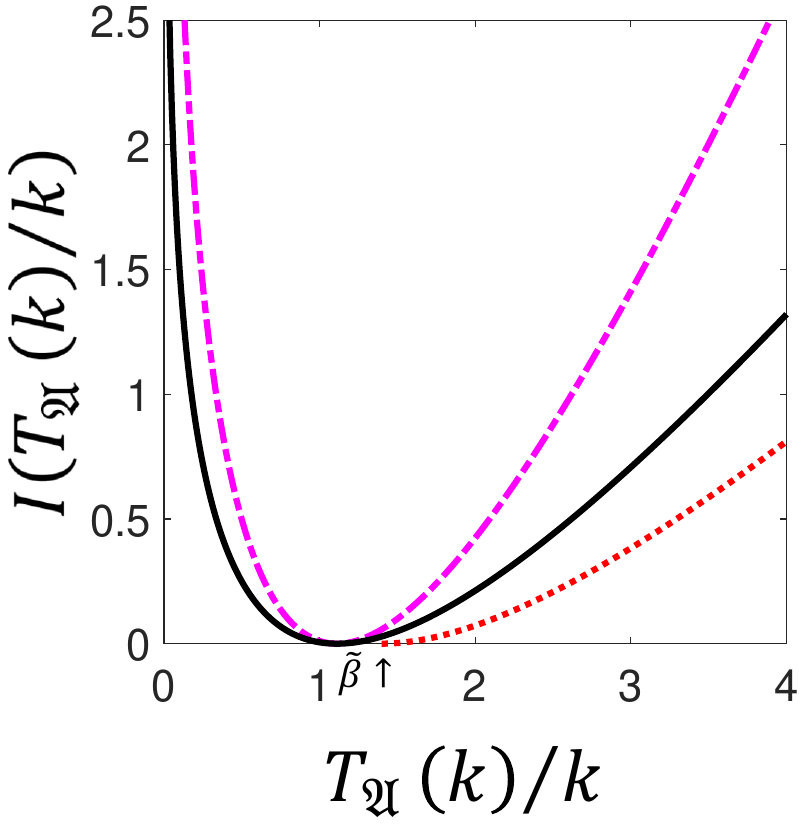}
    \end{center}
    \caption{{\bf Bounds on the rate function of the FPT of a counting observable for a classical three-level system.} 
    Exact rate function $I(T_{\A}(k)/k)$ (full/black) for the FPT of the total count of  
    jumps $0\rightarrow 1$, $1\rightarrow 2$ and $2\rightarrow 0$ in the three-level system of Fig.~\ref{rf3level}, with rates $w_{01}=w_{10}=\omega=1$, $w_{02}=w_{20}=\upsilon=0.9$, $w_{12}=w_{21}=\kappa=0.8$. The tail bound from Theorem \ref{theo:countObs} (dotted/red) bounds deviations in the region $T_{\A}(k)/k>\tilde{\beta}$ with $\beta\leq \tilde{\beta}={(\omega + \upsilon)}/{[\kappa(\kappa+\upsilon)]}$. The rate function is bounded from above using the same method as in Fig.~\ref{rf3level}(dot-dashed/magenta) \cite{garrahan2017simple}.}
    \label{rfCount}
\end{figure}

%%%%%%%%%%%%%%%%%%%%%%%%%%
\subsubsection{Three-Level System Counting Subset of Jumps}
%%%%%%%%%%%%%%%%%%%%%%%%%%

To illustrate the results of Theorem \ref{theo:countObs}, we use the same three-level model as in Sec.~\ref{ex:3LvlAct} but we consider the observable that only counts ``clockwise'' jumps, that is, the $0\rightarrow 1$, $1\rightarrow 2$ and $2\rightarrow 0$ jumps but not their reverses. With this setup, the minimax jump distance is $\tilde{k}=1$ since it is possible to perform a jump in $\A$ which begins at any state. The lower bound to the rate function provided by Theorem \ref{theo:countObs}  is illustrated in Figure \ref{rfCount}.

%%%%%%%%%%%%%%%%%%%%%%%%%%%
\section{Large Deviation Principle and Concentration Bounds for FPTs of Quantum Counting Processes}\label{sec:qmarkov}

%%%%%%%%%%%%%%%%%%%%%%%%%%%%
In this section we present three concentration bounds for \emph{quantum} Markov processes. First we provide a result for \emph{total counts} in quantum Markov processes with general jump operators. By restricting to  \emph{reset processes} (jump operators of rank one), we can weaken the required assumption and obtain a separate result. Finally, we provide a bound for counting a \emph{subset} of jumps. In each of the three cases we illustrate the result in a simple example.

%%%%%%%%%%%%%%%%%%%%%%%%%
\subsection{Preliminaries and Notation}\label{sec:qPrelims}
%%%%%%%%%%%%%%%%%%%%%%
In this section we introduce the basic concepts and tools related to first passage times for  quantum counting processes. We recall that quantum counting processes are used to model detector clicks  when an open quantum system is continuously monitored through the environment \cite{gardiner2004quantum}. Throughout the paper, the system will be finite dimensional and its state space will be the Hilbert space $\mathbb{C}^d$. Quantum states are represented by positive semi-definite matrices with unit trace, that is $\rho\in M_d(\mathbb{C})$ such that $\rho \geq 0$ and $\tr(\rho)=1$.
%Complex functions on the state space are replaced by elements in the algebra of $d$-dimensional complex matrices $M_d(\mathbb{C})$. 
Observables correspond to self adjoint operators on the state space, i.e. $x\in M_d(\mathbb{C})$ such that $x=x^*$. As in the classical case, there is a duality between states and observables expressed by the fact that the expectation of an obsevable $x$ in the state $\rho$ is $\tr(\rho x)$. We can endow $M_d(\mathbb{C})$ with the operator and trace norms which provide natural distances between observables and states, respectively:
\[
\|x\|_\infty:=\max_{u \in \mathbb{C}^d\setminus \{0\}}\frac{\|xu\|}{\|u\|}, \quad 
\|x\|_1:=\tr(|x|).
\]
For any linear map $\Phi$ on $M_d(\mathbb{C})$ describing the evolution of observables (Heisenberg picture), the unique corresponding evolution $\Phi_*$ on states (Schr\"{o}dinger picture) is characterised by
$$
\tr(x\Phi(y))=\tr(\Phi_*(x)y).$$
As in the previous section, we denote by $\|\Phi\|_{\infty \rightarrow  \infty}$ the operator norm on $\Phi$ induced by $\|\cdot \|_\infty$ and analogously for the trace norm. Every physical evolution of a quantum system is given by a quantum channel, i.e. a linear map $\Phi:M_d(\mathbb{C}) \rightarrow M_d(\mathbb{C})$ which satisfies
\begin{enumerate}
\item $\Phi(\I)=\I$ (unital),
\item $\Phi \otimes 
{\bf I}_{M_k(\mathbb{C})}$ is positive for every $k \in \mathbb{N}$ (completely positive)
\end{enumerate}
where $\I$ is the identity matrix and ${\bf I}_{M_k(\mathbb{C})}$ is the identity map on $M_k(\mathbb{C})$. Equivalently, $\Phi_*$ is trace preserving and completely positive. These conditions mirror those satisfied by classical channels/transition operators, but in the quantum setting complete positivity is a stronger requirement that usual positivity, and we refer to  
\cite{nielsen2010quantum} for more details on the theory of quantum channels and the physical interpretation.

%%%%%%%%%%%%%%%%%%%%%%%%%%%%%%%%%%%%%%%%%%%%%%%%%%%%%
\subsubsection{Quantum Markov Semigroups and their Unraveling by Counting Measurements}\label{subsec:QMP}
%%%%%%%%%%%%%%%%%%%%%%%%%%%%%%%%%%%%%%%%%%%%%%%%%%%%%

A quantum Markov semigroup is a family of channels  $\left(\mathcal{T}_t\right)_{t\ge 0}$ acting on $M_d(\mathbb{C})$ such that $\mathcal{T}_0 = {\bf I}_{M_d(\mathbb{C})}$, $\mathcal{T}_t\circ \mathcal{T}_s= \mathcal{T}_{t+s}$ for all $s,t\geq 0$ and 
$t \mapsto {\cal T}_t$ is continuous. Such a semigroup describes (in the Heisenberg picture) the dissipative evolution of a $d$-dimensional open quantum system, in physical situations where certain Markov approximations pertaining to the interaction with the environment apply. A fundamental result \cite{gorini1976completely,lindblad1976on-the-generators} shows that the generator $\mathcal{L}:M_d(\mathbb{C}) \to M_d(\mathbb{C}) $ of such a semigroup takes the form
\begin{equation}
\label{eq:Lindblad}
\mathcal{L}: x \mapsto -i[H,x]
+ \sum_{i \in I}L_i^*xL_i -
-\frac{1}{2}\sum_{i \in I} (L_i^*L_i x +xL_i^*L_i), 
\end{equation}
where $H\in M_d(\mathbb{C})$ is self-adjoint  and 
$L_i \in M_d(\mathbb{C})$ with indices beloging to a finite set $I$. Physically, $H$ is interpreted as being the system hamiltonian while $L_i$ describe the coupling to separate ``emission channels'' in the environment. If the system is prepared in a state $\rho$ and evolves together with the environment for a time period $t$, then its reduced state is given by $\mathcal{T}_{t*}(\rho)$. On the other hand, if the environment is probed by performing continuous-time counting measurements in each of the emission channels, then one  
one observes stochastic trajectories 
$\omega=\{(i_1,t_1),(i_2,t_2),\cdots\}$ 
which record the labels of the jumps together with the times between jumps. In this case one would like to know what is the probability of observing such a trajectory and what is the conditional state of the system given this observation. This is the subject of quantum filtering theory which plays an important role in quantum technology and quantum control theory 
\cite{belavkin1990a-stochastic,bouten2004stochastic,wiseman2009quantum,gough2009the-series}. 
While a full account of the system-environment unitary evolution and subsequent counting measurement goes beyond the scope of this paper, we employ the Dyson series to convey an intuitive answer to the questions formulated above. %$\mathcal{T}_{*t}$ (in the Schr\"{o}dinger picture). 
For this we decompose the generator as 
\begin{equation}
\label{eq:L.split}
\mathcal{L}= 
\mathcal{L}_0 +  \J= \mathcal{L}_0 + \sum_{i\in I} \J_i   ,  
\end{equation}
where
$$
\mathcal{L}_0(x) = G^*x+ xG %-i[H,x]-\frac{1}{2} \sum_i \{L_i^*L_i, x\}
\qquad
\J_i (x)= L_i^* x L_i,
\quad
\mathrm{with}
\quad
G:= iH -\frac{1}{2}\sum_{i\in I} L_i^*L_i.
$$
Note that $\J_i$ is completely positive and  $\mathcal{L}_0$ is the generator of the completely positive semigroup $e^{\mathcal{L}_0} (x) = e^{G^*t} x e^{Gt}$. 
The Dyson expansion of $\mathcal{T}_{*t}$ 
(Schr\"{o}dinger picture) corresponding to the split \eqref{eq:L.split} is
$$
\mathcal{T}_{t*} =
e^{t{\cal L}_{0*}}+\sum_{k=1}^{+\infty}  \int_{\sum_{i=1}^k t_i \leq t} e^{\left (t-\sum_{i=1}^{k}t_i\right ){\cal L}_{0*}}{\cal J}_{i_k*}e^{t_k{\cal L}_{0*}} \cdots {\cal J}_{i_1*} e^{t_1{\cal L}_{0*}} dt_1 \cdots dt_k .
%=\\
%&e^{t {\cal L}_*}(\rho)
$$
By applying both sides to the initial state $\rho$ we find that the evolved system state 
$\rho_t= \mathcal{T}_{*t}(\rho)$ is a mixture of states corresponding to different counting trajectories. Indeed, let us denote
\[\Omega_t=\{\emptyset\} \cup \bigcup_{k=1}^{+\infty} I^k \times \{(t_1,\dots,t_k)\in [0,t]^k: \sum_{i=1}^{k}t_i \leq t\}
\]
the space of counting trajectories up to time $t$, 
and let us endow $\Omega_t$ with the natural $\sigma$-field and denote by $d\mu$ the unique measure such that $\mu(\{\emptyset\})=1$ and $\mu(\{(i_1,\dots, i_k)\times B)$ is the Lebesgue measure of $B$ for every $(i_1,\dots, i_k) \in I^k$, $B \subseteq \{(t_1,\dots,t_k)\in [0,t]^k: \sum_{i=1}^{k}t_i \leq t\}$.
Then we can write
$$
\mathcal{T}_{t*} (\rho)= 
\int_{\Omega_t} \tilde{\varrho}_t(\omega) \mu(d\omega)
=\int_{\Omega_t} 
\frac{d\mathbb{P}_t}{d\mu} (\omega)\varrho_t(\omega) \mu(d\omega)
$$
where for each counting trajectory 
$\omega =\{(i_1 , t_1),\dots (i_k, t_k)\}\in \Omega_t$. The \emph{unnormalised} system state conditional on observing $\omega$ is given by 
\begin{equation}  \label{eq:unnorm.conditional.state}
\tilde{\varrho}_t(\omega) 
 = e^{\left (t-\sum_{i=1}^{k}t_i\right ){\cal L}_{0*}}{\cal J}_{i_k*} e^{t_1{\cal L}_{0*}}\cdots {\cal J}_{i_1*} e^{t_1{\cal L}_{0*}}(\rho) 
\end{equation}
while 
\begin{equation}
    \label{eq:conditional.state}
\frac{d\mathbb{P}_t}{d\mu} (\omega) = \tr(\tilde{\varrho}_t(\omega)), \quad
\varrho_t(\omega) =
\frac{\tilde{\varrho}_t(\omega)}{\tr(\tilde{\varrho}_t(\omega))}
\end{equation}
represent the probability density, and the \emph{normalised} conditional state, respectively.

With this interpretation, the Dyson expansion expresses the fact that by averaging over all the conditional states $\varrho_t(\omega)$ one obtains the reduced system state $\rho_t$. Note that in order to avoid confusion, we use different symbols for the conditional and reduced system states.

Based on equations \eqref{eq:unnorm.conditional.state} and \eqref{eq:conditional.state} we deduce that during time periods with no jumps the conditional state evolves continuously as 
$$
\varrho_t \mapsto \varrho_{t+s }:= 
\frac{e^{s\mathcal{L}_{0*}}(\varrho_t)}
{ \tr [e^{s\mathcal{L}_{0*}}(\varrho_t )]},
$$

and at the time of a count with index $i$ the state has an instantaneous jump
$$
\varrho_t\mapsto 
\frac{\J_{i*}(\varrho_t)}{\tr [\J_{i*}(\varrho_t)]}.
$$
In addition, the probability density for the time of the first jump after $t$ is
$$
w(s) = \tr [\J_*e^{s\mathcal{L}_{0*}}(\varrho_t) ] = -\tr [\mathcal{L}_{0*}e^{s\mathcal{L}_{0*}}(\varrho_t)]
$$ 
where we use the fact that $\J(\mathbf{1})+ \mathcal{L}_{0}(\mathbf{1}) =\mathcal{L}(\mathbf{1}) =0$. 
We will now show how to generate the count trajectories in a recursive manner which is reminiscent of the generation of trajectories of classical Markov processes. Given a trajectory $\omega= \{(i_1,t_1),(i_2,t_2), \dots\} $, we denote by $\varrho_k$ the state immediately after the $k^{\rm th}$ count, with $\varrho_0= \rho$ denoting the initial state. 

\bigskip

\noindent
{\bf Iterative procedure for generating quantum trajectories.} The interarrival times and quantum trajectories can be generated recursively with respect to $k=0,1,\dots$: given $\varrho_k$ we draw $(\varrho_{k+1}, i_{k+1}, t_{k+1})$ as follows:
\begin{enumerate}
    \item the $(k+1)$th interrarival time 
    $t_{k+1}$ is drawn from the density 
    $$
    w(t)=-\tr(\LL_{0*}e^{t\LL_{0*}}(\varrho_{k}))
    $$
    \item \label{step.2}
    given $t_{k+1}$, the label $i_{k+1}$ is sampled from the following distribution:
    $$p(j)=\frac{\tr(\Phi_{j*}\LL_{0*}e^{t_{k+1}\LL_{0*}}(\varrho_{k}))}{\tr(\Phi_*\LL_{0*}e^{t_{k+1}\LL_{0*}}(\varrho_{k}))},  \quad \Phi_j:=-\LL_0^{-1}{\cal J}_j, \quad \Phi := \sum_{j\in I} \Phi_j
    $$
    \item we define
    $$\varrho_{k+1}=\frac{{\cal J}_{i_{k+1}*}(e^{t_{k+1}\LL_{0*}}(\varrho_{k}))}{\tr({\cal J}_{i_{k+1}*}(e^{ t_{k+1}\LL_{0*}}(\varrho_{k})))}.$$
\end{enumerate}
 
The map $\Phi:=\sum_{j\in I} \Phi_{j}=-\LL_0^{-1}{\cal J}$ appearing in step 2. is the analogous of $\P$ in Section \ref{sec:pn} and will play a central role in the following. Lemma \ref{lem:tech1Q} shows that 
$\LL_0^{-1}$ is well defined and is equal to $-\int_0^{+\infty}e^{t \LL_0} dt$, hence $\Phi$ is a completely positive map. Moreover, using that $\LL(\I)=0$, one has
$$
\Phi(\I)=-\LL_0^{-1}{\cal J}(\I)=\LL_0^{-1}\LL_0(\I)=\I,$$
hence $\Phi$ is a quantum channel. 

\bigskip

As in the classical case, establishing results on law of large numbers, large deviations, or concentration bounds, requires some type of assumption on the ergodicity of the dynamics. We now introduce two irreducibility assumptions which will later be invoked in separate occasions in our results.
\bigskip

\begin{hypo}[Irreducibility of $\mathcal{L}$]\label{hypo:irrQuantL}
The generator $\mathcal{L}$ is irreducible. This means that there is no non-trivial projection $P$ such that $\mathcal{L}(P)\geq 0$ or equivalently, there exists a unique strictly positive state $\hat{\sigma}$ satisfying $\mathcal{L}_*(\hat{\sigma})=0$.
\end{hypo}

\bigskip 

\begin{hypo}[Irreducibility of $\Phi$]\label{hypo:irrQuantPhi}
The channel $\Phi$ is irreducible. This means that there is no non-trivial projection $P$ such that $\Phi(P)\geq P$ or equivalently, there exists a unique strictly positive state $\sigma$ satisfying $\Phi_*(\sigma) = \sigma$.
\end{hypo}

As in the classical case, there is a close connection between the continuous-time generator $\LL$ and the channel $\Phi$. The following lemma clarifies these connections and shows that Hypothesis \ref{hypo:irrQuantPhi} is strictly stronger than Hypothesis  \ref{hypo:irrQuantL}.

\bigskip

\begin{lemmaT}\label{lemma.statioarystate.generators}
The generator $\LL$ has a unique invariant state if and only if $\Phi$ does. If $\Phi$ is irreducible then $\mathcal{L}$ is irreducible, but the converse is generally not true.
\end{lemmaT}

The proof of Lemma \ref{lemma.statioarystate.generators} can be found in the appendix, section \ref{app:markovProofQ}.

%%%%%%%%%%%%%%%%%%%%%%%%%%%%%%%%%%%%%%%%%%%%%%%%%%%%%%%%
\subsubsection{First Passage Time for the Counting Process}\label{sec:fptCount}
%%%%%%%%%%%%%%%%%%%%%%%%%%%%%%%%%%%%%%%%%%%%%%%%%%%%%%%%%%

Consider the counting process described in section \ref{subsec:QMP} and let $N_i(t)$ be the stochastic process given by the number of counts with label $i\in I$ up to time $t$ in the measurement trajectory $\omega$. More generally, for any subset $\A\subseteq I$ we define the counting observable 
\[
N_{\A}(t) = \sum_{i\in \A}N_i(t).
\]
When $\A=I$, $N_{I}(t)$ is referred to as the total number of counts. The corresponding first passage times (FPTs) are defined in the same way as in the classical case:
\begin{equation}
\label{eq:FPT.quantum}
T_{\A}(k):=\inf_{t\geq 0}\{t: N_{\A}(t)=k\}.
\end{equation}

The following splitting of the generator is relevant in order to study the properties of the stochastic process $T_{\A}(k)$:
\[
\mathcal{L} = \J_\A + \mathcal{L}_{\infty},
\]
where $\J_\A(x) = \sum_{i\in \A} L_i^* x L_i$ accounts for the change of state after a jump in $\A$ and $\LL_\infty$ for the average evolution between jumps in $\A$. We denote as $\Psi$ the transition operator analogous to $\Q$ in the classical case:
\begin{equation} \label{eq:psi}
\Psi (x) = -\mathcal{L}_{\infty}^{-1}\circ \J_\A (x).
\end{equation}
If Hypothesis \ref{hypo:irrQuantL} holds, then $\Psi$ admits a unique invariant state $\varsigma$ (which might have a non-trivial kernel). Note that if $\A = I$ then $\mathcal{L}_{\infty} = \LL_0$ and $\Psi = \Phi$.

The following lemma shows that all the objects introduced so far are well defined and allows us to write the Laplace transform for general counting observables.

\begin{lemma}\label{lem:tech1Q}
Assume that Hypothesis \ref{hypo:irrQuantL} ($\mathcal{L}$ is irreducible) holds. Then the following statements are true:
    \begin{enumerate}
        \item $\overline{\lambda}:= -\max\{\Re(z): z \in {\rm Sp}({\cal L}_\infty)\}>0$, hence ${\cal L}_\infty$ is invertible;
        \item for every $u < \overline{\lambda}$, one has
        \[
        \E_\rho[e^{uT_{\A}(k)}]=\tr\left ( \rho \left((u+\mathcal{L}_\infty)^{-1}\mathcal{L}_\infty\Psi \right)^k(\I)\right ),
        \]
        \item $\|{\cal L}_\infty^{-1}\|_{\infty \rightarrow \infty}^{-1} \leq \overline{\lambda}.$
    \end{enumerate}
\end{lemma}
The proof of Lemma \ref{lem:tech1Q} can be found in Appendix \ref{app:markovProofQ}. From the expression of the moment generating function given in point 2 of the previous lemma, one can use standard techniques to show that
\begin{equation}\label{eq:asymptotic.fpt}
\frac{T_\A(k)}{k} \xrightarrow{a.s.} \langle t_\A \rangle:=-\tr\left (\varsigma \mathcal{L}_\infty^{-1}(\I)\right)
%=\tr \left(\J_{\A*} (\varsigma)\right) 
\end{equation}
where $\varsigma$ is the unique invariant state of $\Psi$ defined in equation \eqref{eq:psi}. Our goal will be to investigate what is the probability of $T_\A(k)$ deviating from $\langle t_\A \rangle$.

\bigskip We introduce some more notation that will be useful in proving the concentration bounds for the FPTs. We consider the Hilbert space structure of $M_d(\mathbb{C})$ endowed with the following inner product:
\[
\langle x,y\rangle_\sigma:=\tr(\sigma^{\frac{1}{2}}x^*\sigma^{\frac{1}{2}}y), \quad x,y\in M_d(\mathbb{C})
\]
and we denote it by $L^2(\sigma)$. Unlike the classical case, there are infinitely many inner products induced by the stationary state $\sigma$ of $\Phi$; the choice we adopt is known as Kubo-Martin-Schwinger (KMS) inner product. The norm with respect to this inner product will be denoted by $\|x\|_\sigma$. The KMS inner product allows us to define the trace of a map $\mathcal{E}:M_d(\mathbb{C})\rightarrow M_d(\mathbb{C})$, by:
\[
\TR(\mathcal{E})=\sum_{i=1}^{d^2}\langle x_i,\mathcal{E}(x_i)\rangle_\sigma,
\]
for an orthonormal basis $\{x_i\}$ of $M_d(\mathbb{C})$. The adjoint of an operator $\mathcal{E}$ with respect to this inner product can be expressed in terms of the predual map $\mathcal{E}_*$ as 
\begin{equation}
\label{eq:dagger.adjoint}
\mathcal{E}^\dag(x)=\Gamma^{-\frac{1}{2}}\circ\mathcal{E}_*\circ\Gamma^{\frac{1}{2}}(x)
\end{equation}
where $\Gamma^a(x)=\sigma^ax\sigma^a$ for every $a\in \mathbb{R}$. 

Given a quantum channel $\Phi$ with invariant state $\sigma$, its absolute spectral gap $\varepsilon$ is defined as $1$ minus the square root of the second largest eigenvalue of the multiplicative symmetrisation of $\Phi$, namely $\Phi^\dag\Phi$. As in the classical case, the proof of the concentration bound in Theorem \ref{theo:dynActQ} will make use of the following Le\'{o}n-Perron operator corresponding to $\Phi$ 
\[
\hat{\Phi}=(1-\varepsilon)\Id+\varepsilon\Pi,
\]
where $\Id$ is the identity map, and $\Pi$ is the map $\Pi x\mapsto \tr(\sigma x)\I$ for $x\in M_d(\mathbb{C})$. $\hat{\Phi}$ is a quantum channel with unique invariant state $\sigma$ and which is self adjoint with respect to the KMS inner product induced by $\sigma$.

%%%%%%%%%%%%%%%%%%%%%%%%%%%%%
\subsection{Results on Quantum Markov Processes}
%%%%%%%%%%%%%%%%%%%%%%%%%%%%%

We now present our three concentration bounds for quantum Markov processes, two for the total counts process, followed by a bound on the right tail for counts of a subset of jumps. We then illustrate the three results with simple examples.

%%%%%%%%%%%%%%%%%%%%%%%%%%%%%%%%

\subsubsection{Large Deviation Principle for General Counting
Observables}
\label{subsec:qLDP}

We start by obtaining a large deviation principle for quantum counting processes, in analogy with what we did in Sec.~\ref{sec.LD.classical}. 

\bigskip\begin{theo}\label{theo:ldpQ} Consider a nonempty subset $\A$ of the emission channels. The FPT $T_\A(k)/k$ satisfies a large deviation principle with good rate function given by
$$
I_\A(t):=\sup_{u \in \mathbb{R}}\{ut-\log(r(u))\}$$
where
$$
r(u)=\begin{cases} r \left (\Psi_u \right ) & \text{ if } u < \overline{\lambda}\\
+\infty & \text{o.w.}\end{cases}$$
where $\Psi_u(x):= -(u+\LL_\infty)^{-1}\J_\A(x)$ and $\overline{\lambda}:=-\max\{\Re(z):z \in {\rm Sp}(\LL_\infty)\}.$
\end{theo}

The proof of Theorem \ref{theo:ldpQ} is in Appendix \ref{app:proof.ldpQ}.

\subsubsection{Concentration Bound on Total Number of Counts}\label{sec:genQuant}

%%%%%%%%%%%%%%%%%%%%%%%%%%%%%%%

The first passage time for the total number of counts $T_I(k)$ is the time it takes to observe $k$ counts of any kind on the system. Our first main result for quantum Markov processes is a quantum version of Theorem \ref{theo:dynActC} - a bound on the fluctuations of the FPT $T_I(k)$ for total jumps. We note that in the quantum framework, ``activity'' is usually referred to as total ``counts'' or ``jumps'' \cite{garrahan2010thermodynamics} (but other definitions exist \cite{nishiyama2023exact}). From equation \eqref{eq:asymptotic.fpt}, the asymptotic mean in this case is:
\[
\langle t_I \rangle := -\tr(\sigma\mathcal{L}_0^{-1}(\I)).
\]
We define 
\begin{equation}
\label{eq:cq}
c_q:=\|\mathcal{L}_0^{-1}\|_{\sigma\rightarrow \sigma}
\end{equation}
and note that this is the non-commutative counterpart of $c_c$, cf.~\eqref{eq:cc}.

\bigskip

\begin{theo}[Fluctuations of FPT for Total Counts ]\label{theo:dynActQ}
Assume that Hypothesis \ref{hypo:irrQuantPhi} holds ($\Phi$ be irreducible) and let  $\varepsilon$ be the absolute spectral gap of $\Phi$. Then, for every $\gamma>0$:

\begin{equation*}
\begin{split}
&\PP_\rho \left (\frac{T_I(k)}{k} \geq \langle t_I \rangle + \gamma \right) \leq C(\rho) \exp \left ( -k \frac{\gamma^2 \varepsilon}{8c_q^2}h\left ( \frac{5\gamma}{2 c_q}\right )\right )\\{\rm and}\\
&\PP_\rho \left (\frac{T_I(k)}{k} \leq \langle t_I \rangle - \gamma \right) \leq C(\rho) \exp \left ( -k \frac{\gamma^2 \varepsilon}{8{c_q^2}}h\left ( \frac{5\gamma}{2 c_q}\right )\right ),\quad k\in \mathbb{N},
\end{split}
\end{equation*}
where $h(x):=(\sqrt{1+x}+\frac{x}{2}+1)^{-1}$, $C(\rho):=\left\|\sigma^{-\frac{1}{2}}\rho\sigma^{-\frac{1}{2}}\right\|_\sigma$ and $c_q$ is defined in Eq. \eqref{eq:cq}.
\end{theo}
The proof of Theorem \ref{theo:dynActQ} can be found in Appendix \ref{app:proof.th.dynActQ}.
As in the classical case, the following corollary follows from the proof of Theorem     \ref{theo:dynActQ}.

\bigskip

\begin{coro}\label{dynActiTURQ}
The variance of the first passage time for total counts is bounded from above by:
\[
\frac{{\rm var}_\sigma(T_{I}(k))}{k}\leq \left (\frac{4}{\varepsilon}-(1-\varepsilon)\right ) c_q^{2}.
\]
\end{coro}
The proof of Corollary \ref{dynActiTURQ} can be found in Appendix \ref{app:proof.th.dynActQ}.
Recall from section \ref{sec:qPrelims} that if Hypothesis \ref{hypo:irrQuantL} holds ($\LL_*$ admits a unique and strictly positive invariant state), then the uniqueness of the invariant state of $\Phi_*$ is guaranteed, but not its strict positivity, hence the need for Hypothesis \ref{hypo:irrQuantPhi}. One can show however, that if Hypothesis \ref{hypo:irrQuantL} holds, then the invariant state of $\Phi_*$ is strictly positive if and only if $\cap_{i=1}^{|I|}\ker (L_{i}^*)= \emptyset$.
%%%%%%%%%%%%%%%%%%%%%%%%%%%
\subsubsection{Concentration Bound on Total Number of Counts for Reset Processes}\label{sec:reset}
%%%%%%%%%%%%%%%%%%%%%%
In Theorem \ref{theo:dynActQ} we proved a concentration bound for the FPT corresponding to the total number of counts, under the assumption that Hypothesis \ref{hypo:irrQuantPhi} holds. In this subsection we consider \emph{quantum reset processes} which are characterised by jump operators that have rank one, and we  derive a FPT concentration bound using the weaker Hypothesis \ref{hypo:irrQuantL}.

Let us assume that the jump operators are of the form
\begin{equation}\label{eq.reset}
L_{i}=\ket{y_i}\bra{x_i} \quad x_i,y_i \in \mathbb{C}^d \setminus \{0\}.
\end{equation}
Without any loss of generality, we can assume that $\|y_i\|=1$. After observing a click of the $i$-th detector, the state of the system is known and is equal to $\ket{y_i}\bra{y_i}$. In this case, by applying step \ref{step.2} of the iterative procedure in section \ref{subsec:QMP} we find that the sequence of click indices is a classical Markov chain on $I$ with transition matrix $\P:=(p_{ij})$
\begin{equation}\label{transProbReset}
p_{ij} = - \langle x_j|\LL_{0*}^{-1}(\ket{y_i}\bra{y_i})x_j \rangle.
\end{equation}
We remark that Hypothesis \ref{hypo:irrQuantL} is sufficient to imply the irreducibility of the classical transition operator $\P$. Indeed  since $\mathcal{L}_*(\hat{\sigma})=0$ with stationary state $\hat{\sigma}>0$, we have
$$
\mathcal{L}_{0*} (\hat{\sigma}) = 
-\J_* (\hat{\sigma}) =
-\sum_{i\in I} \langle x_i|\hat{\sigma}|x_i\rangle \cdot |y_i\rangle \langle y_i|
$$
which implies
\[\begin{split}
\sum_{i\in I} \langle x_i|\hat{\sigma}|x_i \rangle 
p_{ij} = -\sum_{i \in I}\langle x_i|\hat{\sigma}|x_i \rangle \langle x_j|\LL_{0*}^{-1}(\ket{y_i}\bra{y_i})x_j \rangle=\\
-\left \langle x_j\left |\LL_{0*}^{-1}\left (\sum_{i \in I}\langle x_i|\hat{\sigma}|x_i \rangle \ket{y_i}\bra{y_i}\right )\right |x_j \right \rangle=\langle x_j|\hat{\sigma}|x_j \rangle
\end{split}\]
so the stationary state of $\P$ is
 $$\pi(i):= \frac{\langle x_i| \hat{\sigma}| x_i \rangle}{\sum_{j \in I}\langle x_j| \hat{\sigma}| x_j \rangle}
 $$
 which is fully supported since $\hat{\sigma} >0$. 

The waiting times are not exponentially distributed as in the case of a classical continuous time Markov process, instead their probability density function after observing a click of the type $i$ is given by:
\begin{equation}\label{pdfReset}
f_i(t):=-\tr(\LL_{0*}e^{t\LL_{0*}}(\ket{y_i}\bra{y_i}).
\end{equation}
We now introduce the quantities used in the result of this section:

\begin{itemize}
    \item[1.] asymptotic value of $T_I(k)/k$:
    \[
    \langle t_I \rangle:=-\sum_{i\in I}\pi(i) \tr\left( \mathcal{L}_{0*}^{-1} (\ket{y_i}\bra{y_i})\right);
    \]

    \item[2.] average of 1-norm of $\LL_{0*}^2$ in stationarity:
    \[
    b_r^2 := \sum_{i\in I}\pi(i) \left \| \mathcal{L}_{0*}^{-2}(\ket{y_i}\bra{y_i}) \right \|_{1};
    \]
    \item[3.] superoperator norm of $\mathcal{L}^{-1}_0$:
    \[
    c_r:=\left \|\mathcal{L}_0^{-1}\right \|_{\infty\rightarrow\infty}=\|\LL_0^{-1}(\mathbf{1})\|_\infty.
    \]
\end{itemize}

The last equality in the expression of $c_r$ is due to Theorem \ref{th:RD} in Appendix \ref{app:markovProofQ} and makes the superoperator norm analytically computable.
%We now state the result.

\bigskip

\begin{theo}[Fluctuations of FPT for Total Counts in Reset Processes]\label{theo:dynActReset}
Assume that Hypothesis \ref{hypo:irrQuantL} holds ($\mathcal{L}$ be irreducible) the jump operators are of the form \eqref{eq.reset} (reset process). Let $\varepsilon$ be the spectral gap of $\P^\dag \P$. For every $\gamma>0$:
\begin{equation*}
\begin{split}
&\PP_\nu \left (\frac{T_I(k)}{k} \geq \langle t_I \rangle + \gamma \right) \leq C(\nu) \exp \left ( -k \frac{\gamma^2 \varepsilon}{4b_r^2}h\left ( \frac{5c_r\gamma}{2 b_r^2}\right )\right )\\
{\rm and}\\
&\PP_\nu \left (\frac{T_I(k)}{k} \leq \langle t_I \rangle - \gamma \right) \leq C(\nu) \exp \left ( -k \frac{\gamma^2 \varepsilon}{4b_r^2}h\left ( \frac{5c_r\gamma}{2 b_r^2}\right )\right ),\quad k\in \mathbb{N},
\end{split}
\end{equation*}
where $h(x):=(\sqrt{1+x}+\frac{x}{2}+1)^{-1}$ and $C(\nu):=(\sum_{i \in I} \nu(i)^2/\pi(i))^{\frac{1}{2}}$. Here, $\PP_\nu$ is the probability measure induced by the initial state given by
$$\sum_{i \in I}\nu(i)\ket{y_i}\bra{y_i}, \quad \sum_{i \in I}\nu(i) =1, \quad \nu(i)\geq 0.
$$
\end{theo}

\begin{coro}\label{dynActiTURreset}
The variance of the first passage time for total counts is bounded from above by:
\[
\frac{{\rm var}_\pi(T_I(k))}{k}\leq \left (1 + \frac{2}{\varepsilon} \right )b_r^2.
\]
\end{coro}

The proof of Theorem \ref{theo:dynActReset} can be found in Appendix \ref{app:proof_theo:dynActReset}. Note that a classical Markov chain can be embedded into a quantum Markov process, by setting $H=0$, $I =\EE$, and the jump operators 
%$L_{xy}=\sqrt{w_{xy}}\ket{y}%\bra{x}$.
$L_{ij} = \sqrt{w_{ij}}\ket{x_j}\bra{x_i}$ 
for an orthonormal basis 
$\{\ket{x_j}\}_{j=1}^d$.

%%%%%%%%%%%%%%%%%%%%%%%%%%%%%%%%%%
\subsubsection{Tail Bound for General Counting Observables}\label{sec:countObsQ}

%%%%%%%%%%%%%%%%%%%%%%%%%%%%%%%
Our final result provides the quantum analogue to the bound of Theorem \ref{theo:countObs}. We consider the FPT 
$T_\A (k)$ for the observable $N_\A(t)$ which counts the number of jumps with label in the subset $\A\subseteq I$, cf.~\eqref{eq:FPT.quantum}. The next result gives an upper bound to the tails of the FPT distribution under the assumption that  Hypothesis \ref{hypo:irrQuantL} holds.

Recall that we introduced the generator decomposition 
\[
\mathcal{L} = \J_\A + \mathcal{L}_{\infty},
\]
where $\J_\A(x) = \sum_{i\in \A} L_i^* x L_i$.
We denote 
\[
\beta := \left \| \LL_{\infty*}^{-1}\right\|_{1\rightarrow 1}.
\]
%Using this quantity, we can state the result.
\begin{theo}[Rare Fluctuations of General Quantum Counting Observable FPTs]\label{theo:countObsQ}
    Assume that Hypothesis \ref{hypo:irrQuantL} holds ($\mathcal{L}$ be irreducible), and let $\A\subseteq I$ be nonempty. For every $\gamma > \beta - \langle t_\A \rangle$:
    \[
    \PP_\nu \left (\frac{T_{\A}(k)}{k} \geq \langle t_{\A} \rangle + \gamma \right) \leq \exp \left ( -k \left(\frac{\gamma+\langle t_{\A} \rangle-\beta}{\beta}-\log\left(\frac{\gamma+\langle t_{\A} \rangle}{\beta}\right)\right)\right ), k\in \mathbb{N}.
    \]
    \end{theo}
 The proof of Theorem \ref{theo:countObsQ} can be found in Appendix \ref{proof_theo:countObsQ}. As mentioned above, and in contrast to the classical case, in general there is no explicit expression for the $1\rightarrow 1$ norm of a superoperator, but  thanks to Theorem \ref{th:RD}, we know that
 $$\beta=\|\LL_\infty^{-1}\|_{\infty \rightarrow \infty}=\|\LL_\infty^{-1}(\mathbf{1})\|_{\infty}.$$
 Despite being in general computational demanding, at least there exists an explicit formula for the new expression for $\beta$. We can also derive an upper bound on the variance, in terms of $\beta$, stated below.

\begin{coro} \label{coro:cgiTURq}
Given any non-empty set of jumps $\A$, the variance of the corresponding first passage time at stationarity is bounded from above by:
\[
\frac{{\rm var}_\varsigma(T_{\A}(k))}{k}\leq \left ( 1+ \frac{2}{\tilde{\varepsilon}}\right )\beta^2,
\] 
where
$$\tilde{\varepsilon}:=1-\max\{\|\Psi(x)\|_{\infty\to\infty}: \|x\|_{\infty\to\infty}=1, \,\tr(\varsigma x) = 0\}.$$
\end{coro}
The proof of Corollary \ref{coro:cgiTURq} can be found in the Appendix \ref{proof_theo:countObsQ}.

%%%%%%%%%%%%%%%%%%%%%%%%%%%%%%%%%%
\subsection{Examples: Quantum Concentration Bounds}
%%%%%%%%%%%%%%%%%%%%%%%%%%%%%%%%%

In this section we illustrate the quantum concentration results obtained in the second part of the paper with a few simple examples.

%%%%%%%%%%%%%%%%%%%%%%%%%%%%%%%%%%%
\subsubsection{Three-Level Emitter with Dephasing Channel}\label{ex:qAct}
%%%%%%%%%%%%%%%%%%%%%%%%%%%%%%%%%%%%
 We consider a three-level system, with one dissipative jump and one dephasing channel, as sketched in Fig.~\ref{fig:3levelQAct} (a). The system has Hamiltonian 
 $$
 H = \Omega_{01}(|0\rangle\langle 1| + |1\rangle\langle 0|) + \Omega_{12}(|1\rangle\langle 2| + |2\rangle\langle 1|)
 $$ 
 and jump operators 
 $$
 L_1 = \omega_{12}\ket{1}\bra{2}, \qquad 
 L_2 = \omega_{02}(|0\rangle\langle 0|- |2\rangle\langle 2|)
 $$ 
 We count the total number of jumps of both the emission channel and the dephasing channel and compare the lower bounds on the large deviation rate function obtained from Theorems \ref{theo:dynActQ} and \ref{theo:countObsQ} with the exact rate function, see Fig.~\ref{fig:3levelQAct}(b). The exact rate function (full/black) is bounded in the entire region by Theorem \ref{theo:dynActQ} (dashed/blue). Theorem \ref{theo:countObsQ} for general counting observable allows one to bound (dotted/red) 
 the right tail of the rate function: as in the classical case, cf.\ Fig.~\ref{rf3level}, this tail bound is tighter than the activity bound for large enough deviations.

\begin{figure}[t]
    \begin{center}
    \includegraphics[width=1\linewidth]{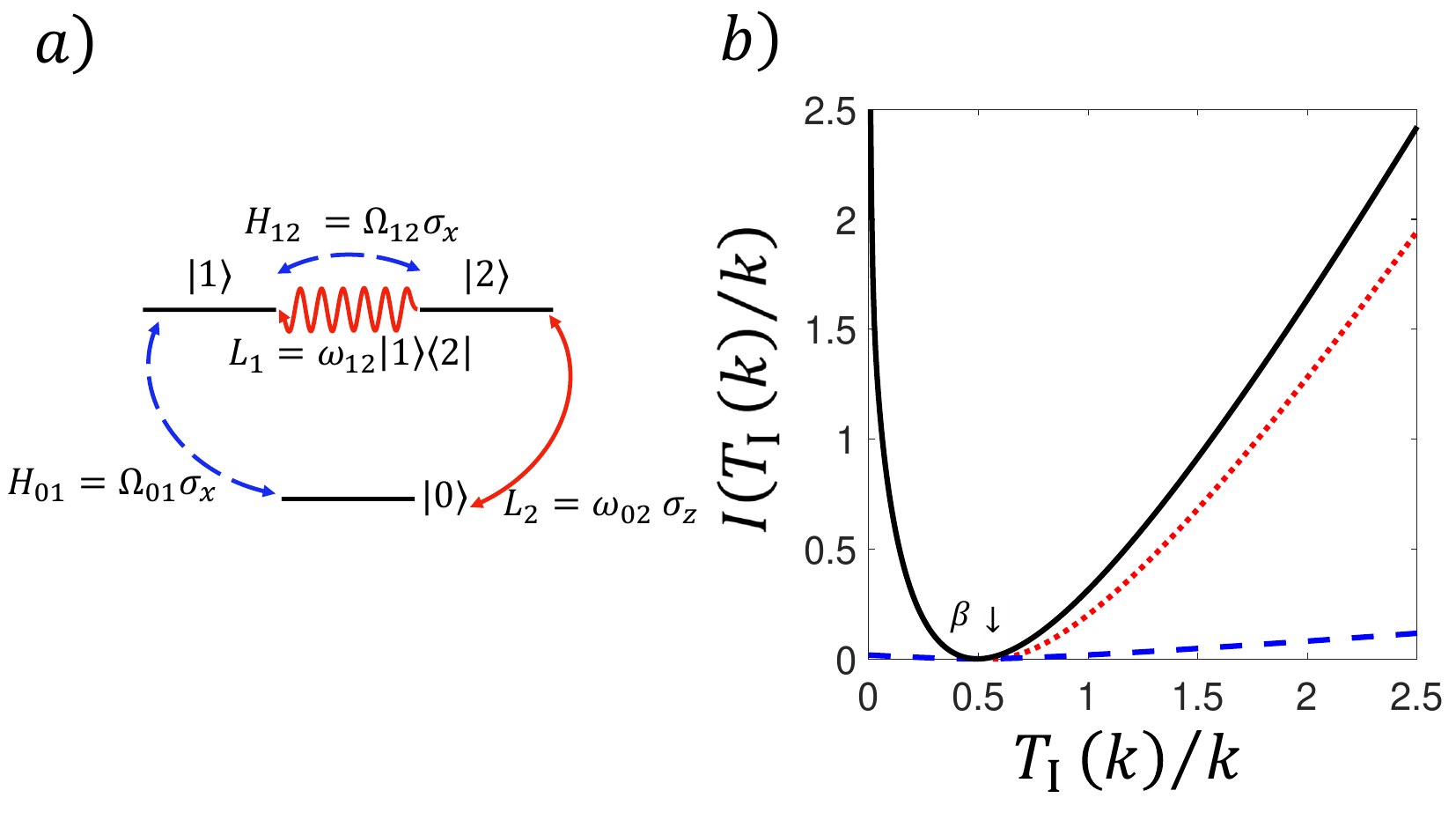}
    \end{center}
    \caption{{\bf Bounds on the rate function of the FPT of the total number of emissions for a quantum three-level system.}  
    (a) Sketch of quantum three-level system. The Hamiltonian (dashed/blue) drives the evolution coherently while the jump operators (solid/red) give rise to dissipative transitions. 
    (b) Exact rate function $I(T_{I}(k)/k)$ (full/black) of the FPT for the total number of quantum jumps, for the case $\Omega_{01}=10$, $\Omega_{12} = 1$, $\omega_{12} = \Omega_{01}$, $\omega_{02}=\frac{1}{5}\Omega_{01}$. Theorem \ref{theo:dynActQ} gives a lower bound on the entire rate function (dashed/blue). Theorem \ref{theo:countObsQ} bounds the tail (dotted/red) in the region $T_{I}(k)/k>\beta$.
    }
    \label{fig:3levelQAct}
\end{figure}

\begin{figure}[t]
    \begin{center}
    \includegraphics[width=1\linewidth]{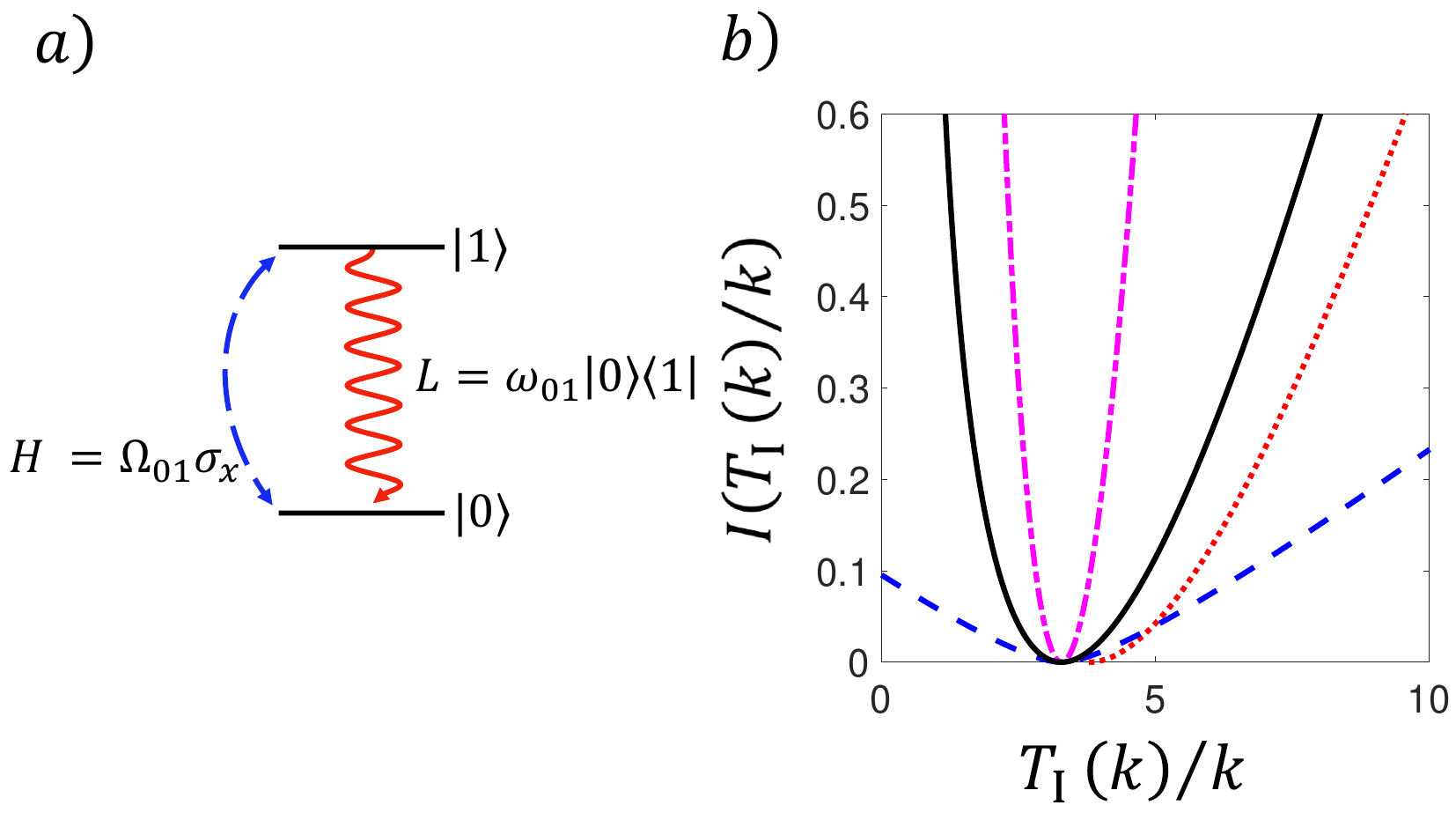}
    \end{center}
    \caption{{\bf Bounds on the rate function of the FPT of the total number of emissions for a two-level emitter.}  
    (a) Sketch of two-level emitter. The Hamiltonian (dashed/blue) drives the evolution coherently while the jump operators (solid/red) give rise to dissipative transitions. 
    (b) Exact rate function $I(T_{I}(k)/k)$ (full/black) of the FPT for the total number of quantum jumps, for the case $\Omega_{01} = 1$, $\omega_{01} = 0.8\Omega_{01}$. As this is a quantum reset process, 
    Theorem \ref{theo:dynActReset} gives a lower bound on the entire rate function (dashed/blue). Theorem \ref{theo:countObsQ} bounds the tail (dotted/red) in the region $T_{I}(k)/k>\beta$. The result from \cite{carollo2019unraveling} gives an upper bound on the rate function (dash-dotted/magenta).
    }
    \label{fig:reset}
\end{figure}

%%%%%%%%%%%%%%%%%%%%%%%%%%%%%%%%
\subsubsection{Two Level Emitter}\label{ex:reset}
%%%%%%%%%%%%%%%%%%%%%%%%%%%%%%%%%%

We illustrate the results of Theorem \ref{theo:dynActReset} by considering a two level emitter with driving Hamiltonian $H=\Omega_{01}(|0\rangle\langle 1| + |0\rangle\langle 1|)$ and jump operator representing the emitted photon $L=\omega_{01}\ket{0}\bra{1}$, see Fig.~\ref{fig:reset}(a). Since $L$ is a rank-one operator, the system jumps to the same state $|0\rangle$ every time there a count. Therefore, the counts process is a renewal process with waiting distribution computed using equation \eqref{pdfReset}. 
In Fig.~\ref{fig:reset}(b) we plot the exact rate function (full/black) and two lower bounds, obtained from our reset process bound of Theorem \ref{theo:dynActReset} (dashed/blue) and the counting observable bound Theorem \ref{theo:countObsQ} (dotted/red). For comparison to known literature we plot the upper bound on the rate function (dot-dashed/magenta) of reset processes obtained via large deviations \cite{carollo2019unraveling}. As in the classical example of Sec.~\ref{ex:3LvlAct}, the bound of Theorem \ref{theo:countObsQ} outperforms that of Theorem \ref{theo:dynActReset} for larger deviations.

\begin{figure}[t]
    \begin{center}
    \includegraphics[width=0.5\linewidth]{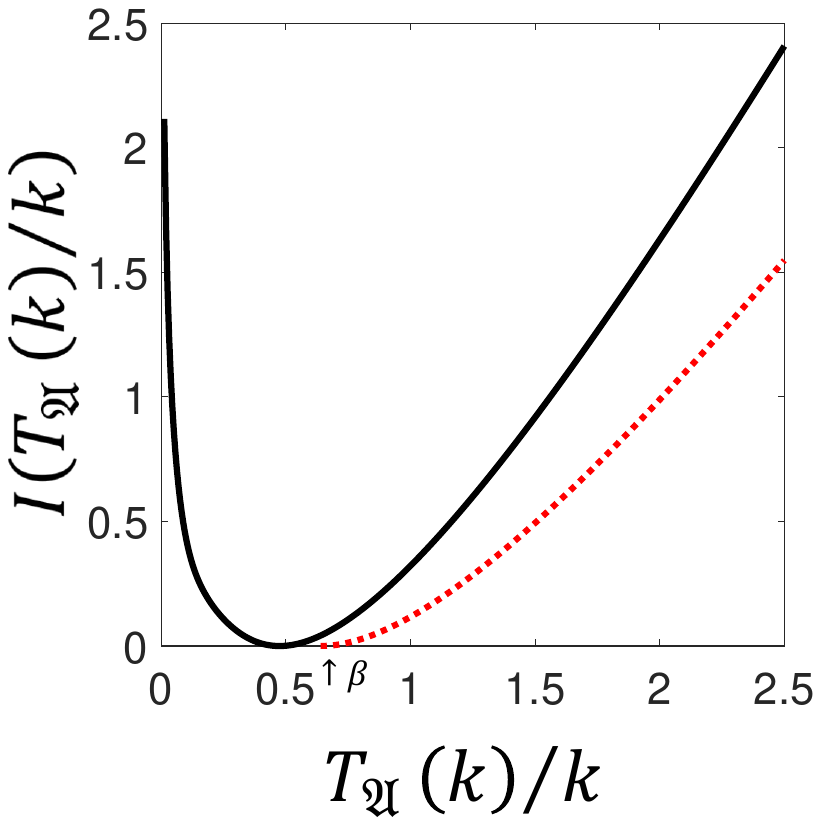}
    \end{center}
    \caption{{\bf Lower bound on the FPT rate function for a counting observable of a quantum three-level system.}
    Exact rate function $I(T_{\A}(k)/k)$ (full/black) of the FPT only counting the dephasing jumps (jump operator by $L_2$), for 
    the same model of Fig.~\ref{fig:3levelQAct}(a). 
    Theorem \ref{theo:countObsQ} gives a lower bound on right deviations (dotted/red) in the region $T_{\A}(k)/k>\beta$.
    }
    \label{fig:countingQ}
\end{figure}

%%%%%%%%%%%%%%%%%%%%%%%%%%%%%%%%%%
\subsubsection{Three Level Emitter Counting Dephasing Jumps}\label{ex:count}
%%%%%%%%%%%%%%%%%%%%%%%%%%%%%%

For our final example we consider a system in which we are only interested in a subset of jumps. We use the same setup as in Sec.~\ref{ex:qAct} but this time we only count the number of dephasing jumps (jump operator $L_2$). In Fig.~\ref{fig:countingQ} we show the exact rate function (full/black) and a lower bound on its right tail from Theorem \ref{theo:countObsQ} (dotted/red). 

\section{Conclusions}
\label{sec:Conc}

When studying stochastic dynamics one often considers a description in terms of stochastic trajectories of fixed overall time and where (time-integrated) observables fluctuate. But there is an alternative description of the same dynamics in terms of trajectories of fluctuating overall time but where one or more observables have a fixed value. It is of interest therefore to formulate general results about dynamics in these two alternative descriptions. This is what we have done in this paper for concentration bounds, by complementing the concentration inequalities for time-integrated quantities of Refs.~\cite{bakewell-smith2023general,girotti2023concentration} by analogous concentration bounds for first passage times in both systems with classical or quantum Markov dynamics. 

The study of FPTs is more involved that that of time-integrated observables, and for that reason we were only able to derive bounds for FTPs for the subset of all trajectory functions known as counting observables (which include fundamental quantities such as the dynamical activity). The concentration inequalities that we find are upper bounds on the probabilities to observe fluctuations in FTPs, and are valid for all values of the observable threshold that defines the FTP and not only in the large threshold limit where large deviation theory applies. The bounds are written in terms of relatively simple quantities which describe the overall properties of the dynamics (and which in an ideal setting can be determined by observation), in particular the longest expected waiting time between events, and the spectral gap of the symmetrised generator (these and similar spectral quantities have been shown to be relevant in other recent works such as \cite{mori2022symmetrized,dechant2023thermodynamic}). Our upper bounds on fluctuations complement the lower bounds from so-called thermodynamic uncertainty relations, thus together providing general two-sided constraints to the likelihood of fluctuations. 

While our results should have wide applicability in the theory of Markov processes and non-equilibrium statistical mechanics, a possible area for further study is their application in constructing confidence intervals for parameter estimation (\cite[Section 5.5]{girotti2023concentration}). It is also possible to apply analogous perturbative techniques as those used here to systems with discrete time dynamics to derive similar results. Also, a better understanding of the transition operators in Eqs.~\eqref{eq:qexp} and \eqref{eq:psi} may allow the derivation of FPT bounds for arbitrary fluctuations of FPTs of generic counting observables, rather than just the tails of their distributions. A further extension is to FPTs of empirical currents, which are of of great interest in the study of non-equilibrium dynamics. Finally, it would be useful to bound the spectral gap in terms of further simpler quantities related to the physics of the process, which would provide more intuitive and and operationally accessible concentration bounds.

\bmhead{Acknowledgments}
This work was supported by the EPSRC grants EP/R04421X/1 and EP/V031201/1. F.G. has been partially supported by the MUR grant Dipartimento di Eccellenza 2023–2027 of Dipartimento di Matematica, Politecnico di Milano and by the INDAM GNAMPA project 2024 “Probabilità quantistica e applicazioni''.

\bigskip

\begin{appendices}
    \section{Upper Bounds on Fluctuations of FPTs for
Classical Markov Processes}
 
\subsection{Proof of Lemma \ref{lem:tech1}} \label{app:tech1}

\begin{lemmaA}
    The following statements hold true:
    \begin{enumerate}
        \item $\overline{\lambda}:=-\max\{\Re(z):z \in {\rm Sp}(\L_\infty)\}>0$, hence $\L_\infty$ is invertible;
        \item $r(\RR^{-1}\W_2)<1,$ therefore $\sum_{k \geq 0} \S^k$ is well defined with $\S = \RR^{-1}\W_2$ and one has
        \[
        -\frac{\I}{\L_\infty}=\sum_{k\geq 0}\left ( \frac{\I}{\RR}\W_2 \right )^k\frac{\I}{\RR},
        \]
        \item for every $u <\overline{\lambda}$, one has
        \[
        \E_\nu[e^{uT_{\A}(k)}]=\left \langle \nu, \left(\frac{\L_\infty}{u+\L_\infty}\Q \right)^k\Iv\right \rangle,
        \]
        \item $\left \|\L_\infty^{-1} \right \|^{-1}_{\infty \rightarrow \infty}\leq \overline{\lambda}$.
    \end{enumerate}
\end{lemmaA}

\begin{proof}

Note that $\L_\infty$ generates a sub-Markov semigroup $e^{t\L_\infty}$. By Perron-Frobenius theory we know that 
\[
r(e^{t\L_\infty})=e^{-t\overline{\lambda}}\Leftrightarrow \overline{\lambda}:=-\max\{\Re(z): \, z \in {\rm Sp}(\L_\infty)\}.
\]
By contradiction suppose $\overline{\lambda} = 0$. Then there exists a nonzero nonnegative function $f:E \rightarrow [0,+\infty)$ such that $\L_\infty f=0$. We then have
\[
\L f=\L_\infty f +\W_1  f=\W_1 f \geq 0.
\]
Therefore one has
\[
0=\langle\hat{\pi},\L f \rangle = \langle \hat{\pi},\W_1 f \rangle,
\]
which implies that $\L f= \W_1 f = 0$ because $\hat{\pi}$ has full support. Since $\L$ is irreducible, $f=\alpha\Iv$ for some nonnegative $\alpha$, however $\alpha \W_1\Iv=0$ implies that $\alpha$ and therefore $f$ are $0$ (it follows from the positivity of $\W_1$). We came to a contradiction, which proves that $\overline{\lambda}>0$. Therefore ${\rm Sp}(\L_\infty)\subset \{z \in \mathbb{C}: \, \Re(z) \leq -\overline{\lambda} <0\}$
and $\L_\infty$ is invertible.

\bigskip 2. The proof is similar to that of point 1. Notice that $\RR^{-1}\W_2$ is a sub-Markov transition kernel and let $r= r(\RR^{-1}\W_2)$ such that $r \in [0,1]$, and the corresponding nonzero nonnegative eigenvector $f:E \rightarrow [0,+\infty)$. Suppose that $r=1$, then we can write
\[
(\P-\I)f=\RR^{-1}\W_1 f + (\RR^{-1}\W_2 -\I)f=\RR^{-1}\W_1 f,
\]
therefore
\[
0=\langle \pi, (\P-\I)f \rangle =\langle \pi,\RR^{-1}\W_1 f \rangle,
\]
which implies that $\RR^{-1}\W_1 f= (\P-\I)f =0$ from which it follows that $f=0$ as before, hence a contradiction. What we have proved so far shows that the derivation of Eq. \eqref{eq:qexp} is correct.

\bigskip 3. Using Dyson expansion, one can see that
\[\mathbb{P}_\nu(T_\A(k) \leq t)=\int_{\sum_{i=1}^{k}t_i \leq t} \langle \nu,e^{t_1\L_\infty}\W_1 e^{(t_2)\L_\infty}\cdots \W_1e^{\left (t-\sum_{i=1}^{k}t_i \right )\L_\infty}\Iv \rangle dt_1\cdots dt_k.
\] 

Therefore, if $u < \overline{\lambda}$, then
\begin{equation}\label{eq:resLinf}
-\frac{\I}{u+\L_\infty}=\int_{0}^{+\infty}e^{t(u+\L_\infty)}dt
\end{equation}
and we get
\[\mathbb{E}_\nu[e^{uT_\A(k)}]=\left \langle \nu, \left (-\frac{\I}{u+\L_\infty} \W_1 \right )^k \Iv \right \rangle.
\]

\bigskip 4. The Spectral Mapping Theorem implies that ${\rm Sp}(\L_\infty^{-1})=\{z^{-1}:\, z \in {\rm Sp}(\L_\infty)\}$, therefore one has that
\[
 \left\| \frac{{\I}}{{\L_\infty}} \right \|_{\infty \rightarrow \infty} \geq r(\L_\infty^{-1})\geq \frac{1}{\overline{\lambda}} \quad\Leftrightarrow \quad \left\| \frac{{\I}}{{\L_\infty}} \right \|_{\infty \rightarrow \infty}^{-1} \leq \overline{\lambda}.
\]
\end{proof}

%%%%%%%%%%%%%%%%%%%%%%%%%
\subsection{Proof of Theorem \ref{theo:ldpC}}\label{app:LDPc}
%%%%%%%%%%%%%%%%%%%%%%%%%%%%%%%%%

\bigskip\begin{theoA} 
Let us consider any nonempty subset $\A$ of the set of possible jumps. The collection of corresponding FPTs $\{T_{\A}(k)/k\}$ satisfies a LDP with good rate function given by
$$
I_\A(t):=\sup_{u \in \mathbb{R}}\{ut-\log(r(u))\}$$
where
$$
r(u)=\begin{cases} r \left (\Q_u \right ) & \text{ if } u < \overline{\lambda}\\
+\infty & \text{o.w.}\end{cases}$$
where $\Q_u:=-(u+\L_\infty)^{-1}\W_1$ and $\overline{\lambda}:=-\max\{\Re(z):z \in {\rm Sp}(\L_\infty)\}.$
\end{theoA}

\begin{proof}
    The proof of Lemma \ref{lem:tech1} shows that if $u<\overline{\lambda}:=-\max\{\Re(z):z \in {\rm Sp}(\L_\infty)\}$, then
    \begin{equation}
    \label{eq:MGFTk}
    \mathbb{E}_\nu[e^{uT_{\A}(k)}]= \left \langle \nu, \Q_u^{k} \Iv \ \right \rangle,
    \end{equation}
    where $\Q_u:=-(u+\L_\infty)^{-1}\W_1$. From the expression
    $$\Q_u=\int_0^{+\infty}e^{(u+\L_\infty)t}\W_1 dt$$
    one sees that $\Q_u$ is a positivity preserving map for every $u<\overline{\lambda}$. From Perron-Frobenius theory (see Theorem \ref{th:PF} in Appendix \ref{app:technical}), we know that $r(u):=r(\Q_u)$ is an eigenvalue of $\Q_u$ that admits a positive eigenvector $x(u)$. With simple algebraic manipulations one can see that
    \[
    \Q_u x(u)=r(u)x(u) \Leftrightarrow \L_{s(u)}x(u)=-u x(u)
    \]
    where $\L_{s(u)}:=\L+(e^{s(u)}-1)\W_1$ and $s(u)=-\log(r(u))$.
    
    The perturbations of $\L$ given by $\L_{s}$ for $s \in \mathbb{R}$ are irreducible (see Lemma \ref{lem:irredC} in Appendix \ref{app:technical}), hence they admit a unique positive eigenvector, which is actually strictly positive and corresponds to the the eigenvalue given by $\max\{\Re(z):z \in {\rm Sp}(\L_s)\}$. Therefore, $-u=\max\{\Re(z):z \in {\rm Sp}(\L_{s(u)})\}$ and $x(u)>0$ is the unique eigenvector for $\Q_u$ corresponding to $r(u)$. One can also show that $r(u)$ is in fact algebraically simple for $\Q_u$ as in the proof of Lemma 5.3 in \cite{carbone2015homogeneous}.

    Summing up, one has that for $u<\overline{\lambda}$ the function $u \mapsto r(u)$ is smooth (actually analytic in a complex neighborhood of the values we are considering) and
    \begin{equation} \label{eq:lim}
    \lim_{k \rightarrow +\infty}\frac{1}{k}\log(\mathbb{E}_\nu[e^{uT_\A(k)}])=\log(r(u)).
    \end{equation}
    Indeed, 
    $$
    \frac{1}{k}\log(\mathbb{E}_\nu[e^{uT_\A(k)}]) \leq \frac{1}{k}\log(\|\Q_u^k\|_{\infty \rightarrow \infty}) \rightarrow_{k\rightarrow +\infty} \log(r(u)) $$
    thanks to Gelfand's formula. On the other hand, using \eqref{eq:MGFTk} and  assuming that $\|x(u)\|_\infty \leq 1$, one has
    $$\frac{1}{k}\log(\mathbb{E}_\nu[e^{uT_\A(k)}]) \geq \log(r(u))+\frac{1}{k}\log(\langle \nu,x(u) \rangle)\rightarrow_{k\rightarrow +\infty} \log(r(u)),
    $$
    since $x(u)>0$ and $\langle \nu, x(u) \rangle >0$.

    In order to apply G\"artner-Ellis Theorem (\cite[Theorem 2.3.6]{dembo2010large}), we only need to show that
    \begin{equation} \label{eq:explosion}
    \lim_{u \uparrow \overline{\lambda}} \log(r(u))=\lim_{u \uparrow \overline{\lambda}} \log(r(u))^\prime=+\infty.\end{equation}
    
    Notice that $r(u)$ and $\log(r(u))^\prime=r^\prime(u)/r(u)$ are both monotone non-decreasing (they are limits of monotone non-decreasing functions cf. Eq. \eqref{eq:lim}): the limits in Eq. \eqref{eq:explosion} exist and we only need to show that they cannot be finite.

    Let $\mathbf{T}$ be the spectral projection of $\L_\infty$ corresponding to $-\overline{\lambda}$; we remark that $\L_\infty$ restricted to the range of $\mathbf{T}$ is diagonalisable. To show this, let us assume that this is not the case. If $-\overline{\lambda}$ is an eigenvalue of $\L_\infty$, then $0$ is an eigenvalue of $\L^\prime:=\L_\infty+\overline{\lambda}\I$. The matrix $\L'$ restricted to the range $\mathbf{T}$ is also not diagonalisable. %corresponding to eigenvalue 0. 
    However, this means that the restriction contains a Jordan block, in which case the norm of $e^{t\L^\prime}$ explodes for large $t$ which contradicts the fact that $e^{t\L^\prime}$ generates a contraction semigroup by Lumer-Phillips Theorem (see, for instance, Theorem 3.15 and the following corollaries in \cite{engel2000one-parameter}).

    Let us first show that
    $$\lim_{u \uparrow \overline{\lambda}}r(u)=+\infty.$$
    Notice that $\mathbf{T}\Q \neq 0$, since $\mathbf{T}\Q(\Iv)=\mathbf{T}\Iv \neq 0$. Therefore
    $$
    \Q_u=\frac{\L_\infty}{\L_\infty+u}\Q=\frac{\overline{\lambda}}{\overline{\lambda}-u}\mathbf{T}\Q+ (\I-\mathbf{T})\frac{\L_\infty}{u+\L_\infty}\Q $$
    has a norm that explodes for $u \uparrow \overline{\lambda}$. By contradiction, let us assume that for $u \uparrow \overline{\lambda}$, $r(u)\uparrow r(\overline{\lambda})<+\infty$. This implies that we can choose $x(u)$ such that it converges to the unique strictly positive Perron-Frobenius eigenvector of $\L_{s(\overline{\lambda})}$ and that $\min{\rm Sp}(x(u)) \not\rightarrow0.$ Therefore we have for every $0 \leq u<\overline{\lambda}$
    $$
    \|\Q_u\|_{\infty \rightarrow \infty}=\|\Q_u(\Iv)\|_\infty \leq \frac{1}{\min {\rm Sp}(x(u))}\|\Q_u(x(u))\|_\infty = r(u) \frac{\|x(u)\|_\infty}{\min {\rm Sp}(x(u))}$$.
   Since the right side remains finite as $u\uparrow \bar{\lambda}$, this contradicts the fact that 
   $\|{\bf Q}_u\|_{\infty\to \infty}$ diverges. 
    Notice that the first equality in the previous Eq. is due to Theorem \ref{th:RD} in Appendix \ref{app:technical}.

    Let $l(u)$ be the left eigenvector of $\Q_u$; we can assume that $\langle l(u), x(u) \rangle \equiv 1$, therefore one has
    \[
    \begin{split}
    r(u)&=\langle l(u),-(u+\L_\infty)^{-1}\W_1 x(u) \rangle\\
    &=\underbrace{\langle l(u),-(u+\L_\infty)^{-1}\mathbf{T}\W_1 x(u) \rangle}_{(I)}+\underbrace{\langle l(u),-(u+\L_\infty)^{-1}(\mathbf{1}-\mathbf{T})\W_1 x(u) \rangle}_{(II)}.\\
    \end{split}
    \]
    Since $(II)$ stays bounded, for $u \uparrow \overline{\lambda}$ one has
    $$r(u) \asymp \frac{\overline{\lambda}}{\overline{\lambda}-u} \langle l(u),\mathbf{T}\Q x(u) \rangle $$
    with both sides diverging as $u\uparrow \bar{\lambda}$.
    Differentiating the previous expression for $r(u)$ and dividing for $r(u)$ one gets
    \[\begin{split}
    \frac{r^\prime(u)}{r(u)}&=\frac{\langle l(u),(u+\L_\infty)^{-2}\W_1 x(u) \rangle}{r(u)}\\
    &=\underbrace{\frac{\overline{\lambda}}{(\overline{\lambda}-u)^2}\frac{\langle l(u),\mathbf{T}\Q x(u) \rangle}{r(u)}}_{(I)}+ \underbrace{\frac{\langle l(u),(u+\L_\infty)^{-2}(\I-\mathbf{T})\W_1 x(u) \rangle}{r(u)}}_{(II)}.
    \end{split}\]
    When $u \uparrow \overline{\lambda}$, $(II)\rightarrow 0$, while
    $(I) \asymp (\overline{\lambda}-u)^{-1}$ and we are done.
\end{proof}
\subsection{Proof of Theorem \ref{theo:dynActC} and Corollary \ref{dynActiTURC}} \label{app:dynActC}

\begin{theoA}[Fluctuations of FPT for Activity]
Suppose Hypothesis \ref{hypo:irrC} holds ($\L$ is irreducible) and let $\varepsilon$ be the spectral gap of $\P^\dag \P$. For every $\gamma > 0$ the following holds true:
\begin{equation*}
\begin{split}
&\PP_\nu \left (\frac{T_{\EE}(k)}{k}\geq \langle t_{\EE} \rangle + \gamma \right) \leq C(\nu) \exp \left ( -k \frac{\gamma^2 \varepsilon}{4b_c^2}h\left ( \frac{5c_c\gamma}{2 b_c^2}\right )\right )\\
{\rm and}\\
&\PP_\nu \left (\frac{T_{\EE}(k)}{k} \leq \langle t_{\EE} \rangle - \gamma \right) \leq C(\nu) \exp \left ( -k \frac{\gamma^2 \varepsilon}{4b_c^2}h\left ( \frac{5c_c\gamma}{2 b_c^2}\right )\right ), \quad k\in \mathbb{N},
\end{split}
\end{equation*}
where $h(x):=(\sqrt{1+x}+\frac{x}{2}+1)^{-1}$ and $C(\nu):=\max_{x\in E} \left\{\nu(x)/\pi(x)\right\}$.
\end{theoA}

\begin{proof}

We begin by using the Chernoff bound to upper bound the probability of $T_{\EE}(k)/k$ right deviating from $\langle t_{\EE} \rangle$ by more than $\gamma>0$, using the moment generating function:
\begin{equation}\label{chernBoundC}
\PP_\nu\left(\frac{T_{\EE}(k)}{k}\geq\langle t_{\EE} \rangle+\gamma\right)\leq e^{-ku(\langle t_{\EE} \rangle+\gamma)}\E_\nu[e^{uT_{\EE}(k)}], \quad u\geq0.
\end{equation}
We now focus on upper bounding the moment generating function using the definition in Eq. \eqref{eq:MGFcda}. Introducing the notation
$$\F_u:=\frac{\RR}{\RR-u}, \quad u <R_{\min},$$
one has that for every $0\leq u <R_{\min}$ the following holds true:
\begin{equation*}\label{mgfChySwz}
\begin{split}
\E_\nu[e^{uT_{\EE}(k)}]&= \left\langle \nu,(\F_u\P)^k\Iv\right\rangle=\left\langle \frac{\nu}{\pi},(\F_u\P)^k \Iv\right\rangle_\pi\\
&=\left\langle \F_u^{\frac{1}{2}}\frac{\nu}{\pi},(\F_u^{\frac{1}{2}}\P \F_u^{\frac{1}{2}})^{k-1}\F_u^{\frac{1}{2}}\Iv\right\rangle_\pi\\
&\leq\left\|\F_u^{\frac{1}{2}}\frac{\nu}{\pi}\right\|_\pi\left\|\F_u^{\frac{1}{2}}\P \F_u^{\frac{1}{2}}\right\|_\pi^{k-1}\left\|\F_u^{\frac{1}{2}}\Iv\right\|_\pi,
\end{split}
\end{equation*}
 where $\frac{\nu}{\pi}(x)=\frac{\nu(x)}{\pi(x)}, \forall x\in E$ and, with a little abuse of notation, we denote $\|\cdot \|_{\pi \rightarrow \pi}$ instead by $\| \cdot\|_\pi$.
We use the notation $\mathbf{M}_{\frac{\nu}{\pi}}$ to denote the multiplication operator corresponding to $\frac{\nu}{\pi}$. We can write the following:
\[
\left\|\F_u^{\frac{1}{2}}\frac{\nu}{\pi}\right\|_\pi = \left\|\F_u^{\frac{1}{2}}\mathbf{M}_{\frac{\nu}{\pi}}\Iv\right\|_\pi
= \left\|\mathbf{M}_{\frac{\nu}{\pi}}\F_u^{\frac{1}{2}}\Iv\right\|_\pi\leq \left\|\mathbf{M}_{\frac{\nu}{\pi}}\right\|_\pi\left\|\F_u^{\frac{1}{2}}\Iv\right\|_\pi=\left\|\frac{\nu}{\pi}\right\|_\infty\left\|\F_u^{\frac{1}{2}}\Iv\right\|_\pi.
\]
Note that $\|\mathbf{M}_{\frac{\nu}{\pi}}\|_\pi = \|\frac{\nu}{\pi}\|_\infty$ since we know that $\mathbf{M}_{\frac{\nu}{\pi}}$ is a diagonal matrix. Applying Lemma \ref{multOpLem} with $f=\F_u^{\frac{1}{2}}\Iv$, Lemma \ref{lpLem} with $\mathbf{A}=\mathbf{B}=\F_u^{\frac{1}{2}}$ and remembering that $\F_u^{\frac{1}{2}}$ is self adjoint one can derive the following inequalities:
\begin{enumerate}
\item $\left\|\F_u^{\frac{1}{2}}\Iv\right\|_\pi\leq\left\|\F_u^{\frac{1}{2}}\hat{\P}\F_u^{\frac{1}{2}}\right\|_\pi^{\frac{1}{2}}$,
\item $\left\|\F_u^{\frac{1}{2}}\P\F_u^{\frac{1}{2}}\right\|_\pi\leq\left\|\F_u^{\frac{1}{2}}\hat{\P}\F_u^{\frac{1}{2}}\right\|_\pi$,
\end{enumerate}
where we recall that $\hat{\P}$ is the Le\'{o}n-Perron matrix associated to $\P$, cf. equation \eqref{eq.LeonPerron}. Therefore, we get:
\begin{equation}\label{expectationUB}
\mathbb{E}_\nu[e^{uT_{\EE}(k)}] \leq \left\|\frac{\nu}{\pi}\right\|_\infty\left\|\F_u^{\frac{1}{2}}\hat{\P}\F_u^{\frac{1}{2}}\right\|_\pi^k, \quad 0\leq u <R_{\min}.
\end{equation}

If we set $C(\nu): = \left \|\frac{\nu}{\pi}\right \|_\infty$, the problem is now reduced to finding an upper bound on $\left\|\F_u^{\frac{1}{2}}\hat{\P}\F_u^{\frac{1}{2}}\right\|_\pi$. Notice that $\F_u^{\frac{1}{2}}\hat{\P}\F_u^{\frac{1}{2}}$ is self adjoint, therefore its norm coincide with its spectral radius $r(u)$. Moreover, $\F_u^{\frac{1}{2}}\hat{\P}\F_u^{\frac{1}{2}}$ is similar to $\hat{\P}\F_u$, hence they share the same spectral radius; finally, Perron-Frobenius theory ensures that $r(u)$ is an eigenvalue of $\P(u):=\hat{\P}\F_u$. 

We can write $\P(u)$ as a power series:
\begin{equation}
\label{eq:series.P(u)}
\P(u) = \hat{\P} + \sum_{l=1}^{\infty}u^l\hat{\P}\left(\frac{\I}{\RR}\right)^l, \quad 0\leq u <R_{\min}.
\end{equation}
For conciseness of notation, we denote $\D:=\RR^{-1}$. Perturbation theory (see for instance \cite{kato1976perturbation}) implies that, if we can bound $\|\hat{\P}\D^l\|_\pi\leq \delta\zeta^{l-1}$ for some $\delta,\zeta>0$ and $l>1$, within the range $|u|<(2\delta\varepsilon^{-1}+\zeta)^{-1}$ - with $\varepsilon$ the spectral gap of $\hat{\P}$ (which is equal to the absolute spectral gap of $\P$). The spectral radius $r(u)$ can be expressed in the following way:
\begin{equation}\label{ruExp}
r(u) = 1 + \sum_{l=1}^{\infty}u^l r^{(l)},
\end{equation}
where:
\begin{equation}\label{eq:ru_exp}
r^{(l)}=\sum_{p=1}^l \frac{(-1)^p}{p} \sum_{\substack{\nu_1+\dots +\nu_p=l,\, \nu_i\geq 1\\
\mu_1+\dots +\mu_p=p-1, \, \mu_j\geq 0}} \tr\left(\hat{\P}\D^{\nu_1}\S^{(\mu_1)}\cdots \hat{\P}\D^{\nu_p}\S^{(\mu_p)}\right),
\end{equation}
with $\S^{(0)}=-\Pi$, $\S^{(1)}=(\hat{\P}-\I+\Pi)^{-1}-\Pi=-\varepsilon^{-1}(\I-\Pi)$ and $\S^{(\mu)}$ the $\mu^{\rm th}$ power of $\S^{(1)}$. Note that 
$\S^{(1)}$ is equal to the Moore-Penrose inverse of $-({\bf 1}-\hat{\bf P})$ and $\|\S^{(\mu)}\|_\pi=\varepsilon^{-\mu}$ for $\mu\geq 1$.
If we set our estimates $\delta=\zeta=c_c:=R_{\min}^{-1}$, we can indeed bound $\|\hat{\P}\D^l\|_\pi$ by:\[
\|\hat{\P}\D^l\|_\pi\leq c_c^l,
\]
which gives a radius of convergence $|u|<\frac{\varepsilon}{c_c(2+\varepsilon)}<\frac{1}{c_c}=R_{\min}$. Using equation \eqref{eq:ru_exp}, we can explicitly determine 
%$r^{(0)}=1$, 
$r^{(1)}=\langle t_{\EE} \rangle$ and 
$r^{(2)}=b_c^2 - \left\langle \D\Iv,
%(\I-\hat{\P})^{-1}
{\bf S}^{(1)}
\hat{\P} \D\Iv\right\rangle_\pi$. We then seek to bound $r^{(l)}$ for $l\geq 3$.
For $p=1$:
\[
-\tr(\hat{\P}\D^l(-\Pi))=\left\langle \D^l\Iv,\Iv\right\rangle_\pi=\sum_{x\in E}\pi(x)\frac{1}{R_x^l}.
\]
For the $p\geq 2$ cases, we get:\[
\begin{split}
-\tr(\hat{\P}\D^{\nu_1}\S^{(\mu_1)}\cdots \hat{\P}\D^{\nu_p}\S^{(\mu_p)}) &= \left\langle \D\Iv,\D^{\nu_1-1}\S^{(\mu_1)}\hat{\P}\D^{\nu_2}\S^{(\mu_2)}\cdots \hat{\P}\D^{\nu_{p-1}}\S^{(\mu_{p-1})}\hat{\P}\D^{\nu_{p}-1}\D\Iv\right\rangle_\pi\\
&\leq\|\D\Iv\|_\pi^2 \|\D\|_\pi^{l-2}\|\S^{(1)}\|_\pi^{p-1},
\end{split}
\]
where we have taken $\mu_p=0$, which is justified since $\mu_1+\cdots+\mu_p=p-1$, there is at least one $\mu_j=0$, and the trace is cyclic. Again using Cauchy-Schwarz we can bound the terms as follows: $\|\S^{(1)}\|_\pi= \frac{1}{\varepsilon}$, $\|\D\|_\pi = c_c$, $\|\D\Iv\|_\pi^2= b_c^2$. We also have that $\varepsilon\leq1$. For $p\geq 2$, each term in the inner sum of equation \eqref{eq:ru_exp} is then bounded by:
\[
b_c^2 \frac{c_c^{l-2}}{\varepsilon^{l-1}}.
\]
From \cite{lezaud1998chernoff-type} the number of terms $N(l)$ in equation \eqref{eq:ru_exp} is bounded by:\begin{equation}\label{numTerms}
N(l)=\sum_{p=1}^n\binom{l-1}{p-1}\binom{2(p-1)}{p-1}\frac{1}{p}\leq 5^{l-2},
\end{equation}
which is valid for $l\geq3$. Combining everything together, the bound on each $r^{(l)}$ becomes:
\[
|r^{(l)}| \leq \left\langle \D^l\Iv,\Iv\right\rangle_\pi + 5^{l-2}b_c^2 \frac{c_c^{l-2}}{\varepsilon^{l-1}}=\left\langle \D^l\Iv,\Iv\right\rangle_\pi+\frac{b_c^2}{5c_c}\left(\frac{5c_c}{\varepsilon}\right)^{l-1}.
\]
Which, through a simple computation is in fact valid for $l=2$ as well. Therefore, the eigenvalue $r(u)$ can be bounded above by:
\[
\begin{split}
r(u)&\leq 1 +\langle t_{\EE} \rangle u+ \sum_{l=2}^\infty \left\langle \D^l\Iv,\Iv\right\rangle_\pi u^l + \frac{b_c^2u}{5c_c}\left(\frac{5c_cu}{\varepsilon}\right)^{l-1}\\
&\leq \exp\left( \langle t_{\EE} \rangle u + \sum_{l=2}^\infty \left\langle \D^l\Iv,\Iv\right\rangle_\pi u^l + \frac{b_c^2u}{5c_c}\left(\frac{5c_cu}{\varepsilon}\right)^{l-1} \right),
\end{split}
\]
where we have used the fact that $1+x\leq e^x$. We can further bound this by focusing on the latter two terms inside the exponential:
\begin{enumerate}
\item
\[
\sum_{l=2}^{\infty}\langle \D^l {\Iv},{\Iv} \rangle_\pi u^l\leq\sum_{l=2}^\infty b_c^2 c_c^{l-2}u^l=b_c^2\frac{u^2}{1-c_cu}
\]
\item
\[
\begin{split}
\sum_{l=2}^\infty\frac{b_c^2u}
{5c_c}\left(\frac{5c_cu}{\varepsilon}\right)^{l-1}&=\sum_{l=2}^\infty\frac{b_c^2u^2}{\varepsilon}\left ( \frac{5c_cu}{\varepsilon}\right)^{l-2}\\
&=\frac{b_c^2u^2}{\varepsilon}\left( 1-\frac{5c_cu}{\varepsilon} \right)^{-1}.\\
\end{split}
\]
\end{enumerate}

The power series for point 1. again gives a radius of convergence of $0\leq u<\frac{1}{c_c}$. Point 2. gives a radius of convergence of $0\leq u < \frac{\varepsilon}{5c_c}<\frac{\varepsilon}{c_c(2+\varepsilon)}<\frac{1}{c_c}$. Combining these terms together and using the upper bound on the Laplace transform in equation \eqref{expectationUB}, we have, for $0<u<\frac{\varepsilon}{5c_c}$:
\begin{equation}\label{notOptConc1}
    \E_\nu[e^{(uT_{\EE}(k))}]\leq C(\nu)\exp\left(k\left(\langle t_{\EE} \rangle u+b_c^2u^2\left(\frac{1}{1-c_cu}+\frac{1}{\varepsilon-5c_cu}\right)\right)\right).
\end{equation}

Since $1-c_cu>\varepsilon-5c_cu$, we can relax slightly the bound on the moment generating function, such that when we apply the Chernoff bound in equation \eqref{chernBoundC}, we get that:
\begin{equation}\label{notOptConc2}
\PP_\nu \left (\frac{T_{\EE}(k)}{k} \geq \langle t_{\EE} \rangle + \gamma \right) \leq C(\nu)\exp\left(-k\left(\gamma u - \frac{2b_c^2u^2}{\varepsilon}\left(1-\frac{5c_cu}{\varepsilon}\right)^{-1}\right)\right).
\end{equation}
Consider the more general expression below, with $\alpha,\beta>0$:
\[
\gamma u - \alpha u^2\left(1-\beta u\right)^{-1}.
\]
Elementary calculations show that, for $|u|<\frac{1}{\beta}$:
\begin{equation}\label{optU}
\sup_u \left\{\gamma u - \alpha u^2\left(1-\beta u\right)^{-1}\right\}= \frac{\gamma^2}{2\alpha}h\left(\frac{\beta\gamma}{\alpha}\right)
\end{equation}
where $h(x):=(1+\frac{x}{2}+\sqrt{1+x})^{-1}$. In our case, $\alpha=\frac{2b_c^2}{\varepsilon}$, $\beta=\frac{5c_c}{\varepsilon}$. Therefore taking the infimum on the right hand side of the bound in \eqref{notOptConc2}, which is valid $\forall u\in [0,\frac{\varepsilon}{5c_c})$, yields the final result for right deviations:
\[
\PP_\nu \left (\frac{T_{\EE}(k)}{k} \geq \langle t_{\EE} \rangle + \gamma \right) \leq C(\nu) \exp \left ( -k \frac{\gamma^2 \varepsilon}{4{b_c^2}}h\left ( \frac{5c_c\gamma}{2 b_c^2}\right )\right ).
\]

To prove the concentration inequality for left deviations, we write the Chernoff bound for this case, this time with $u\leq0$:
\[
\PP_\nu\left(\frac{T_{\EE}(k)}{k}\leq\langle t_{\EE} \rangle-\gamma\right)\leq e^{-ku(\langle t_{\EE} \rangle-\gamma)}\E_\nu[e^{uT_{\EE}(k)}].
\]
We can repeat the proof for right deviations, due to the fact we are upper bounding the absolute value of the terms in the expansion of \eqref{eq:ru_exp} for $l\geq 2$. We obtain an upper bound on the moment generating function:
\[
\E_\nu[e^{(uT_{\EE}(k))}]\leq C(\nu)\exp\left(k\left(\langle t_{\EE} \rangle u+b_c^2u^2\left(\frac{1}{1-c_c|u|}+\frac{1}{\varepsilon-5c_c|u|}\right)\right)\right),
\]
which is valid for $0\leq|u|<\frac{\varepsilon}{5c_c}$. One obtains a concentration bound in terms of $u$ of a similar form to \eqref{notOptConc2}, which when optimised over the allowed $u$ gives the final result for left deviations and concludes the proof of Theorem \ref{theo:dynActC}:
\[
\PP_\nu \left (\frac{T_{\EE}(k)}{k} \leq \langle t_{\EE} \rangle - \gamma \right) \leq C(\nu) \exp \left ( -k \frac{\gamma^2 \varepsilon}{4{b_c^2}}h\left ( \frac{5c_c\gamma}{2 b_c^2}\right )\right ).
\]
\end{proof}

\bigskip \begin{coroA}
The variance of the first passage time for the total activity at stationarity is bounded from above by:
\[
\frac{{\rm var}_\pi(T_{\EE}(k))}{k}\leq \left(1 + \frac{2}{\varepsilon}\right)b_c^2.
\]
\end{coroA}
\begin{proof} 
\bigskip Notice that for $u\geq 0$ small enough, one has
\[
\begin{split}
\log(\E_\pi[e^{uT_{\EE}(k)}])&=\langle t_\EE \rangle ku+\frac{1}{2}{\rm var}_{\pi}(T_{\EE}(k))u^2+o(u^2)\\
&\leq k\log(r(u))=kr^\prime(0) u + \frac{k}{2}(r^{\prime\prime}(0)-(r^\prime(0))^2)u^2 + o(u^2)
\end{split}
\]

where $r(u)$ is given by equation \eqref{ruExp}. We recall that
\[r^\prime(0)=r^{(1)}=\langle t_\EE \rangle = \langle \D\Iv, \Iv \rangle_\pi, \quad r^{\prime\prime}(0)=2r^{(2)}=2\langle \D\Iv, \D\Iv\rangle_\pi + 2\left\langle \D\Iv,\frac{\hat{\P}}{\I-\hat{\P}} \D\Iv\right\rangle_\pi
\]
therefore
\[
r^{\prime\prime}(0)-(r^\prime(0))^2=\langle {\bf D} \Iv, \Iv \rangle_\pi^2 + 2\left \langle \D\Iv, \frac{\I}{\I-\hat{\P}} \D\Iv\right \rangle_\pi \leq \left (1+\frac{2}{\varepsilon} \right )b_c^2.
\]
Hence ${\rm var}_\pi(T_{\EE}(k))\leq \left (1+\frac{2}{\varepsilon} \right )b_c^2k$.

\end{proof}

%%%%%%%%%%%%%%%%%%%%%%%%%%%%%%%%%%%%%%
\section{Proof of Theorem \ref{theo:countObs}, Corollary \ref{coro:cgiTUR} and  Corollary \ref{countObsBeta} %and Corollary \ref{countObsiTUR}
}
\label{app:countObs}
%%%%%%%%%%%%%%%%%%%%%%%%%%%%%%%%%%%%%%%

\begin{theoA}[Rare Fluctuations of General Counting Observable FPTs]
Let $\L$ be irreducible and $\A\subseteq \EE$ be nonempty. For every $k\in \mathbb{N}$ and $\gamma > \beta - \langle t_\A \rangle$
\[
\PP_\nu \left (\frac{T_{\A}(k)}{k} \geq \langle t_{\A} \rangle + \gamma \right) \leq \exp \left ( -k \left(\frac{\gamma+\langle t_{\A} \rangle-\beta}{\beta}-\log\left(\frac{\gamma+\langle t_{\A} \rangle}{\beta}\right)\right)\right ). 
\]
\end{theoA}

\begin{proof}The proof begins with the same procedure as Theorems \ref{theo:dynActC} and \ref{theo:dynActQ}, but differs in that we do not use the $L_2(\pi)$ Hilbert space.
We again begin by applying Chernoff bound:
\begin{equation}\label{countObsChern}
\PP_\nu\left(\frac{T_{\A}(k)}{k}\geq\langle t_{\A} \rangle+\gamma\right)\leq e^{-uk(\langle t_{\A} \rangle+\gamma)}\E_\nu[e^{uT_{\A}(k)}], \quad u>0.
\end{equation}
The next step is to upper bound the Laplace transform: for $0 \leq u < \|\L_\infty^{-1}\|^{-1}_{\infty\rightarrow \infty}$ one has
\[
\begin{split}
    \E_\nu[e^{uT_{\A}(k)}]&=\left\langle \nu,\left(\frac{\I}{\I+\frac{u}{\L_\infty}} \Q\right)^k\Iv\right\rangle\\
    &=\left\langle \nu,\left(\sum_{i=0}^\infty u^i\left(-\frac{\I}{\L_\infty}\right)^i\Q \right)^k\Iv\right\rangle\\
    &\leq\underbrace{\|\nu\|_{1}}_{=1}\left(\sum_{i=0}^\infty u^i\left\|\frac{\I}{\L_{\infty}}\right\|_{\infty\rightarrow \infty}^i\right)^k\underbrace{\|\Iv\|_{\infty}}_{=1}.
\end{split}
\]
If we denote $\beta:=\left \|\L_{\infty*}^{-1}\right \|_{1\rightarrow 1}$, then for $u<\frac{1}{\beta}$:
\begin{equation}\label{countObsMgfUB}
\E_\nu[e^{uT_{\A}(k)}]\leq \left(\frac{1}{1-\beta u}\right)^k = \exp\left(-k\log(1-\beta u)\right).
\end{equation}
Placing this back into equation \eqref{countObsChern} we get:
\begin{equation}\label{mgfCountObsNonOpt}
\PP_\nu\left(\frac{T_{\A}(k)}{k}\geq\langle t_{\A} \rangle+\gamma\right)\leq\exp\left(-k\left(u(\gamma + \langle t_{\A} \rangle)+\log(1-\beta u)\right)\right), \quad 0 \leq u <\beta^{-1}.
\end{equation}
The minimum of the RHS is achieved at
\[
u^*=\frac{1}{\beta}-\frac{1}{\gamma+\langle t_{\A} \rangle}=\frac{\gamma+\langle t_{\A} \rangle-\beta}{\beta(\gamma + \langle t_{\A} \rangle)}.
\]
We have $u^*<\beta^{-1}$, and  $u^*>0$ %which is smaller than $\beta^{-1}$ 
if 
$\gamma>\beta-\langle t_{\A} \rangle$. Substituting $u^*$ into equation \eqref{mgfCountObsNonOpt} gives the final result:
\[
\PP_\nu \left (\frac{T_{\A}(k)}{k} \geq \langle t_{\A} \rangle + \gamma \right) \leq \exp \left ( -k \left(\frac{\gamma+\langle t_{\A} \rangle-\beta}{\beta}-\log\left(\frac{\gamma+\langle t_{\A} \rangle}{\beta}\right)\right)\right ).
\]
This concludes the proof of Theorem \ref{theo:countObs}.
\end{proof}
\begin{coroA}
Given any non-empty set of jumps $\A$, the variance of the corresponding first passage time at stationarity is bounded from above by:
\[
\frac{{\rm var}_\varphi(T_{\A}(k))}{k}\leq \left ( 1+ \frac{2}{\tilde{\varepsilon}}\right )\beta^2,
\] 
where
$$\tilde{\varepsilon}:=1-\max\{\|\Q f\|_\infty: \|f\|_\infty=1, \,\langle\varphi,f \rangle=0\}.$$ 
\end{coroA}
\begin{proof}
From the proof of Lemma \ref{lem:asympVar} one can see that
\[
\begin{split}
\frac{{\rm var}_{\varphi}[T_\A(k)]}{k} &= \left\langle \varphi, \L_\infty^{-2}\Iv \right\rangle^2  + 2\left\langle \varphi, \L_\infty^{-1}(\I-\Pi_\varphi)\L_\infty^{-1}\Iv \right\rangle\\
&+\frac{2}{k} \left \langle \varphi, \L_\infty^{-1}\sum_{i=2}^k\sum_{j=1}^{i-1}\Q^j (\I -\Pi_\varphi)\L_\infty^{-1} \Iv \right \rangle\\
&\leq \left (1+2\left ( 1+\frac{1}{k}\sum_{i=2}^k\sum_{j=1}^{i-1}(1-\tilde{\varepsilon})^j\right ) \right )\beta^2\\
&=\left (1+\frac{2}{\tilde{\varepsilon}} \right )\beta^2 -\frac{2((1-\tilde{\varepsilon})-(1-\tilde{\varepsilon})^{k+1})}{k \tilde{\varepsilon}^2}\beta^2\\
& \leq \left (1+\frac{2}{\tilde{\varepsilon}} \right )\beta^2.
\end{split}\]
\end{proof}
\begin{coroA}[Simple Upper Bound on $\beta$]
For general counting observables, the norm $\beta:=\left\|\L_{\infty}^{-1}\right\|_{\infty\rightarrow \infty}$ is bounded from above by:
\[
\beta \leq c_c\tilde{k}\max_{(x,y) \notin \A} \left \{ \frac{R_{x}}{w_{xy}}\right \}^{\tilde{k}-1}\max_{(x,y) \in \A} \left \{ \frac{R_{x}}{w_{xy}}\right \}\leq c_c\tilde{k}\left (\frac{R_{\rm max}}{w_{\rm min}}\right)^{\tilde{k}}=:\tilde{\beta},
\]
with $c_c, R_{\rm max}, w_{\rm min}$ defined in Eqs. \eqref{eq:cc}, \eqref{eq:Rmax}, and respectively \eqref{eq:wmin}.
The concentration bound in Theorem \ref{theo:countObs} holds with $\beta$ replaced by any of the two upper bounds above.
\end{coroA}

\begin{proof} 
From the expression \eqref{eq:Linfty.inverse} for $-\L_\infty^{-1}$  we obtain
\[
-\frac{\I}{\L_{\infty}}= \sum_{k=0}^\infty \left(\frac{\I}{\RR}\W_{2}\right)^k\frac{\I}{\RR}.
\]
Let $\tilde{k}>0$ be the minimax jump distance as it has been defined before Corollary \ref{countObsBeta}, i.e. the maximum over all states of the minimum number of jumps which suffice to get from that stat to a final jump between states in $\A$. The previous sum can be written in the following way:
\[
\sum_{l=0}^\infty\left(\frac{\I}{\RR}\W_{2}\right)^l=\sum_{m=0}^{\tilde{k}-1}\left(\frac{\I}{\RR}\W_{2}\right)^m\sum_{n=0}^\infty\left(\frac{\I}{\RR}\W_{2}\right)^{\tilde{k}n},
\]
where we break up $l$ into multiples of $\tilde{k}$ and a remainder term, since $\mathbb{N}_0=\cup_{m=0}^{\tilde{k}-1}m+\tilde{k}\mathbb{N}_0$. 
We can upper bound as
\[
\beta=\left \|\frac{\I}{\L_{\infty*}}\right \|_{\infty\rightarrow \infty}\leq\left \|\frac{\I}{\RR}\right \|_{\infty\rightarrow \infty}\sum_{m=0}^{\tilde{k}-1}\left \|\frac{\I}{\RR}\W_{2}\right \|_{\infty\rightarrow \infty}^m\sum_{n=0}^\infty\left\|\left(\frac{\I}{\RR}\W_{2}\right)^{\tilde{k}}\right\|_{\infty\rightarrow \infty}^n.
\]
Notice that $\left\|\RR^{-1}\right\|_{\infty\rightarrow \infty}=c_c$.
As $\A$ is non-empty, the spectral radius of $ \RR^{-1}\W_{2}$ is strictly less than 1 (see Lemma \ref{lem:tech1} ii)); therefore, there exists a $k$ such that $\left\|\left(\RR^{-1}\W_{2}\right)^{k}\right\|_{1\rightarrow 1}\leq q < 1$. We will now show that $k$ can be taken equal to $\tilde{k}$; in this case, we can write the upper bound of $\beta$ in terms of $q$ as:
\begin{equation}\label{betaEq}
\beta \leq c_c\tilde{k}\frac{1}{1-q}.
\end{equation}
So we just need to find $q$. Note that for any matrix ${\bf G}$ with positive entries, one has
$\|{\bf G}\|_{1\rightarrow 1} = \|G \Iv\|_\infty=\max_{x \in E}\langle \delta_x,G \Iv \rangle$. Therefore
\begin{equation}
\label{eq:norm11}
\left\|\left(\frac{\I}{\RR}\W_{2}\right)^{\tilde{k}}\right\|_{1\rightarrow 1} = \left \langle \delta_{x_0},\left(\frac{\I}{\RR}\W_{2}\right)^{\tilde{k}}\Iv\right \rangle,
\end{equation}
where $x_{0}$ is the state which attains the norm. From the structure of the generator, we know  that
\begin{equation}\label{w2xyEq}
\frac{\I}{\RR} \W_2\Iv=\Iv-\frac{\I}{\RR} \W_1\Iv.
\end{equation}
By the definition of the minimax jump distance, we know that there exists a path $x_0,\dots,x_l$ that happens with positive probability and such that $(x_{l-1},x_{l}) \in \A$. We can rewrite the right side of \eqref{eq:norm11} as
\[\begin{split}
&\left \langle \delta_{x_0},\left(\frac{\I}{\RR}\W_{2}\right)^{\tilde{k}}\Iv\right \rangle \leq \left \langle \delta_{x_0},\left(\frac{\I}{\RR}\W_{2}\right)^{l}\Iv\right \rangle =\\
&\left \langle \delta_{x_0},\left(\frac{\I}{\RR}\W_{2}\right)^{l-1}\Iv\right \rangle -\left \langle \delta_{x_0},\left(\frac{\I}{\RR}\W_{2}\right)^{l-1}\left(\frac{\I}{\RR}\W_{1}\right)\Iv\right \rangle=:g_+-g_{-}. 
\end{split}\]
In the first inequality we used the fact that $\mathbf{R}^{-1}\W_1$ is submarkovian, while in the second equality we made use of Eq. \eqref{w2xyEq}.
It is easy to see that $g_+ \leq 1$. Moreover, we know that
\[\begin{split}
    g_- &\geq \frac{w_{x_0x_1}}{R_{x_0}} \cdots \frac{w_{x_{l-1}x_{l}}}{R_{x_{l-1}}} \geq \min_{(x,y) \notin \A} \left \{ \frac{w_{xy}}{R_{x}}\right \}^{\tilde{k}-1}\min_{(x,y) \in \A} \left \{ \frac{w_{xy}}{R_{x}}\right \} \\
    &\geq \min_{(x,y) \in \EE}\left \{ \frac{w_{xy}}{R_{x}}\right \}^{\tilde{k}}\geq \left (\frac{w_{\rm min}}{R_{\max}} \right )^{\tilde{k}}.
\end{split}\]
Hence we can take
$$q=1-\min_{(x,y) \notin \A} \left \{ \frac{w_{xy}}{R_{x}}\right \}^{\tilde{k}-1}\min_{(x,y) \in \A} \left \{ \frac{w_{xy}}{R_{x}}\right \} \leq 1-\left(\frac{w_{\rm min}}{R_{\rm max}}\right)^{\tilde{k}}.$$ Applying this to equation \eqref{betaEq}, we get that
$$\beta \leq c_c\tilde{k}\max_{(x,y) \notin \A} \left \{ \frac{R_{x}}{w_{xy}}\right \}^{\tilde{k}-1}\max_{(x,y) \in \A} \left \{ \frac{R_{x}}{w_{xy}}\right \} \leq c_c\tilde{k}\left(\frac{R_{\rm max}}{w_{\rm min}}\right)^{\tilde{k}}.$$
\end{proof}

%%%%%%%%%%%%%%%%%%%%%%%
%Computation of variance (classical)
%%%%%%%%%%%%%%%%%%%%%%%
\begin{lemma}[Variance of Classical FPTs]   \label{lem:asympVar}
Let $\varphi$ be the invariant measure of $\Q$, cf. Eq \eqref{eq:invmeas}, and let $\Pi_\varphi$ be the map $\Pi_\varphi: f\mapsto \langle\varphi,f\rangle \Iv $. The variance of the first passage time for counting observables is given by:
    \begin{equation}
    \begin{split}
    \frac{{\rm var}_\varphi\left( T_{\A}(k)\right)}{k} = \left \langle \varphi, \L_\infty^{-1} \Iv \right\rangle ^2 &+ 2 \left \langle \varphi, \L_\infty^{-1}\frac{\I}{\I - \Q}(\I -\Pi_\varphi)\L_\infty^{-1}\Iv\right\rangle\\
    &- \frac{2}{k}\left\langle \varphi,\L_\infty^{-1} \frac{\Q - \Q^{k+1}}{(\I - \Q)^2}(\I -\Pi_\varphi)\L_\infty^{-1} \Iv \right\rangle, \quad \forall k \geq 0.
    \end{split}
    \end{equation}
\end{lemma}

\begin{proof}
    We recall the explicit expression for the moment generating function from Lemma \ref{lem:tech1}
    \[
        \E_\varphi[e^{uT_{\A}(k)}]=\left \langle \varphi, \left(\frac{\L_\infty}{u+\L_\infty}\Q \right)^k\Iv\right \rangle.
    \]

    We can write the first moment as

    \begin{equation}\label{eq:mgfDeriv}
        \E_\varphi[e^{uT_\A(k)}]' = -\left \langle \varphi, \sum_{i=1}^k \left ( \frac{\L_\infty}{u + \L_\infty} \Q\right )^{i-1}\left ( \frac{\L_\infty}{(u + \L_\infty)^2} \Q\right )\left ( \frac{\L_\infty}{u + \L_\infty} \Q\right )^{k-i}\Iv \right \rangle.
    \end{equation}
    
    At $u = 0$ this gives us the form of the asymptotic mean.
    \begin{equation}\label{eq:asympMean}
        \E_\varphi[T_\A(k)] = -k\left \langle \varphi, \L_\infty^{-1}\Iv \right \rangle.
    \end{equation}

    Differentiating equation \eqref{eq:mgfDeriv} at $u=0$ gives us the second moment:
    \begin{equation*}
    \begin{split}
        \E_\varphi[T_\A(k)^2] &= 2 \left \langle \varphi, \L_\infty^{-2}\Iv \right \rangle k + \sum_{i=1}^k \left \langle \varphi, \L_\infty^{-1}\left ( \sum_{j=1}^{i-1}\Q^{i-j} + \sum_{j=1}^{k-i}\Q^j\right ) \L_\infty^{-1}\Iv \right \rangle.\\
        &= 2\left\langle \varphi, \L_\infty^{-2}\Iv \right\rangle k + 2 \left \langle \varphi, \L_\infty^{-1} \sum_{1 \leq j < i \leq k}\Q^j \L_\infty^{-1}\Iv \right \rangle\\
        &= 2\left\langle \varphi, \L_\infty^{-2}\Iv \right\rangle k + 2 \sum_{i=2}^k\sum_{j=1}^{i-1}\left \langle \varphi, \L_\infty^{-1}\Iv \right \rangle^2\\
        &\qquad\qquad\qquad\quad+ 2 \left \langle \varphi, \L_\infty^{-1}\sum_{i=2}^k\sum_{j=1}^{i-1}\Q^j (\I -\Pi_\varphi)\L_\infty^{-1} \Iv \right \rangle\\
    \end{split}
    \end{equation*}
    where to arrive at the third line, after $\Q^j$ we have inserted $\Pi_\varphi + \I - \Pi_\varphi$, where $\Pi_\varphi$ is the projection onto $\Iv$. Using the fact that $\sum_{i=2}^k\sum_{j=1}^{i-1}1 = \frac{k}{2}(k-1)$, and recalling equation \eqref{eq:asympMean} for the expression for the first moment, hence
    
    \begin{equation*}
        \begin{split}
        \frac{{\rm var}_\varphi\left ( T_{\A}(k)\right)}{k} = 2\left \langle \varphi, \L_\infty^{-2}\Iv\right \rangle &- \left \langle \varphi, \L_\infty^{-1}\Iv \right \rangle^2 + 2 \left \langle \varphi, \L_\infty^{-1}\frac{\Q}{\I - \Q}(\I -\Pi_\varphi)\L_\infty^{-1}\Iv \right \rangle\\
        &-\frac{2}{k} \left \langle \varphi,\L_\infty^{-1}\frac{\Q - \Q^{k+1}}{(\I - \Q)^2}(\I -\Pi_\varphi)\L_\infty^{-1}\Iv \right \rangle.
    \end{split}
    \end{equation*}
    Finally, we can again place $\Pi_\varphi + \I - \Pi_\varphi$ in the first term, in between the two $\L_\infty^{-1}$. Rearranging this gives the final result
    \begin{equation*}
        \begin{split}
        \frac{{\rm var}_\varphi\left ( T_{\A}(k)\right)}{k} = \left \langle \varphi, \L_\infty^{-1}\Iv\right \rangle^2 &  + 2 \left \langle \varphi, \L_\infty^{-1}\frac{\I}{\I - \Q}(\I -\Pi_\varphi)\L_\infty^{-1}\Iv \right \rangle\\
        &-\frac{2}{k} \left \langle \varphi,\L_\infty^{-1}\frac{\Q - \Q^{k+1}}{(\I - \Q)^2}(\I -\Pi_\varphi)\L_\infty^{-1}\Iv \right \rangle.
    \end{split}
    \end{equation*}
    
\end{proof}

%%%%%%%%%%%%%%%%%%%%%%%%%%%%%%%%%%%%%%%
\section{Quantum Markov Processes}\label{app:markovProofQ}

%%%%%%%%%%%%%%%%%%%%%%%%%%%%%%%%%%%%%%%
\subsection{Proof of Lemma \ref{lemma.statioarystate.generators}}
%%%%%%%%%%%%%%%%%%%%%

\begin{lemmaA}
The generator $\LL$ has a unique invariant state if and only if $\Phi$ does. If $\Phi$ is irreducible then $\mathcal{L}$ is irreducible, but the converse is generally not true.
\end{lemmaA}

\begin{proof}
Indeed if  $\LL_*(\nu)=0$ for some state $\nu$ then 
$$
\Phi_*(\J_*(\nu)) = 
-\J_* \mathcal{L}_{0*}^{-1} (- \mathcal{L}_{0*} (\nu)) = \J_*(\nu) 
$$
so $\nu^\prime:= \J_*(\nu)/\tr [\J_*(\nu)]$ is a stationary state for $\Phi$. Notice that if $\nu$ is strictly positive, then this is not necessarily true for
$\nu^\prime $, depending on the form of the jump operators. Conversely, if $\Phi_* (\nu) = \nu$ for some state $\nu$ then 
$$
\mathcal{L}_* [\mathcal{L}_{0*}^{-1} (\nu)]= 
(\mathcal{L}_{0*}+ \J_*) [\mathcal{L}_{0*}^{-1} (\nu)] = \nu + 
\J_* \mathcal{L}_{0*}^{-1} (\nu) =
\nu -\Phi_*(\nu ) = 0
$$
so $\nu^\prime = \mathcal{L}_{0*}^{-1} (\nu)/ \tr[\mathcal{L}_{0*}^{-1} (\nu)]$ is a stationary state for $\LL$. Here we used the fact that $-\mathcal{L}^{-1}_{0} = \int_0^{\infty} e^{t\mathcal{L}_0}$ is completely positive. If $\Phi$ is irreducible then $\nu>0$ and $-\LL_{0*}^{-1}(\nu)>0$, therefore $\mathcal{L}$ is irreducible, and Hypothesis \ref{hypo:irrQuantPhi} implies Hypothesis \ref{hypo:irrQuantL}.

\end{proof}

%%%%%%%%%%%%%%%%%%%%%%%%%%%%%%%%%%%%%%%
\subsection{Proof of Lemma \ref{lem:tech1Q}}
%%%%%%%%%%%%%%%%%%%%%

\begin{lemmaA}
Assume that Hypothesis \ref{hypo:irrQuantL} ($\mathcal{L}$ is irreducible) holds. Then the following statements are true:
    \begin{enumerate}
        \item $\overline{\lambda}:= -\max\{\Re(z): z \in {\rm Sp}({\cal L}_\infty)\}>0$, hence ${\cal L}_\infty$ is invertible;
        \item for every $u < \overline{\lambda}$, one has
        \[
        \E_\rho[e^{uT_{\A}(k)}]=\tr\left ( \rho \left((u+\mathcal{L}_\infty)^{-1}\mathcal{L}_\infty\Psi \right)^k(\I)\right ),
        \]
        \item $\|{\cal L}_\infty^{-1}\|_{\infty \rightarrow \infty}^{-1} \leq \overline{\lambda}.$
    \end{enumerate}
\end{lemmaA}

\begin{proof}
$\mathcal{L}_\infty$ generates a sub-Markov semigroup $e^{t\mathcal{L}_\infty}$. From the Spectral Mapping Theorem, one has that
\[
{\rm Sp}(e^{t\mathcal{L}_\infty}) = e^{t{\rm Sp}(\mathcal{L}_\infty)}, \quad \forall t\geq 0.
\]
Moreover, since $e^{t\mathcal{L}_\infty}$ is completely positive, by Perron-Frobenius Theorem (Theorem \ref{th:PF}) the spectral radius and largest eigenvalue of $e^{t\mathcal{L}_\infty}$ coincide. Hence:
\[
r(e^{t\mathcal{L}_\infty})=e^{t\lambda}, \lambda:=\max\{\Re(z):z\in {\rm Sp}(\mathcal{L}_\infty)\},
\]
and $\exists x\in M_d(\mathbb{C})$ with $ x\geq 0$:
\[
\mathcal{L}_\infty(x) = \lambda x.
\]
By contradiction, suppose that $\lambda = 0$, then
\[
\LL(x) = \J_\A(x) + \mathcal{L}_\infty(x) = \J_\A(x).
\]
Therefore
\[
\tr(\hat{\sigma} \LL(x)) = \tr(\hat{\sigma} \J_\A(x)) = 0,
\]
since $\LL_*(\hat{\sigma}) = 0$. From the irreducibility assumption (Hypothesis \ref{hypo:irrQuantL}), we have that $\hat{\sigma}>0$, and so $\LL(x) = \J_\A(x) = 0$. Moreover, irreducibility also implies that $x= \alpha\I$ for some $\alpha\in\mathbb{R}$. Therefore $\J_\A(x) = \alpha \J_\A(\I) = 0$. Since $\J_\A(\I)$ is the sum of positive operators, we have that $\alpha = 0 $, consequently $x=0$ and we reach a contradiction.

\bigskip 2. Integrating over all trajectories, one can write:
\[
\PP(T_{\A}(k)\leq t) = \int_{t_1<t_2<\cdots <t_k <t}\tr \left (e^{(t-t_k)\LL_{\infty*}}\J_{\A*}\cdots e^{(t_2-t_1)\LL_{\infty*}}\J_{\A*} e^{t_1\LL_{\infty*}}(\rho)\right )dt_1\cdots dt_k.
\]

For $u < \overline{\lambda}$, one has $$-(u+\LL_\infty)^{-1}=\int_0^{+\infty} e^{(u+\LL_\infty)t}dt,$$ 

hence one can write the Laplace transform of $T_\A(k)$ as:
\[
\mathbb{E}_\rho[e^{uT_{\A}(k)}] = \tr \left( \left( -\J_{\A*}(u+\LL_{\infty*})^{-1} \right)^k (\rho) \right),
\]
and by the definition of $\Psi$ in equation \eqref{eq:psi}, we obtain the statement.

\bigskip 3. The Spectral Mapping Theorem implies that ${\rm Sp}(\LL_\infty^{-1})=\{z^{-1}:\, z \in {\rm Sp}(\LL_\infty)\}$, therefore one has that
\[
 \left\| \LL_\infty^{-1} \right \|_{\infty \rightarrow \infty} \geq r(\LL_\infty^{-1})\geq \frac{1}{\overline{\lambda}} \quad\Leftrightarrow \quad \left\| \LL_\infty^{-1} \right \|_{\infty \rightarrow \infty}^{-1} \leq \overline{\lambda}.
\]
    \end{proof}

%%%%%%%%%%%%%%%%%%%%%%%%%%%%%%
\subsection{Proof of Theorem \ref{theo:ldpQ}}\label{app:proof.ldpQ}
%%%%%%%%%%%%%%%%%%%%%%%%%%%%%%

\bigskip\begin{theoA} Consider a nonempty subset $\A$ of the emission channels. The FPT $T_\A(k)/k$ satisfies a large deviation principle with good rate function given by
$$
I_\A(t):=\sup_{u \in \mathbb{R}}\{ut-\log(r(u))\}$$
where
$$
r(u)=\begin{cases} r \left (\Psi_u \right ) & \text{ if } u < \overline{\lambda}\\
+\infty & \text{o.w.}\end{cases}$$
where $\Psi_u(x):= -(u+\LL_\infty)^{-1}\J_\A(x)$ and $\overline{\lambda}:=-\max\{\Re(z):z \in {\rm Sp}(\LL_\infty)\}.$
\end{theoA}

\begin{proof}
The proof follows the same method as the proof of the classical case, Theorem \ref{theo:ldpC}; for completeness we write the quantum proof in full. Lemma \ref{lem:irredQ} states that in the domain $u<\overline{\lambda}$, then
$$
\E_\rho[e^{uT_\A(k)}] = \tr\left ( \rho \Psi_u^k(\I)\right ),
$$
where $\Psi_u(x):= -(u+\LL_\infty)^{-1}\J_\A(x)$. Writing $\Psi_u$ as the integral
$$
\Psi_u(x) = \int_0^{\infty}e^{(u + \LL_\infty)t}\J(x)dt
$$
which is the composition of two completely positive maps, hence $\Psi_u$ is completely positive as well. Therefore, Perron-Frobenius Theorem tells us that $r(u):=r(\Psi_u)$ is an eigenvalue of $\Psi_u$, with a positive eigenvector $x(u)$. We can relate this operator with the generator of a quantum dynamical semigroup:
$$
\Psi_u (x(u))=r(u)x(u) \Leftrightarrow \LL_{s(u)}(x(u))=-u x(u),
$$
where $\LL_{s(u)}:=\LL + (e^{s(u)} - 1)\J_\A$ and $s(u):= - \log(r(u))$. Notice that $\LL_{s(u)}$ has the form in Eq. \eqref{eq:perturbed}. Therefore, by Lemma \ref{lem:irredQ} it is irreducible and $x(u)$ is in fact a unique and strictly positive eigenvector of $\LL_{s(u)}$ corresponding to the eigenvalue $-u$.

We first need to show that for $u < \overline{\lambda}$
\begin{equation}\label{eq:QLDPconverge}
\lim_{k\to\infty}\frac{1}{k}\log(\E_\rho[e^{uT_\A(k)}])<\infty,
\end{equation}

Using Holder's inequality, we have that $\tr\left ( \rho \Psi_u^k(\I)\right ) \leq \|\Psi_u^k\|_{\infty \to \infty}$, hence due to Gelfand's formula
$$
\frac{1}{k}\log(\E_\rho[e^{uT_\A(k)}]) \leq \frac{1}{k}\log(\|\Psi_u^k\|_{\infty \rightarrow \infty}) \rightarrow_{k\rightarrow +\infty} \log(r(u)).
$$
Furthermore, we can take $\|x(u)\|_{\infty\to\infty}\leq 1$ and bound from below:
$$
\frac{1}{k}\log(\mathbb{E}_\rho[e^{uT_\A(k)}]) \geq \log(r(u))+\frac{1}{k}\log(\tr\left(\rho x(u)\right))\rightarrow_{k\rightarrow +\infty} \log(r(u)),
$$
which we can do since $x(u)$ is strictly positive, $\rho$ is positive semidefinite so $\tr(\rho x(u))>0$. Therefore, we have shown that the limit in equation \eqref{eq:QLDPconverge} converges to $\log(r(u))<\infty$ in the range $u<\overline{\lambda}$.

We now use the same version of the G\"{a}rtner-Ellis  Theorem (\cite[Theorem 2.3.6]{dembo2010large}) as in the proof of Theorem \ref{theo:ldpC}.
All that remains is to show that $\log(r(u))$ is steep, i.e. as u approaches the boundary $\overline{\lambda}$, both $\log(r(u))$ and $\log(r(u))'$ diverge to $+\infty$. Denote by $\mathbf{T}$ the spectral projection of $\LL_\infty$ with respect to the eigenvalue $-\overline{\lambda}$. We can show that $\LL_\infty\mathbf{T}$ is diagonalisable, i.e. the restriction does not feature any Jordan blocks or equivalently, the algebraic and geometric multiplicity of $\overline{\lambda}$ coincide. To do this, assume $\LL_\infty \mathbf{T}$ is not diagonalisable. Then the map $\LL'\mathbf{T} = (\LL + \overline{\lambda}\Id)\mathbf{T}$ - corresponding to the eigenvalue 0 - is not diagonalisable. One can then always choose $x,y\in M_d(\mathbb{C})$ with $\LL'(y)=0$ and $\LL'(x) = y$. Therefore one has

$$
e^{t\LL'}(x) = x + ty.
$$
The semigroup $e^{t\LL'}$ is in fact a contraction semigroup, which contradicts the above equation. Therefore, we have
$$
\LL_\infty\mathbf{T} = -\overline{\lambda}\mathbf{T}.
$$

Firstly, we show that $\lim_{u \uparrow \overline{\lambda}}r(u)=+\infty$. Note that $\mathbf{T}\Psi\neq 0$ since $\mathbf{T}\Psi(\I)\neq 0$. The map $\Psi_u$ can be written as

\[
\Psi_u = \frac{\LL_\infty}{\LL_\infty + u} = \frac{\overline{\lambda}}{\overline{\lambda} - u}\mathbf{T}\Psi + (\Id - \mathbf{T})\frac{\LL_\infty}{\LL_\infty + u}\Psi
\]

and one can see it has a norm which explodes as $u\uparrow \overline{\lambda}$. Let us assume by contradiction that $r(u)\uparrow r(\overline{\lambda})<\infty$ as $u\uparrow \overline{\lambda}$. Then the eigenvactor $x(u)$ can be chosen to converge to the Perron-Frobenius eigenvector of $\LL_{s(\overline{\lambda})}$ and $\min{\rm Sp}(x(u)) \not\rightarrow0$. We then have for $0\leq u < \overline{\lambda}$

$$
    \|\Psi_u\|_{\infty \rightarrow \infty}=\|\Psi_u(\I)\|_\infty \leq \frac{1}{\min {\rm Sp}(x(u))}\|\Psi_u(x(u))\|_\infty = r(u) \frac{\|x(u)\|_\infty}{\min {\rm Sp}(x(u))}<+\infty
$$
which is a contradiction. Let us show that $\lim_{u \uparrow \overline{\lambda}}\log(r(u))'=+\infty$. For each $u$, we can choose the left eigenvector of $\Psi_u$, $l(u)$, to be such that $\tr(l(u)x(u)) \equiv 1$. Then we can write $r(u)$ as

\[
    \begin{split}
    r(u)&=-\tr \left( l(u)(u+\LL_\infty)^{-1}\LL_\infty \Psi(x(u))\right)\\
    &=\underbrace{-\tr \left( l(u)(u+\LL_\infty)^{-1}\LL_\infty\mathbf{T}\Psi(x(u)) \right)}_{(I)}+\underbrace{-\tr \left( l(u)(u+\LL_\infty)^{-1}\LL_\infty(\Id - \mathbf{T})\Psi(x(u)) \right)}_{(II)}.\\
    \end{split}
    \]
The term $(II)$ remains bounded, hence as $u\uparrow \overline{\lambda}$
$$
r(u) \asymp \frac{\overline{\lambda}}{\overline{\lambda} - u} \tr\left(l(u) \mathbf{T}\Psi(x(u))\right)\rightarrow +\infty.
$$
Since $\log(r(u))' = \frac{r'(u)}{r(u)}$, we can differentiate the expression for $r(u)$ to obtain

\[
\begin{split}
    \frac{r^\prime(u)}{r(u)}&=\frac{\tr\left( l(u)(u+\LL_\infty)^{-2}\LL_\infty\Psi(x(u)) \right)}{r(u)}\\
    &=\underbrace{\frac{\overline{\lambda}}{(\overline{\lambda}-u)^2}\frac{\tr \left( l(u)\mathbf{T}\Psi(x(u)) \right)}{r(u)}}_{(I)}+ \underbrace{\frac{\tr \left( l(u)(\Id - \mathbf{T})\Psi(x(u)) \right)}{r(u)}}_{(II)}.
    \end{split}
    \]
We have again that part $(II)$ is bounded, but for $u\uparrow \overline{\lambda}$, $(I)\asymp (\overline{\lambda} - u)^{-1}$ which completes the proof.

\end{proof}

%%%%%%%%%%%%%%%%%%%%%%%%%%%%%%
\subsection{Proof of Theorem \ref{theo:dynActQ} and Corollary \ref{dynActiTURQ}}
\label{app:proof.th.dynActQ}
%%%%%%%%%%%%%%%%%%%%%%%%%%%%
\begin{theoA}[Fluctuations of FPT for Total Counts]
Assume that Hypothesis \ref{hypo:irrQuantPhi} holds ($\Phi$ be irreducible) and let  $\varepsilon$ be the absolute spectral gap of $\Phi$. Then, for every $\gamma>0$:
\begin{equation*}
\begin{split}
&\PP_\rho \left (\frac{T_I(k)}{k} \geq \langle t_I \rangle + \gamma \right) \leq C(\rho) \exp \left ( -k \frac{\gamma^2 \varepsilon}{8c_q^2}h\left ( \frac{5\gamma}{2 c_q}\right )\right )\\{\rm and}\\
&\PP_\rho \left (\frac{T_I(k)}{k} \leq \langle t_I \rangle - \gamma \right) \leq C(\rho) \exp \left ( -k \frac{\gamma^2 \varepsilon}{8{c_q^2}}h\left ( \frac{5\gamma}{2 c_q}\right )\right ),\quad k\in \mathbb{N},
\end{split}
\end{equation*}
where $h(x):=(\sqrt{1+x}+\frac{x}{2}+1)^{-1}$, $C(\rho):=\left\|\sigma^{-\frac{1}{2}}\rho\sigma^{-\frac{1}{2}}\right\|_\sigma$ and $c_q$ is defined in Eq. \eqref{eq:cq}.
\end{theoA}

\begin{proof}

Applying Chernoff bound, we get, for $u\geq0$:
\begin{equation}\label{chernBoundQ}
\PP_\rho \left (\frac{T_I(k)}{k} \geq \langle t_I \rangle + \gamma \right) \leq e^{-uk(\langle t_I \rangle +\gamma)}\E_\rho\left [ e^{uT_I(k)}\right].
\end{equation}
For $u < \overline{\lambda}$, we define the tilted operator $\Phi_u(x):=\FF_u\Phi(x)$, with $\FF_u(x):=(u\Id+\LL_0)^{-1}\LL_0(x)$. Using Lemma \ref{lem:tech1Q}, for $0 \leq u < \overline{\lambda}$ we can write
\[
\begin{split}
\E_\rho[e^{uT_I(k)}]={\rm tr} \left ( \rho \Phi_u^k(\I) \right ) &= \left\langle \sigma^{-\frac{1}{2}}\rho\sigma^{-\frac{1}{2}},\Phi_u^k(\I) \right\rangle_\sigma\\
&\leq\underbrace{\|\sigma^{-\frac{1}{2}}\rho\sigma^{-\frac{1}{2}}\|_\sigma}_{:=C(\rho)}\underbrace{\|\I\|_\sigma}_{=1}\|\Phi_u^k\|_\sigma,
\end{split}
\]
where with a small abuse of notation $\|\Phi_u^k\|_\sigma$ denotes the operator norm of the map $\Phi_u^k$ with respect to the KMS inner product associated to $\sigma$. We can further break up this operator norm:
\[
\|\Phi_u^k\|_\sigma\leq\|\Phi_u\|_\sigma^k=\|\FF_u\Phi\|_\sigma^k.
\]
We now seek to upper bound $\|\FF_u\Phi_u\|_\sigma$. Conversely to the classical case, $\FF_u\Phi$ is not self adjoint, but we can upper bound its norm with an operator which is. This can be done using Lemma \ref{lpLem} with $\mathbf{A}=\FF_u,\mathbf{B}=\Id$ to get $\|\FF_u\Phi\|_\sigma\leq\|\FF_u\hat{\Phi} \FF_u^\dag\|_\sigma^{\frac{1}{2}}$, with 
$\FF_u^{\dag}(x)=\Gamma_\sigma^{-\frac{1}{2}}\circ (\Id+u\LL_{0*}^{-1})^{-1}\circ\Gamma_\sigma^{\frac{1}{2}}(x)$, cf. equation \eqref{eq:dagger.adjoint}. Since $\FF_u\hat{\Phi}\FF_u^\dag$ is a positive, self adjoint, irreducible map, operator Perron Frobenius theory \cite{wolf2012quantum} says $\|\FF_u\hat{\Phi}\FF_u^\dag\|_\sigma=r(u)$, where $r(u)=\sup\{|\lambda|:\lambda\in {\rm Sp}(\FF_u\hat{\Phi}\FF_u^\dag)\}$ is the spectral radius of $\FF_u\hat{\Phi}\FF_u^\dag$. Hence, the Laplace transform is upper bounded by:
\begin{equation}\label{expectationUBQ}
\E_\rho[e^{uT_I(k)}]\leq C(\rho)r(u)^{\frac{k}{2}}.
\end{equation}
For $u$ small enough, we can expand $\FF_u\hat{\Phi}\FF_u^\dag$ as the power series below:
\begin{equation*}
\begin{split}
\FF_u\hat{\Phi}\FF_u^\dag &= \sum_{j\geq0}u^j\left (-\frac{1}{\LL_0}\right )^j\circ\hat{\Phi}\circ\sum_{l\geq0}u^l\left (-\frac{1}{\LL_0^\dag}\right )^l\\
&=\sum_{l\geq0}u^l\underbrace{\sum_{j=0}^l\left (-\frac{1}{\LL_0}\right )^{l-j}\circ\hat{\Phi}\circ\left(-\frac{1}{\LL_0^\dag}\right )^j}_{\Phi^{(l)}},
\end{split}
\end{equation*}
which is the quantum analogue of equation \eqref{eq:series.P(u)} from the classical case.
Again we bound each term in the power series in order to use operator perturbation theory. Using the definition of $c_q$ in equation \eqref{eq:cq} and Cauchy-Schwarz, we can upper bound $\Phi^{(l)}$ in the form:
\[
\|\Phi^{(l)}\|_\sigma\leq (l+1) c_q^l\leq (2c_q)^l.
\]
Using perturbation theory \cite{kato1976perturbation}, we find that for  $u<\frac{\varepsilon}{2c_q(2+\varepsilon)}$, the spectral radius $r(u)$ can then be expressed as  
\begin{equation} \label{eq:qru1}
r(u) = 1 + \sum_{l=1}^{\infty}u^l r^{(l)},
\end{equation}
where
\begin{equation}
r^{(l)}=\sum_{p=1}^l \frac{(-1)^p}{p} \sum_{\substack{\nu_1+\dots +\nu_p=l,\, \nu_i\geq 1\\
\mu_1+\dots +\mu_p=p-1, \, \mu_j\geq 0}} \TR\left(\Phi^{(\nu_1)}\S^{(\mu_1)}\cdots \Phi^{(\nu_p)}\S^{(\mu_p)}\right).
\label{eq:ru_exp_q}
\end{equation}
We have $\S^{(0)}=-\Pi$, $\S^{(1)}=(\hat{\Phi}-\Id+\Pi)^{-1}-\Pi = 
-\epsilon^{-1}(\Id - \Pi)$, and $\S^{(\mu)}$ the $\mu^{\rm th}$ power of $\S^{(1)}$, and also $\|\S^{(\mu)}\|_\sigma=\varepsilon^{-\mu}$ for $\mu\geq 1$.
We now want to bound the terms in the expression for $r^{(l)}$. For $p=1$:
\begin{eqnarray}\label{rkpBound1}
\left|\TR\left(\Phi^{(\nu_1)}\S^{(\mu_1)}\cdots\Phi^{(\nu_p)}\S^{(\mu_p)}\right)\right|&=&\left|\TR\left( \Phi^{(l)}\S^{(0)}\right)\right|=\left|\left\langle \I,\Phi^{(l)}(\I)\right\rangle_\sigma\right|
\nonumber 
\\
&\leq&\|\I\|_\sigma\|\Phi^{(l)}\|_\sigma
\leq(2c_q)^l\leq 2c_q\left(\frac{2c_q}{\varepsilon}\right)^{l-1}
\end{eqnarray}
since $\varepsilon\leq 1$. For $p\geq 2$, using the fact that one of $\mu_i$ is zero and using trace cyclicity:
\begin{equation}\label{rkpBound2}
\begin{split}
    \left|\TR\left(\Phi^{(\nu_1)}\S^{(\mu_1)}\cdots\Phi^{(\nu_p)}\S^{(\mu_p)}\right)\right|&=\left|\left\langle\I,\Phi^{(\nu_1)}\S^{(\mu_1)}\cdots \S^{(\mu_{p-1})}\Phi^{(\nu_p)}(\I)\right\rangle_\sigma\right|\\
    &\leq\frac{1}{\varepsilon^{p-1}}(2c_q)^l\\
    &\leq\frac{1}{\varepsilon^{l-1}}(2c_q)^l\leq 2c_q\left(\frac{2c_q}{\varepsilon}\right)^{l-1}.
\end{split}
\end{equation}

We can explicitly calculate $r^{(1)}$:
\begin{equation*}
\begin{split}
r^{(1)}&=-\TR\left(\Phi^{(1)}(-\Pi)\right)=\left\langle \I,\Phi^{(1)}(\I)\right\rangle_\sigma\\
&=\left\langle \I,\left(-\frac{\Id}{\LL_0}\right)\hat{\Phi}(\I)\right\rangle_\sigma+\left\langle\I,\hat{\Phi}\left(-\frac{\Id}{\LL_0^\dag}\right)(\I)\right\rangle_\sigma=2\langle t_I \rangle.
\end{split}
\end{equation*}
The term $|r^{(2)}|$ can be bounded using equations \eqref{rkpBound1} and \eqref{rkpBound2}:
\begin{equation*}
\begin{split}
|r^{(2)}|&=\left|\TR\left(\Phi^{(2)}\S^{(0)} \right)-\left\langle\I,\Phi^{(1)}\S^{(1)}\Phi^{(1)}(\I)\right\rangle_\sigma\right|\\
&\leq \frac{(2c_q)^2}{\varepsilon} + \frac{(2c_q)^2}{\varepsilon} = \frac{8c_q^2}{\varepsilon}.
\end{split}
\end{equation*}
Using the bound on the number of terms in \eqref{eq:ru_exp_q}, given by equation \eqref{numTerms}, we can bound the $r^{(l)}$ by $|r^{(l)}|\leq\frac{2c_q}{5}\left(\frac{10c_q}{\varepsilon}\right)^{l-1}$. Now we have all the ingredients in place, we can finish bounding $r(u)$:
\begin{equation*}
    \begin{split}
        r(u)&\leq 1+\sum_{l\geq1}u^l|r^{(l)}|\\
        &\leq 1+2\langle t_I \rangle u +\frac{8c_q^2}{\varepsilon}u^2+\sum_{l\geq 3}u^l\frac{2c_q}{5}\left (\frac{10c_q}{\varepsilon}\right)^{l-1}\\
        &=1+2\langle t_I \rangle u+\frac{8c_q^2}{\varepsilon}u^2 + u^2\frac{4 c_q^2}{\varepsilon}\sum_{l\geq 1}u^l\left (\frac{10c_q}{\varepsilon}\right)^l\\
        &\leq 1+2\langle t_I \rangle u + \frac{8c_q^2}{\varepsilon}u^2\left (1-\frac{10c_qu}{\varepsilon}\right)^{-1}\\
        &\leq \exp\left({2\langle t_I \rangle u+\frac{8c_q^2}{\varepsilon}u^2\left (1-\frac{10c_qu}{\varepsilon}\right)^{-1}}\right),
    \end{split}
\end{equation*}
which is valid for $u<\frac{\varepsilon}{10c_q}<\frac{\varepsilon}{2c_q(2+\varepsilon)}$. Putting this upper bound on $r(u)$ back into equation \eqref{expectationUBQ} gives the upper bound on the Laplace transform of $T_I(k)$:
\begin{equation*}
    \E_\rho[e^{uT_I(k)}]\leq C(\rho) \exp\left(k\left(\langle t_I \rangle u +\frac{4c_q^2u^2}{\varepsilon}\left(1-\frac{10c_qu}{\varepsilon}\right)^{-1}\right)\right).
\end{equation*}
Applying the Chernoff bound in equation \eqref{chernBoundQ} gives, $\forall u\in [0,\frac{\varepsilon}{10c_q})$:
\[
\PP_\rho \left (\frac{T_I(k)}{k} \geq \langle t_I \rangle + \gamma \right) \leq C(\rho)\exp\left(-k\left(\gamma u - \frac{4c_q^2u^2}{\varepsilon}\left(1-\frac{10c_qu}{\varepsilon}\right)^{-1}\right)\right).
\]
Optimisation over the allowed $u$, using result \eqref{optU}, with $\alpha=\frac{4c_q^2}{\varepsilon}$, $\beta=\frac{10c_q}{\varepsilon}$ gives the final concentration inequality for right deviations:
\[
\PP_\rho \left (\frac{T_I(k)}{k} \geq \langle t_I \rangle + \gamma \right) \leq C(\rho) \exp \left ( -k \frac{\gamma^2 \varepsilon}{8{c_q^2}}h\left ( \frac{5\gamma}{2 c_q}\right )\right ).
\]

By considering $u\leq 0$ we can prove the bound for left deviations:
\[
\PP_\rho\left(\frac{T_{I}(k)}{k}\leq\langle t_I \rangle-\gamma\right)\leq e^{-ku(\langle t_I \rangle-\gamma)}\E_\rho[e^{uT_I(k)}].
\]
By repeating the process we get a similar bound on the Laplace transform, valid for $0\leq|u|<\frac{\varepsilon}{10c_q}$:
\begin{equation*}
    \E_\rho[e^{uT_I(k)}]\leq C(\rho) \exp\left(k\left(\langle t_I \rangle u +\frac{4c_q^2u^2}{\varepsilon}\left(1-\frac{10c_q|u|}{\varepsilon}\right)^{-1}\right)\right).
\end{equation*}
One obtains the upper bound on left deviations, a symmetric bound to right deviations:
\[
\PP_\rho \left (\frac{T_I(k)}{k} \leq \langle t_I \rangle - \gamma \right) \leq C(\rho) \exp \left ( -k \frac{\gamma^2 \varepsilon}{8{c_q^2}}h\left ( \frac{5\gamma}{2 c_q}\right )\right ).
\]
\end{proof}
\begin{coroA}
The variance of the first passage time for total counts is bounded from above by:
\[
\frac{{\rm var}_\sigma(T_{I}(k))}{k}\leq \left (\frac{4}{\varepsilon} -(1-\varepsilon) \right )c_q^{2}.
\]
\end{coroA}
\begin{proof}
Notice that for $u\geq 0$ small enough, one has
\[
\begin{split}
\log(\E_\sigma[e^{uT_{I}(k)}])&=\langle t_I \rangle ku+\frac{1}{2}{\rm var}_{\sigma}(T_{I}(k))u^2+o(u^2)\\
&\leq \frac{k}{2}\log(r(u))=\frac{k}{2}r^{\prime}(0)  u + \frac{k}{4}(r^{\prime\prime}(0)-(r^\prime(0))^2)u^2 + o(u^2)
\end{split}
\]

where $r(u)$ is given by equation \eqref{eq:qru1}. We recall that
\[r^\prime(0)=r^{(1)}=2\langle t_I \rangle = -2\left\langle \I,\LL_0^{-1}(\I)\right\rangle_\sigma\]
and
\[\begin{split}\frac{r^{\prime\prime}(0)}{2}\quad &=r^{(2)}=\langle \I, \Phi^{(2)} (\I)\rangle_\sigma + \left\langle \I,\Phi^{(1)}(\Id-\hat{\Phi})^{-1}\Phi^{(1)}(\I)\right\rangle_\sigma\\
&=2\langle \I, \LL_0^{-2} (\I) \rangle_\sigma+\langle \I, \LL_0^{-1}\hat{\Phi}(\LL^{\dagger}_{0})^{- 1} (\I) \rangle_\sigma\\
&+\langle \I, ((\LL_0^\dagger)^{-1}+\LL_0^{-1}\hat{\Phi})(\Id-\hat{\Phi})^{-1} (\hat{\Phi}(\LL_0^\dagger)^{-1}+\LL_0^{-1}) (\I) \rangle_\sigma\\
&=2\langle \I, \LL_0^{-1} (\I) \rangle^2_\sigma+2\langle \I, \LL_0^{-1}(\Id-\Pi) \LL_{0}^{-1}(\I) \rangle_\sigma \\
&+\langle \I, \LL_0^{-1}\hat{\Phi} (\LL_{0}^\dagger)^{-1}(\I)  \rangle_\sigma+ \varepsilon^{-1}\langle \I, (\LL_0^\dagger)^{-1}(\Id-\Pi) \LL_{0}^{-1}(\I) \rangle_\sigma\\
&+(1-\varepsilon)\varepsilon^{-1}(\langle \I, (\LL_0^\dagger)^{-1}(\Id-\Pi) (\LL_{0}^\dagger)^{-1}(\I) \rangle_\sigma+\langle \I, \LL_0^{-1}(\Id-\Pi) \LL_{0}^{-1}(\I) \rangle_\sigma)\\
&+(1-\varepsilon)^{2}\varepsilon^{-1}\langle \I, \LL_0^{-1}(\Id-\Pi) (\LL_{0}^\dagger)^{-1}(\I) \rangle_\sigma\\
&\leq 2\langle \I, \LL_0^{-1} (\I) \rangle^2_\sigma + \left (\frac{4}{\varepsilon}-(1-\varepsilon) \right )c_q^2.
\end{split}\]
Therefore one has
\[
\frac{r^{\prime\prime}(0)-(r^\prime(0))^2}{2}\leq \left (\frac{4}{\varepsilon}-(1-\varepsilon)\right ) c_q^{2}
\]
and we proved the statement.

\end{proof}

%%%%%%%%%%%%%%%%%%%
\subsection{Proof of Theorem \ref{theo:dynActReset} and Corollary \ref{dynActiTURreset}}
\label{app:proof_theo:dynActReset}
%%%%%%%%%%%%%%%%%%%%%%%
\begin{theoA}[Fluctuations of FPT for Total Counts in Reset Processes]
Assume that Hypothesis \ref{hypo:irrQuantL} holds ($\mathcal{L}$ be irreducible) the jump operators are of the form \eqref{eq.reset} (reset process). Let $\varepsilon$ be the spectral gap of $\P^\dag \P$. For every $\gamma>0$:
\begin{equation*}
\begin{split}
&\PP_\nu \left (\frac{T_I(k)}{k} \geq \langle t_I \rangle + \gamma \right) \leq C(\nu) \exp \left ( -k \frac{\gamma^2 \varepsilon}{4b_r^2}h\left ( \frac{5c_r\gamma}{2 b_r^2}\right )\right )\\
{\rm and}\\
&\PP_\nu \left (\frac{T_I(k)}{k} \leq \langle t_I \rangle - \gamma \right) \leq C(\nu) \exp \left ( -k \frac{\gamma^2 \varepsilon}{4b_r^2}h\left ( \frac{5c_r\gamma}{2 b_r^2}\right )\right ),\quad k\in \mathbb{N},
\end{split}
\end{equation*}
where $h(x):=(\sqrt{1+x}+\frac{x}{2}+1)^{-1}$ and $C(\nu):=(\sum_x \nu(x)^2/\pi(x))^{\frac{1}{2}}$.
\end{theoA}

\begin{proof}
The general formula for the Laplace transform of $T_\EE(k)$ in the case of reset processes reads:
\[
\E_\nu[e^{uT_I(k)}]=\left\langle\nu, (\F_u \P)^k\Iv\right\rangle, \quad \text{ for }u<\overline{\lambda}.
\]
$\P$ is given by equation \eqref{transProbReset} and $\F_u$ is a diagonal matrix whose entries are $(\F_{u})_{ii}=\tr(\ket{y_i}\bra{y_i}(\Id+u\LL_0^{-1})^{-1}(\I))$, $y_i\in \mathbb{C}_d$. $\F_u$ is self adjoint and we can use Lemmas \ref{multOpLem} and \ref{lpLem} as in the classical case to upper bound the norm $\left\|\P\F_u\right\|_\pi$. We seek as before to expand $\hat{\P}\F_u$ ($\hat{\P}$ given by its definition in section \ref{CMarkov}):
\[
\hat{\P}\F_u = \hat{\P} + \sum_{l=1}^\infty u^l \hat{\P}\D^{(l)},
\]
where $\D^{(l)}$ is a diagonal matrix with entries $\tr\left(\ket{y_i}\bra{y_i}\left(-\LL_0^{-1}\right)^l(\I)\right)$. Note that contrary to the classical case this is not simply some diagonal matrix $\D$ raised to a power $l$. The $L_2(\pi)$ norm $\|\D^{(l)}\|_\pi=\D^{(l)}_{\rm max}$, where $\D^{(l)}_{\rm max}$ is the maximum absolute element in $\D^{(l)}$. Therefore, $\D^{(l)}$ can be upper bounded in $L_2(\pi)$ by:
\begin{equation}\label{eq:Dl_estim}
\begin{split}
\|\D^{(l)}\|_\pi & = \sup_{i\in I } ~\tr\left (\ket{y_{\tilde{i}}}\bra{y_{\tilde{i}}}\left(-\frac{\Id}{\LL_0}\right)^l(\I)\right )\\
&\leq \left\|\frac{\Id}{\LL_0}\right\|_{\infty\rightarrow \infty}^l=c_r^l.
\end{split}
\end{equation}
Similarly, if we set $b_r^2=\sum_{i\in I}\pi(i)\left\|\LL_{0*}^{-2}(\ket{y_i}\bra{y_i})\right\|_{1}$, we can upper bound $\|\D^{(l)}\Iv\|_\pi$:
\begin{equation}\label{eq:Dl1_estim}
    \begin{split}
        \|\D^{(l)}\Iv\|_\pi^2 &= \sum_{i\in I}\pi(i)\tr\left (\ket{y_i}\bra{y_i}\left(-\frac{\Id}{\LL_0}\right)^l(\I)\right)^2\\
        &\leq \sum_{i\in I}\pi(i)\tr\left(\ket{y_i}\bra{y_i}\left(-\frac{\Id}{\LL_0}\right)^l(\I)\right)\sup_{i}\left\{\tr\left(\ket{y_i}\bra{y_i}\left(-\frac{\Id}{\LL_0}\right)^l(\I)\right)\right\}\\
        &\leq\sum_{i\in I}\pi(i)\tr\left(\left(-\frac{\Id}{\LL_{0*}}\right)^2(\ket{y_i}\bra{y_i})\left(-\frac{\Id}{\LL_0}\right)^{l-2}(\I)\right)c_r^l\\
        &\leq b_r^2c_r^{2(l-1)}.
    \end{split}
\end{equation}
Therefore $\|\D^{(l)}\Iv\|_\pi\leq b_rc_r^{l-1}$. The final step which differs slightly from the classical case is in calculating the $r^{(l)}$ from equation \eqref{eq:ru_exp}. We get the same results for $r^{(0)}$, $r^{(1)}$ and $r^{(2)}$, but for $l\geq 3$ and $p=1$:
\[
-\tr(\hat{\P}\D^{(l)}(-\Pi))=\left\langle \Iv,\D^{(l)}\Iv\right\rangle_\pi\leq b_r^2c_r^{l-2},
\]
in which we used the same trick as when bounding $\|\D^{(l)}\Iv\|_\pi$. For $p\ge 2$:
\begin{equation*}
\begin{split}
-\tr(\hat{\P}\D^{(\nu_1)}\S^{(\mu_1)}\cdots \hat{\P}\D^{(\nu_p)}\S^{(\mu_p)})&=\left\langle \Iv, \D^{(\nu_1)}\S^{(\mu_1)}\cdots \hat{\P}\D^{(\nu_p)}\Iv\right\rangle_\pi\\
&\leq \|\D^{(\nu_1)} \Iv\|_\pi \|\D^{(\nu_p)}\Iv\|_\pi \|\S\|_\pi^{p-1} \|\D^{(\nu_2)}\|_\pi\cdots \|\D^{(\nu_{p-1})}\|_\pi\\
&\leq b_r^2\frac{c_r^{l-2}}{\varepsilon^{l-1}}.
\end{split}
\end{equation*}
From here the proof for the reset and classical process are identical and so we obtain Theorem \ref{theo:dynActReset}. Note that for left deviations we can indeed repeat the classical proof.

\end{proof}

\begin{coroA}
The variance of the first passage time for total counts is bounded from above by:
\[
\frac{{\rm var}_\pi(T_I(k))}{k}\leq \left (1 + \frac{2}{\varepsilon} \right )b_r^2.
\]
\end{coroA}
\begin{proof}
\bigskip Notice that for $u\geq 0$ small enough, one has
\[
\begin{split}
\log(\E_\pi[e^{uT_{I}(k)}])&=\langle t_I \rangle ku+\frac{1}{2}{\rm var}_{\pi}(T_{I}(k))u^2+o(u^2)\\
&\leq k\log(r(u))=kr^\prime(0) u + \frac{k}{2}(r^{\prime\prime}(0)-(r^\prime(0))^2)u^2 + o(u^2)
\end{split}
\]

where $r(u)$ is given by equation \eqref{ruExp}. We recall using the definition of $\D$ from section \ref{app:proof_theo:dynActReset} that for reset processes
\[r^\prime(0)=r^{(1)}=\langle t_I \rangle = \langle \D\Iv, \Iv \rangle_\pi, \quad r^{\prime\prime}(0)=2r^{(2)}=2\langle \D\Iv, \D\Iv\rangle_\pi + 2\left\langle \D\Iv,\frac{\hat{\P}}{\I-\hat{\P}} \D\Iv\right\rangle_\pi
\]
and remind ourselves that $\|\D\Iv\|_\pi \leq b_r$. Therefore
\[
r^{\prime\prime}(0)-(r^\prime(0))^2=\langle {\bf D} \Iv, \Iv \rangle_\pi^2 + 2\left \langle \D\Iv, \frac{\I}{\I-\hat{\P}} \D\Iv\right \rangle_\pi \leq \left (1+\frac{2}{\varepsilon} \right )b_r^2.
\]
Hence ${\rm var}_\pi(T_{I}(k))\leq \left (1+\frac{2}{\varepsilon} \right )b_r^2k$.
\end{proof}

%%%%%%%%%%%%%%%%%%%%%%%%
\subsection{Proof of Theorem \ref{theo:countObsQ} and Corollary \ref{coro:cgiTURq}}
\label{proof_theo:countObsQ}
%%%%%%%%%%%%%%%%%%%%%%
\begin{theoA}[Rare Fluctuations of General Quantum Counting Observable FPTs]
    Assume that Hypothesis \ref{hypo:irrQuantL} holds ($\mathcal{L}$ be irreducible), and let $\A\subseteq I$ be nonempty. For every $\gamma > \beta - \langle t_\A \rangle$:
    \[
    \PP_\nu \left (\frac{T_{\A}(k)}{k} \geq \langle t_{\A} \rangle + \gamma \right) \leq \exp \left ( -k \left(\frac{\gamma+\langle t_{\A} \rangle-\beta}{\beta}-\log\left(\frac{\gamma+\langle t_{\A} \rangle}{\beta}\right)\right)\right ), k\in \mathbb{N}.
    \]
    \end{theoA}

\begin{proof}
We once again use Chernoff bound to upper bound the probability in terms of the Laplace transform:
\[
\PP_\nu\left(\frac{T_{\A}(k)}{k}\geq\langle t_{\A} \rangle+\gamma\right)\leq e^{-uk(\langle t_{\A} \rangle+\gamma)}\E_\nu[e^{uT_{\A}(k)}], \quad u>0.
\]
Next, we write the Laplace transform in terms of the Hilbert-Schmidt inner product:
\begin{equation*}
\begin{split}
    \E_\rho[e^{uT_\A(k)}] &= \left \langle \rho, \left( \frac{\Id}{\Id + \frac{u}{\LL_\infty}}\Psi\right)^k(\I) \right \rangle_{\rm HS}\\
    & = \left \langle \rho, \left( \sum_{i=0}^\infty u^i \left(-\frac{\Id}{\LL_\infty}\right)\Psi\right)^k(\I) \right \rangle_{\rm HS}, \quad u<\|\LL_{\infty*}^{-1}\|_{1\rightarrow 1}^{-1}\\
    &\leq \|\rho\|_{1\rightarrow 1}\left( \sum_{i=0}^\infty u^i \left\|\frac{\Id}{\LL_{\infty*}}\right\|_{1\rightarrow 1}\|\Psi\|_{\infty\rightarrow \infty}\right)^k\left \|\I\right \|_{\infty\rightarrow \infty}.
\end{split}
\end{equation*}
We can denote $\beta:= \|\LL_{\infty*}^{-1}\|_{1\rightarrow 1}$ to obtain:
\begin{equation}\label{eq:mgfBoundCountQ}
    \E_\nu[e^{uT_{\A}(k)}]\leq \left(\frac{1}{1-\beta u}\right)^k = \exp\left(-k\log(1-\beta u)\right).
\end{equation}

From here on, the proof is identical to that of Theorem \ref{theo:countObs}.

\end{proof}

\begin{coroA}
Given any non-empty set of jumps $\A$, the variance of the corresponding first passage time at stationarity is bounded from above by:
\[
\frac{{\rm var}_\varsigma(T_{\A}(k))}{k}\leq \left ( 1+ \frac{2}{\tilde{\varepsilon}}\right )\beta^2,
\] 
where
$$\tilde{\varepsilon}:=1-\max\{\|\Psi(x)\|_{\infty\to\infty}: \|x\|_{\infty\to\infty}=1, \,\tr(\varsigma x) = 0\}.$$
\end{coroA}

\begin{proof}
From the proof of Lemma \ref{lem:asympVarQ} one can see that
\[
\begin{split}
\frac{{\rm var}_{\varsigma}[T_\A(k)]}{k} &= \tr\left (\varsigma \LL_\infty^{-2}(\I) \right )^2  + 2\tr\left (\varsigma\LL_\infty^{-1}(\Id-\Pi_\varsigma)\LL_\infty^{-1}(\I) \right )\\
&+\frac{2}{k} \tr\left (\varsigma\LL_\infty^{-1}\sum_{i=2}^k\sum_{j=1}^{i-1}\Psi^j (\Id -\Pi_\varsigma)\LL_\infty^{-1}(\I) \right ) \\
& \leq\left (1+2\left ( 1+\frac{1}{k}\sum_{i=2}^k\sum_{j=1}^{i-1}(1-\tilde{\varepsilon})^j\right ) \right )\beta^2\\
&=\left (1+\frac{2}{\tilde{\varepsilon}} \right )\beta^2 -\frac{2((1-\tilde{\varepsilon})-(1-\tilde{\varepsilon})^{k+1})}{k \tilde{\varepsilon}^2}\beta^2\\
&\leq\left (1+\frac{2}{\tilde{\varepsilon}} \right )\beta^2.\\
\end{split}
\]
\end{proof}

%%%%%%%%%%%%%%%%%%%%
%Computation of variance (quantum)
%%%%%%%%%%%%%%%%%%%%

\begin{lemma}[Variance of Quantum FPTs]   \label{lem:asympVarQ}
Let $\varsigma$ be the invariant state of $\Psi$, cf. section \ref{sec:fptCount}, and let $\Pi_\varsigma$ be the map $\Pi_\varsigma: x\mapsto \tr\left( \varsigma x\right )\I$. The variance of the first passage time for counting observables $\forall k \geq 0$ is given by:
    \begin{equation}
    \begin{split}
    \frac{{\rm var}_\varsigma\left( T_{\A}(k)\right)}{k} = \tr \left (\varsigma \LL_\infty^{-1}(\I)\right) ^2 &+ 2 \tr\left ( \varsigma \LL_\infty^{-1}\frac{\Id}{\Id - \Psi}(\Id - \Pi_\varsigma)\LL_\infty^{-1}(\I)\right )\\
    &- \frac{2}{k}\tr\left ( \varsigma \LL_\infty^{-1}\frac{\Psi - \Psi^{k+1}}{(\Id - \Psi)^2}(\Id - \Pi_\varsigma)\LL_\infty^{-1}(\I)\right ).
    \end{split}
    \end{equation}
\end{lemma}

\begin{proof}
    We recall the explicit expression for the moment generating function from Lemma \ref{lem:tech1Q}
    \[
        \E_\varsigma[e^{uT_{\A}(k)}]=\tr\left ( \varsigma \left((u+\mathcal{L}_\infty)^{-1}\mathcal{L}_\infty\Psi \right)^k(\I)\right ).
        \]

    We can write the first moment as

    \begin{equation}\label{eq:mgfDerivQ}
        \E_\varsigma[e^{uT_\A(k)}]' = -\sum_{i=1}^k \tr\left(\varsigma\left ( \frac{\LL_\infty}{u + \LL_\infty} \Psi\right )^{i-1}\left ( \frac{\LL_\infty}{(u + \LL_\infty)^2} \Psi\right )\left ( \frac{\LL_\infty}{u + \LL_\infty} \Psi\right )^{k-i}(\I) \right).
    \end{equation}
    
    At $u = 0$ this gives us the form of the asymptotic mean.
    \begin{equation}\label{eq:asympMeanQ}
        \E_\varsigma[T_\A(k)] = -k\tr \left(  \varsigma \LL_\infty^{-1}(\I) \right).
    \end{equation}

    Differentiating equation \eqref{eq:mgfDeriv} at $u=0$ gives us the second moment:
    \begin{equation*}
    \begin{split}
        \E_\varsigma[T_\A(k)^2] &= 2 \tr \left(\varsigma\LL_\infty^{-2}(\I) \right) k + \sum_{i=1}^k \tr\left (\varsigma \LL_\infty^{-1}\left ( \sum_{j=1}^{i-1}\Psi^{i-j} + \sum_{j=1}^{k-i}\Psi^j\right ) \LL_\infty^{-1}(\I) \right ).\\
        &= 2\tr\left(\varsigma \LL_\infty^{-2}(\I) \right) k + 2 \tr\left (\varsigma \LL_\infty^{-1} \sum_{1 \leq j < i \leq k}\Psi^j \LL_\infty^{-1}(\I) \right)\\
        &= 2\tr\left (\varsigma \LL_\infty^{-2}(\I) \right) k + 2 \sum_{i=2}^k\sum_{j=1}^{i-1}\tr\left (\varsigma \LL_\infty^{-1}(\I) \right )^2\\
        &\qquad\qquad\qquad\qquad+ 2 \tr\left (\varsigma \LL_\infty^{-1}\sum_{i=2}^k\sum_{j=1}^{i-1}\Psi^j (\Id -\Pi_\varsigma)\LL_\infty^{-1}(\I) \right )
    \end{split}
    \end{equation*}
    where to arrive at the third line, after $\Psi^j$ we have inserted $\Pi_\varsigma + \Id - \Pi_\varsigma$, where $\Pi_\varsigma$ is the projection onto $\I$. Using the fact that $\sum_{i=2}^k\sum_{j=1}^{i-1}1 = \frac{k}{2}(k-1)$, and recalling equation \eqref{eq:asympMeanQ} for the expression for the first moment, hence
    
    \begin{equation*}
        \begin{split}
        \frac{{\rm var}_\varsigma\left ( T_{\A}(k)\right)}{k} = 2\tr\left (\varsigma\LL_\infty^{-2}(\I)\right ) &- \tr \left (\varsigma \LL_\infty^{-1}(\I) \right )^2\\
        &+ 2 \tr\left (\varsigma\LL_\infty^{-1}\frac{\Psi}{\Id - \Psi}(\Id -\Pi_\varsigma)\LL_\infty^{-1}(\I) \right )\\
        &-\frac{2}{k} \tr \left (\varsigma\LL_\infty^{-1}\frac{\Psi - \Psi^{k+1}}{(\Id - \Psi)^2}(\Id -\Pi_\varsigma)\LL_\infty^{-1}(\I) \right ).
    \end{split}
    \end{equation*}
    Finally, we can again place $\Pi_\varsigma + \Id - \Pi_\varsigma$ in the first term, in between the two $\LL_\infty^{-1}$. Rearranging this gives the final result, $\forall k \geq 0$
    \begin{equation*}
    \begin{split}
        \frac{{\rm var}_\varsigma\left( T_{\A}(k)\right)}{k} = \tr \left (\varsigma \LL_\infty^{-1}(\I)\right) ^2 &+ 2 \tr\left ( \varsigma \LL_\infty^{-1}\frac{\Id}{\Id - \Psi}(\Id - \Pi_\varsigma)\LL_\infty^{-1}(\I)\right )\\
        &- \frac{2}{k}\tr\left ( \varsigma \LL_\infty^{-1}\frac{\Psi - \Psi^{k+1}}{(\Id - \Psi)^2}(\Id - \Pi_\varsigma)\LL_\infty^{-1}(\I)\right ).
    \end{split}
    \end{equation*}
    
\end{proof}

\subsection{Technical Results}
\label{app:technical}
\subsubsection{Perron-Frobenius Theory}

For readers' convenience we report below the results from Perron-Frobenius theory and, more in general, some well known results in the theory of positivity preserving maps acting on finite dimensional functional and matrix spaces ($C^*$-algebras) that we used in the proofs of this work. We refer to Chapter XIII in \cite{gantmacher1959the-theory} for the commutative case and to \cite{evans1978spectral} for the general case of $C^*$-algebras (which includes the cases of our interest). Let ${\cal A}$ be either $\ell^\infty(E)$ or $M_n(\mathbb{C})$ and let $\Psi$ be a positivity preserving map acting on ${\cal A}$. $\Psi$ is said to be irreducible if there exists no non-trivial projection $p \in {\cal A}$ such that
\begin{equation} \label{eq:irred}
\Psi(p) \leq \alpha p
\end{equation}
for some positive constant $\alpha$. One can show that if $\Psi$ is a transition matrix or a quantum channel, the definition above of irreducibility coincide with having a unique faithful invariant state.

\bigskip \begin{theo}[Perron-Frobenius Theory] \label{th:PF}
Let $r$ be the spectral radius of $\Psi$. The following statements hold true.
\begin{enumerate}
    \item $r \in {\rm Sp}(\Psi)$ and there exists $x \in {\cal A}$, $x \geq 0$ such that $\Psi(x)=rx$.
\end{enumerate}
If $\Psi$ is irreducible, then one has some further results.
\begin{enumerate}
\item $r$ is a geometrically simple eigenvalue;
\item $x$ is strictly positive and is the unique positive eigenvector.
\end{enumerate}
\end{theo}

\noindent 
Another important result is the Russo-Dye Theorem.

\bigskip \begin{theo}[Russo-Dye Theorem] \label{th:RD}
$\Psi$ attains its norm at the identity of ${\cal A}$.
\end{theo}

\subsubsection{Irreducibility of the tilted semigroup}\label{sec:irred}

In this section we will show that when $\LL$ generates an irreducible quantum Markov semigroup (Hypothesis \ref{hypo:irrQuantL}), then the semigroup generated by any perturbation of the form
\begin{equation}\label{eq:perturbed}
\LL_s=\LL_0+\sum_{i\in I}e^{a_is}{\cal J}, \quad a_i,s \in \mathbb{R}\end{equation}
generates an irreducible semigroup as well, in the sense that for every $t>0$, $e^{t\LL_s}$ is irreducible according to the definition in Eq. \eqref{eq:irred}. In fact, what we will prove is even stronger: it is well known (see for instance \cite[Proposition 7.5]{wolf2012quantum}) that the irreducibility of $e^{t\LL}$ for $t>0$ is equivalent to the (a priori stronger) property that $e^{t\LL_*}(\rho)>0$ for every $t>0$ and every initial state $\rho$; we will show that such property, often called primitivity, is owned by the semigroup generated by $\LL_s$ as well.

\bigskip \begin{lemma} \label{lem:irredQ}
If $\LL$ satisfies Hypothesis \ref{hypo:irrQuantL}, then $\LL_s$ generates a primitive semigroup.
\end{lemma}
\begin{proof}
    Let us consider a vector $v \in \mathbb{C}^d$ and a state $\rho$ such that $v \in \ker(e^{t\LL_{s*}}(\rho))$ for some $t>0$. Then, using the Dyson series, one can easily see that $\langle v, e^{t\LL_{s*}}(\rho) v \rangle=0$ implies that $\langle v, e^{t\LL_{0*}}(\rho) v \rangle=0$ and for every $k \geq 1$ and $i_1,\dots, i_k \in I$
    $$
    e^{s\sum_{j=1}^{k}a_{i_j}}\int_{\sum_{j=1}^{k}t_j \leq t}\langle v, e^{\left (t -\sum_{j=1}^{k}t_j \right )\LL_{0*}} {\cal J}_{i_k} \cdots {\cal J}_{i_1}e^{t_1 \LL_{0*}}(\rho) v\rangle dt_1 \cdots dt_k =0,$$
    which is equivalent (since $e^{s\sum_{j=1}^{k}a_{i_j}}>0$) to
    $$\int_{\sum_{j=1}^{k}t_j \leq t}\langle v, e^{\left (t -\sum_{j=1}^{k}t_j \right )\LL_{0*}} {\cal J}_{i_k} \cdots {\cal J}_{i_1}e^{t_1 \LL_{0*}}(\rho) v\rangle dt_1 \cdots dt_k=0.$$
    However, using now the Dyson series for $e^{t\LL_{0*}}(\rho)$, one sees that equations above imply that $\langle v,e^{t\LL_{0*}}(\rho) v \rangle =0$. However, since $e^{t\LL}$ is primitive, this implies that $v=0$ and we are done.
\end{proof}

In the case of a classical Markov chain, the equivalence of irreducibility and primitivity for a Markov chain is a consequence of Levy's Theorem (\cite[Theorem 8]{freedman2012approximating}). With the same proof line, one can show the result also for the generator $\L$ of a classical Markov chain and its perturbations of the form
$$
\L_s=\sum_{x \neq y} e^{sa_{xy}}w_{xy}-\RR, \quad a_{xy},s \in \mathbb{R}.$$
Therefore we can state the following.

\bigskip \begin{lemma} \label{lem:irredC}
If $\L$ satisfies Hypothesis \ref{hypo:irrC}, then $\L_s$ generates a primitive semigroup.
\end{lemma}

\subsubsection{Technical Lemmas} Below we report in our notation two technical lemmas that were proved in \cite{fan2021hoeffdings} and which are used in the proofs of some of the bounds obtained in this paper. 

\bigskip \begin{lemma}[Lemma 21 (i) in \cite{fan2021hoeffdings}]\label{multOpLem}
Let $\mathbf{M}_f:L_\pi^2(E)\rightarrow L_\pi^2(E)$ be the multiplication operator associated to a real valued function $f$, i.e. $\mathbf{M}_fg=fg$ for every $g \in L_\pi^2(E)$. Let $\hat{\bf P}$ be the Le\'{o}n-Perron  operator defined in equation \eqref{eq.LeonPerron}. Then the following statement holds:
\[
\|f\|_\pi^2\leq\|\mathbf{M}_f\hat{\P}\mathbf{M}_f\|_\pi.
\]

\end{lemma}

The following result is a slight generalisation of Lemma 2.1 (iii) in \cite{fan2021hoeffdings}. In the following, we will use $\Psi$ to denote either a transition matrix $\P$ or a quantum channel $\Phi$, $\chi$ to indicate their invariant state, i.e. $\pi$ and $\sigma$, and we denote by ${\cal H}$ the Hilbert space corresponding to their invariant state, i.e. $L^2_\pi(E)$ and $L^2(\sigma)$, respectively. We recall that the notation $\hat{\Psi}$ stands for the Le\'{o}n-Perron version of $\Psi$.

\bigskip \begin{lemma}\label{lpLem}
For any operators $\mathbf{A},\mathbf{B}$ acting on $\mathcal{H}$, the following holds true:
\[
\|\mathbf{A}\Psi \mathbf{B}\|_\chi\leq\left\|\mathbf{B}^\dag\hat{\Psi}\mathbf{B}\right\|_\chi^{\frac{1}{2}}\left\|\mathbf{A}\hat{\Psi}\mathbf{A}^\dag\right\|_\chi^{\frac{1}{2}}.
\]
\begin{proof}
Let us consider $h_1, h_2\in \mathcal{H}$, then
\begin{equation*}
\begin{split}
|\langle \Psi h_1, h_2 \rangle_\chi|&=|\langle(\Id-\Pi)(\Psi-\Pi)(\Id-\Pi)h_1,h_2 \rangle_\chi + \langle \Pi h_1,h_2\rangle_\chi|\\
&=|\langle(\Psi-\Pi)(\Id-\Pi)h_1,(\Id-\Pi)h_2 \rangle_\chi + \langle \Pi h_1,h_2\rangle_\chi|\\
&\leq|\langle(\Psi-\Pi)(\Id-\Pi)h_1,(\Id-\Pi)h_2 \rangle_\chi| + |\langle \Pi h_1,h_2\rangle_\chi|\\
&\leq (1-\varepsilon) \|(\Id-\Pi)h_1\|_\chi\|(\Id-\Pi)h_2\|_\chi+|\langle h_1,\chi \rangle\langle\chi, h_2\rangle|\\
&\leq \sqrt{(1-\varepsilon)\|(\Id-\Pi)h_1\|_\chi^2+|\langle\chi, h_1\rangle|^2}\sqrt{(1-\varepsilon)\|(\Id-\Pi)h_2\|_\chi^2+|\langle\chi, h_2\rangle|^2}\\
&=\langle \hat{\Psi}h_1,h_1 \rangle_\chi^{\frac{1}{2}}\langle \hat{\Psi}h_2,h_2 \rangle_\chi^{\frac{1}{2}}.
\end{split}
\end{equation*}
We can then proceed to complete the lemma:
\begin{equation*}
\begin{split}
\|\mathbf{A}\Psi \mathbf{B}\|_\chi&=\sup_{h_1,h_2:\|h_i\|_\chi=1}|\langle \mathbf{A}\Psi \mathbf{B} h_1,h_2\rangle_\chi|\\
&=\sup_{h_1,h_2:\|h_i\|_\chi=1}|\langle \Psi \mathbf{B} h_1,\mathbf{A}^\dag h_2\rangle_\chi|\\
&\leq\sup_{h_1,h_2:\|h_i\|_\chi=1}\langle \hat{\Psi}\mathbf{B} h_1,\mathbf{B} h_1\rangle_\chi^{\frac{1}{2}}\langle \hat{\Psi}\mathbf{A}^\dag h_2,\mathbf{A}^\dag h_2\rangle_\chi^{\frac{1}{2}}\\
&=\sup_{h_1:\|h_1\|_\chi=1}\langle \mathbf{B}^\dag\hat{\Psi}\mathbf{B} h_1,h_1\rangle_\chi^{\frac{1}{2}}\sup_{h_2:\|h_2\|_\chi=1}\langle \mathbf{A}\hat{\Psi}\mathbf{A}^\dag h_2,h_2\rangle_\chi^{\frac{1}{2}}\\
&=\left\|\mathbf{B}^\dag\hat{\Psi}\mathbf{B}\right\|_\chi^{\frac{1}{2}}\left\|\mathbf{A}\hat{\Psi}\mathbf{A}^\dag\right\|_\chi^{\frac{1}{2}}.
\end{split}
\end{equation*}
\end{proof}
\end{lemma}
\end{appendices}

\bibliography{bibliography}

\end{document}